\documentclass[11pt]{article}

\usepackage{times}

\def\colorful{0}

\usepackage{verbatim}

\usepackage{ifthen}
\usepackage{amsmath,amssymb,amsfonts,amsthm}
\usepackage{url}
\usepackage{hyperref}
\usepackage{color}
\usepackage{hyperref}
\usepackage{mathtools}

\def\eps{\epsilon}

\def\to{\rightarrow}

\newcommand{\prob}[2][]{\text{\bf Pr}\ifthenelse{\not\equal{}{#1}}{_{#1}}{}\!\left[#2\right]}
\newcommand{\expect}[2][]{\text{\bf E}\ifthenelse{\not\equal{}{#1}}{_{#1}}{}\!\left[#2\right]}

\newcommand{\dtv}{d_{\mathrm {TV}}}

\newcommand{\Tr}{{\mathrm {Tr}}}

\newcommand{\wt}[1]{{\widetilde{#1}}}
\newcommand{\wh}[1]{{\widehat{#1}}}

\newtheorem{theorem}{Theorem}[section]
\newtheorem{lemma}[theorem]{Lemma}

\newtheorem{proposition}[theorem]{Proposition}
\newtheorem{corollary}[theorem]{Corollary}
\newtheorem{claim}[theorem]{Claim}

\newtheorem{definition}[theorem]{Definition}

\newtheorem{fact}[theorem]{Fact}
\newtheorem{question}[theorem]{Question}

\newcommand{\ignore}[1]{}

\newcommand{\bg}[1]{\medskip\noindent{\bf #1}}

\definecolor{Red}{rgb}{1,0,0}

\newcommand{\oldbound}[1]{{}}

\renewcommand{\epsilon}{\varepsilon}

\DeclareMathOperator{\R}{\mathbb{R}}

\DeclareMathOperator{\N}{\mathbb{N}}
\DeclareMathOperator{\Z}{\mathbb{Z}}

\DeclareMathOperator*{\Var}{Var}
\DeclareMathOperator*{\var}{Var}
\DeclareMathOperator*{\E}{\mathbb{E}}

\DeclareMathOperator{\poly}{poly}

\renewcommand{\[ }{\begin{eqnarray*}}
\renewcommand{\]}{\end{eqnarray*}}

\renewcommand{\emptyset}{\varnothing}

\definecolor{darkpastelred}{rgb}{0.76, 0.23, 0.13}

\ifnum\colorful=1
\newcommand{\new}[1]{{\color{red} #1}}
\newcommand{\blue}[1]{{\color{blue} #1}}

\else
\newcommand{\new}[1]{{#1}}
\newcommand{\blue}[1]{{#1}}

\fi

\def\E{\mathbb{E}}

\newcommand{\eqdef}{\stackrel{{\mathrm {\footnotesize def}}}{=}}

\DeclarePairedDelimiter\floor{\lfloor}{\rfloor}

\newcommand{\polylog}{\mathrm{polylog}}

\newcommand{\hmg}{\mathrm{Hom}}
\newcommand{\He}{\mathrm{He}}

\usepackage{amsthm,amsfonts,amsmath,amssymb,epsfig,color,float,graphicx,verbatim, enumitem}
\usepackage{multirow}

\newif\ifhyper\IfFileExists{hyperref.sty}{\hypertrue}{\hyperfalse}
\hypertrue
\ifhyper\usepackage{hyperref}\fi

\usepackage{enumitem}

\makeatletter
\renewcommand{\section}{\@startsection{section}{1}{0pt}{-12pt}{5pt}{\large\bf}}
\makeatother

\usepackage{framed}
\usepackage{nicefrac}

\def\nnewcolor{1}
\ifnum\nnewcolor=1

\fi
\ifnum\nnewcolor=0

\fi

\newcommand{\p}{\mathbf{P}}

\newcommand{\littlesum}{\mathop{\textstyle \sum}}

\title{List-Decodable Robust Mean Estimation and\\
Learning Mixtures of Spherical Gaussians}

\author{
Ilias Diakonikolas\thanks{Supported by NSF Award CCF-1652862 (CAREER) and a Sloan Research Fellowship.}\\
University of Southern California\\
{\tt diakonik@usc.edu}\\
\and
Daniel M. Kane\thanks{Supported by NSF Award CCF-1553288 (CAREER) and a Sloan Research Fellowship.}\\
University of California, San Diego\\
{\tt dakane@cs.ucsd.edu}\\
\and
Alistair Stewart\\ University of Southern California\\
{\tt alistais@usc.edu}
}

\begin{document}

\maketitle

\thispagestyle{empty}

\vspace{-0.5cm}

\begin{abstract}
We study the problem of {\em list-decodable (robust) Gaussian mean estimation} and the related problem
of {\em learning mixtures of separated spherical Gaussians}. In the former problem, we are given a set $T$ 
of points in $\R^n$ with the promise that an $\alpha$-fraction of points in $T$, where $0< \alpha < 1/2$, 
are drawn from an unknown mean identity covariance Gaussian $G$, 
and no assumptions are made about the remaining points.
The goal is to output a small list of candidate vectors with the guarantee that at least one of 
the candidates is close to the mean of $G$. In the latter problem, we are given samples from a $k$-mixture
of spherical Gaussians on $\R^n$ and the goal is to estimate the unknown model parameters up to small accuracy.
We develop a set of techniques that yield new efficient algorithms with significantly improved
guarantees for these problems. Specifically, our main contributions are as follows:

\medskip

\noindent {\bf List-Decodable Mean Estimation.} Fix any $d \in \Z_+$ and $0< \alpha <1/2$.
 We design an algorithm with sample complexity $O_d (\poly(n^d/\alpha))$ and runtime $O_d (\poly(n/\alpha)^{d})$
that outputs a list of $O(1/\alpha)$ many candidate vectors such that with high probability
one of the candidates is within $\ell_2$-distance $O_d(\alpha^{-1/(2d)})$ from the mean of $G$.
The only previous algorithm for this problem~\cite{CSV17}
achieved error $\tilde O(\alpha^{-1/2})$ under second moment conditions.
For $d = O(1/\eps)$, where $\eps>0$ is a constant, 
our algorithm runs in polynomial time and achieves error $O(\alpha^{\eps})$.
For $d = \Theta(\log(1/\alpha))$, our algorithm runs in time $(n/\alpha)^{O(\log(1/\alpha))}$
and achieves error $O(\log^{3/2}(1/\alpha))$, almost matching the information-theoretically
optimal bound of $\Theta(\log^{1/2}(1/\alpha))$ that we establish.
We also give a Statistical Query (SQ) lower bound 
suggesting that the complexity of our algorithm is qualitatively close to best possible.

\medskip

\noindent {\bf Learning Mixtures of Spherical Gaussians.} We give a learning algorithm
for mixtures of spherical Gaussians,
with unknown spherical covariances, that succeeds under significantly weaker
separation assumptions compared to prior work. For the prototypical case 
of a uniform $k$-mixture of identity covariance Gaussians we obtain the following:
For any $\eps>0$, if the pairwise separation between the means is at least 
$\Omega(k^{\eps}+\sqrt{\log(1/\delta)})$, our algorithm learns the unknown parameters 
within accuracy $\delta$ with sample complexity and running time $\poly (n, 1/\delta, (k/\eps)^{1/\eps})$. 
Moreover, our algorithm is robust to a small dimension-independent
fraction of corrupted data. The previously best known polynomial time algorithm~\cite{VempalaWang:02} 
required separation at least $k^{1/4} \polylog(k/\delta)$. 
Finally, our algorithm works under separation of $\new{\tilde O(\log^{3/2}(k)+\sqrt{\log(1/\delta)})}$
with  sample complexity and running time $\poly(n, 1/\delta, k^{\log k})$. 
This bound is close to the information-theoretically minimum separation of $\Omega(\sqrt{\log k})$~\cite{RV17}.

\medskip

Our main technical contribution is a new technique, using degree-$d$ multivariate polynomials, 
to remove outliers from high-dimensional datasets where the majority of the points are corrupted.
\end{abstract}

\thispagestyle{empty}
\setcounter{page}{0}

\newpage

\section{Introduction} \label{sec:intro}

\subsection{Background} \label{sec:background}
This paper is concerned with the problem of efficiently learning high-dimensional spherical Gaussians in the presence
of a large fraction of corrupted data, and in the related problem of
parameter estimation for mixtures of high-dimensional spherical Gaussians (henceforth, spherical GMMs).
Before we state our main results, we describe and motivate these two fundamental
learning problems.

The first problem we study is the following:

\medskip

\fbox{\parbox{6.1in}{
\smallskip
\noindent {\bf Problem 1: List-Decodable Gaussian Mean Estimation.}
Given a set $T$ of points in $\R^n$ and a parameter $\alpha \in (0, 1/2]$
with the promise that an $\alpha$-fraction of the points in $T$ are drawn
from $G \sim N(\mu, I)$ --- an unknown mean, identity covariance Gaussian ---
we want to output a ``small'' list of candidate vectors $\{ \wh{\mu}_1, \ldots, \wh{\mu}_s \}$
such that at least one of the $\wh{\mu}_i$'s  is ``close'' to the mean $\mu$ of $G$,
in Euclidean distance.
\smallskip

}}

\medskip

A few remarks are in order: We first note that we make no assumptions
on the remaining $(1-\alpha)$-fraction of the points in $T$. These points
can be arbitrary and may be chosen by an adversary that is computationally unbounded and
is allowed to inspect the set of good points. We will henceforth call such a set of
points {\em $\alpha$-corrupted}. Ideally, we would like to output a single hypothesis
vector $\wh{\mu}$ that is close to $\mu$ (with high probability). Unfortunately,
this goal is information-theoretically impossible when the fraction $\alpha$ of good samples is less than $1/2$.
For example, if the input distribution is a uniform mixture of $1/\alpha$ many Gaussians whose means 
are pairwise far from each other, there are $\Theta(1/\alpha)$ different valid 
answers and the list must by definition contain approximations to each of them.
It turns out that the information-theoretically best possible size of the candidates list
is $s = \Theta(1/\alpha)$. Therefore, the feasible goal is to design an efficient algorithm
that minimizes the Euclidean distance between the unknown $\mu$ and its closest $\wh{\mu}_i$.

\medskip

The second problem we consider is the familiar task of learning the parameters
of a spherical GMM. Let us denote by $N(\mu, \Sigma)$ the Gaussian
with mean $\mu \in \R^n$ and covariance $\Sigma \in \R^{n \times n}$.
A Gaussian is called {\em spherical} if its covariance is a multiple
of the identity, i.e., $\Sigma = \sigma^2 \cdot I$, for $\sigma \in \R_+$.
An $n$-dimensional {\em $k$-mixture of spherical Gaussians} (spherical $k$-GMM)
is a distribution on $\R^n$ with density function $F(x) = \sum_{i=1}^k w_i N(\mu_i, \sigma_i^2 \cdot I),$
where $w_i, \sigma_i  \geq 0$, and $\sum_{i=1}^k w_i = 1$.

\medskip

\fbox{\parbox{6.1in}{
\smallskip
\noindent {\bf Problem 2: Parameter Estimation for Spherical GMMs.}
Given $k \in \Z_+$, a specified accuracy $\delta>0$, and samples
from a spherical $k$-GMM $F(x) = \sum_{i=1}^k w_i N(\mu_i, \sigma_i^2 \cdot I)$ on $\R^n$,
we want to estimate the parameters $\{ (w_i, \mu_i, \sigma_i), i \in [k] \}$
up to accuracy $\delta$. More specifically, we want to return a list 
$\{(u_i, \nu_i, s_i), i \in [k] \}$ so that for some permutation $\pi\in \mathbf{S}_k$,
\new{we have that for all $i \in [k]$: $|w_i-u_{\pi(i)}|\leq \delta$, 
$\|\mu_i-\nu_{\pi(i)}\|_2 / \sigma_i \leq \delta/w_i$,
and $|\sigma_i-s_{\pi(i)}|/\sigma_i \leq (\delta/w_i)/\sqrt{n}$.}
\smallskip
}}

\medskip
\new{If $F' =  \sum_{i=1}^k u_iN(\nu_i,s^2_i \cdot I)$ is the hypothesis distribution,
the above definition implies that
$\dtv(F, F')= O(k \delta)$.}
We will also be interested in the {\em robust} version of Problem~2. This corresponds
to the setting when the input is an $\eta$-corrupted set of samples from a $k$-mixture
of spherical Gaussians, where $\eta \ll \new{\min_i w_i}$.

Before we proceed with a detailed background and motivation,
we point out the connection between these two problems.
Intuitively, Problem~2 can be reduced to Problem~1, as follows:
We can think of the samples drawn from a spherical GMM
as a set of corrupted samples from a single Gaussian --- where
the Gaussian in question can be \emph{any} of the mixture components.
The output of the list decoding algorithm will produce
a list of hypotheses with the guarantee that \emph{every} mean vector
in the mixture is relatively close to some hypothesis. If in addition the distances
between the means and their closest hypotheses are substantially smaller
than the distances between the means of different components,
this will allow us to reliably cluster our sample points based on which
hypothesis they are closest to. We can thus cluster points
based on which component they came from, and then we can learn each component independently.

\subsection{List-Decodable Robust Learning} \label{ssec:list-dec-intro}
The vast majority of efficient high-dimensional learning algorithms with provable guarantees
make strong assumptions about the input data.
In the context of unsupervised learning (which is the focus of this paper),
the standard assumption is that the input points are independent samples drawn from
a known family of generative models (e.g., a mixture of Gaussians). However,
this simplifying assumption is rarely true in practice and it is important to design
estimators that are {\em robust} to deviations from their model assumptions.

The field of robust statistics~\cite{HampelEtalBook86, Huber09} traditionally studies the
setting where we can make {\em no} assumptions about a {\em ``small'' constant}
fraction $\eta$ of the data. The term ``small'' here means that $\eta < 1/2$, hence
the input data forms a reasonably accurate representation of the true model.
From the information-theoretic standpoint, robust estimation in this ``small error regime''
is fairly well understood. For example, in the presence of $\eta$-fraction
of corrupted data, where $\eta < 1/2$,
the Tukey median~\cite{Tukey75} is a robust estimator of location
that approximates the mean of a high-dimensional Gaussian
within $\ell_2$-error $O(\eta)$ --- a bound which is known to be
information-theoretically best possible for {\em any} estimator. The catch is that computing
the Tukey median can take exponential time (in the dimension).
This curse of dimensionality in the running time holds for essentially
all known estimators in robust statistics~\cite{Bernholt}.

This phenomenon had raised the following question:
{\em Can we reconcile computational efficiency and robustness in high dimensions?}
Recent work in the TCS community made the first algorithmic progress on this front:
Two contemporaneous works~\cite{DKKLMS16, LaiRV16} gave the first {\em computationally efficient}
robust algorithms for learning high-dimensional Gaussians (and many other high-dimensional models)
with error close to the information-theoretic optimum. Specifically, for the problem of robustly
learning an unknown mean Gaussian $N(\mu, I)$ from an $\eta$-corrupted set of samples,
$\eta < 1/2$, we now know a polynomial-time algorithm that achieves the information-theoretically
optimal error of $O(\eta)$~\cite{DKKLMS17}.

The aforementioned literature studies the setting where the fraction of corrupted data
is relatively small (smaller than $1/2$), therefore the real data is the {\em majority} of the input points.
A related setting of interest focuses on the regime when the fraction $\alpha$
of real data is small --- strictly smaller than $1/2$. From a practical standpoint,
this ``large error regime'' is well-motivated by a number of pressing
machine learning applications (see, e.g.,~\cite{CSV17, SVC16-nips, SKL17-nips}).
From a theoretical standpoint, understanding this regime is of fundamental interest
and merits investigation in its own right. A specific motivation comes
from a previously observed connection to learning mixture models:
Suppose we are given samples from the mixture
$\alpha \cdot N(\mu, I) + (1-\alpha)E$,
i.e., $\alpha$-fraction of the samples are drawn from an unknown Gaussian,
while the rest of the data comes from several other populations
for which we have limited (or no) information.
Can we approximate this ``good'' Gaussian component, independent of the structure of
the remaining components?

More broadly, we would like to understand
what type of learning guarantees are possible when the fraction $\alpha$
of good data is strictly less than $1/2$. While outputting a {\em single}
accurate hypothesis is information-theoretically impossible,
one may be able to efficiently compute a {\em small} list
of candidate hypotheses with the guarantee that {\em at least one of them} is accurate.
This is the notion of {\em list-decodable learning}, a model introduced by~\cite{BalcanBV08}.
Very recently, \cite{CSV17} first studied the problem of robust high-dimensional
estimation in the list-decodable model. In the context of robust mean estimation,
\cite{CSV17} gave an efficient list-decodable learning algorithm with the following performance
guarantee: Assuming the true distribution of the data has bounded covariance, 
their algorithm outputs a list of $O(1/\alpha)$ candidate vectors one of which
is guaranteed to achieve $\ell_2$-error $\tilde{O}(\alpha^{-1/2})$ from the true mean.

Perhaps surprisingly, several aspects of list-decodable robust mean estimation
are poorly understood. For example, is the $\tilde{O}(\alpha^{-1/2})$ error bound
of the \cite{CSV17} algorithm best possible? If so, can we obtain significantly better error
guarantees assuming additional structure about the real data?
Notably --- and in contrast to the small error regime --- even basic {\em information-theoretic}
aspects of the problem are open. That is, ignoring statistical and computational efficiency considerations,
what is the minimum error achievable  with $O(1/\alpha)$ (or $\poly(1/\alpha)$)
candidate hypotheses for a given family of distributions?

The main focus of this work is on the fundamental setting where the good data comes from a Gaussian distribution.
Specifically, we ask the following question:

\begin{question} \label{qn:list-decoding}
What is the best possible error guarantee (information-theoretically) achievable
for list-decodable mean estimation,
when the true distribution is an unknown $N(\mu, I)$?
More importantly, what is the best error guarantee that we can achieve
with a computationally efficient algorithm?
\end{question}

As our first main result, we essentially resolve Question~\ref{qn:list-decoding}.


\subsection{Learning Mixtures of Separated Spherical Gaussians} \label{sssec:gmm-intro}
A mixture of Gaussians or {\em Gaussian mixture model (GMM)} is
a convex combination of Gaussian distributions, i.e.,
a distribution in $\R^n$ of the form $F = \sum_{i=1}^k w_i N(\mu_i, \Sigma_i)$, where
the weights $w_i$, mean vectors $\mu_i$, and covariance matrices $\Sigma_i$ are unknown.
GMMs are one of the most ubiquitous and extensively studied latent variable models in the literature,
starting with the pioneering work of Karl Pearson~\cite{Pearson:94}.
In particular, the problem of parameter learning of a GMM from samples
has received tremendous attention in statistics and computer science.
(See Section~\ref{ssec:related} for a summary of prior work.)

In this paper, we focus on the natural and important case where each of the components
is {\em spherical}, i.e., each covariance matrix 
is an unknown multiple of the identity. The majority of prior algorithmic work
on this problem studied the setting where there is a minimum {\em separation} between the means
of the components\footnote{Without any separation assumptions,
it is known that the sample complexity of the problem becomes exponential
in the number of components~\cite{MoitraValiant:10, HardtP15}.}.
For the simplicity of this discussion, let us consider the case that the mixing weights are uniform
(i.e., equal to $1/k$, where $k$ is the number of components)
and each component has identity covariance.
(We emphasize that the positive results of this paper hold for the general case of an arbitrary mixture
of high-dimensional spherical Gaussians, and apply even in the presence of a small dimension-independent
fraction of corrupted data.) The problem of learning separated spherical GMMs was first studied
by Dasgupta~\cite{Dasgupta:99}, followed by a long line of works that obtained
efficient algorithms under weaker separation assumptions.

The currently best known algorithmic result in this context is the learning algorithm
by Vempala and Wang~\cite{VempalaWang:02} from 2002. 
Vempala and Wang gave a spectral algorithm 
with the following performance guarantee~\cite{VempalaWang:02}: 
their algorithm uses $\poly(n, k, 1/\delta)$ samples and time, 
and learns a spherical $k$-GMM in $n$ dimensions within parameter distance $\delta$, as long as
the pairwise distance (separation) between the component mean vectors is at least
$k^{1/4} \polylog(nk/\delta)$. 
Obtaining a $\poly(n, k, 1/\delta)$ time algorithm for this problem
that succeeds under weaker separation conditions has been an important open problem since.

Interestingly enough, until very recently, even the information-theoretic aspect of this problem
was not understood. Specifically, what is the minimum separation that allows the problem
to be solvable with $\poly(n, k, 1/\delta)$ samples?
Recent work by Regev and Vijayraghavan~\cite{RV17} characterized this aspect of the problem:
Specifically, \cite{RV17} showed that the problem of learning
spherical $k$-GMMs (with equal weights and identity covariances) 
can be solved with $\poly(n, k, 1/\delta)$ {\em samples}
if and only if the means are pairwise separated by at least $\Theta(\sqrt{\log k})$.
Unfortunately, the approach of~\cite{RV17} is non-constructive in high dimensions.
Specifically, they gave a sample-efficient learning algorithm 
whose running time is exponential in the dimension.
This motivates the following question:

\begin{question} \label{qn:gmm}
Is there a $\poly(n, k)$ time algorithm for learning spherical $k$-GMMs
with separation $o(k^{1/4})$, or better $O(k^{\eps})$, for any fixed $\eps>0$?
More ambitiously, is there an efficient algorithm that succeeds under the
information-theoretically optimal separation?
\end{question}

As our second main result, we make substantial progress towards the resolution of
Question~\ref{qn:gmm}.

\subsection{Our Contributions} \label{ssec:results}
In this paper, we develop a set of techniques that yield
new efficient algorithms with significantly better
guarantees for Problems~1 and~2.
Our algorithms depend in an essential way
on the analysis of high degree multivariate polynomials.
\new{We obtain a detailed structural understanding of the behavior
of high degree polynomials under the standard multivariate Gaussian distribution,
and leverage this understanding to design our learning algorithms. More concretely,
our main technical contribution is a new technique, using degree-$d$ multivariate polynomials, 
to remove outliers from high-dimensional datasets where the majority of the points are corrupted.}

\paragraph{List-Decodable Mean Estimation.}
Our main result is an efficient algorithm for list-decodable
Gaussian mean estimation with a significantly improved error guarantee:

\begin{theorem}[List-Decodable Gaussian Mean Estimation] \label{thm:list-decoding-inf}
Fix $d \in \Z_+$ and $0< \alpha <1$.
There is an algorithm with the following performance guarantee:
Given $d, \alpha$, and a set $T \subset \R^n$ of cardinality
\new{$|T| = O(d^{2d}) \cdot n^{O(d)}/\poly(\alpha)$} with the promise
that $\alpha$-fraction of the points in $T$ are
independent samples from an unknown $G \sim N(\mu, I)$, $\mu \in \R^n$,
the algorithm runs in time \new{$O(nd/\alpha)^{O(d)}$} and with high probability outputs a list
of $O(1/\alpha)$ vectors one of which is within $\ell_2$-distance \new{$\tilde{O}_d(\alpha^{-1/(2d)})$} 
of the mean $\mu$ of $G$.
\end{theorem}

\new{
We note that the $\tilde{O}(\cdot)$ notation hides polylogarithmic factors in its argument.
See Theorem~\ref{thm:main-ld} for a more detailed formal statement.}

\paragraph{\em Discussion and Comparison to Prior Work.}
As already mentioned in Section~\ref{ssec:list-dec-intro}, the only previously known algorithm
for list-decodable mean estimation (for $\alpha<1/2$) is due to~\cite{CSV17} and achieves
error $\tilde O(\alpha^{-1/2})$ under a bounded covariance assumption for the good data.
As we will show later in this section (Theorem~\ref{thm:minimax-main}),
this error bound is information-theoretically (essentially) best possible
under such a second moment condition. Hence, additional assumptions about the good data
are necessary to obtain a stronger bound. It should also be noted that the algorithm~\cite{CSV17}
does not lead to a better error bound, even for the case that the good distribution 
is an identity covariance Gaussian\footnote{Intuitively, this holds because the \cite{CSV17} 
algorithm only uses the first two empirical moments. It can be shown that more moments
are necessary to improve on the $O(\alpha^{-1/2})$ error bound 
(see the construction in the proof of Theorem~\ref{thm:sq-main}).}.

\new{
Our algorithm establishing Theorem~\ref{thm:list-decoding-inf} achieves substantially 
better error guarantees under stronger assumptions about the good data. 
The parameter $d$ quantifies the tradeoff between the error guarantee and the sample/computational 
complexity of our algorithm. Even though it is not stated explicitly in Theorem~\ref{thm:list-decoding-inf}, 
we note that for $d=1$ our algorithm straightforwardly extends to 
all subgaussian distributions (with parameter $\nu = O(1)$), 
and gives error $\tilde{O}(\alpha^{-1/2})$. We also remark that our algorithm is spectral --- in contrast to \cite{CSV17} 
that relies on semidefinite programming --- and it may be practical 
for small constant values of $d$.}

There are two important parameter regimes
we would like to highlight: First, for $d = O(1/\eps)$, where $\eps>0$ is an arbitrarily small constant,
Theorem~\ref{thm:list-decoding-inf} yields a polynomial time algorithm that
achieves error of $O(\alpha^{\eps})$. Second, for $d = \Theta(\log(1/\alpha))$,
Theorem~\ref{thm:list-decoding-inf} yields an algorithm that runs in time \new{$(n/\alpha)^{O(\log(1/\alpha))}$}
and achieves error of \new{$\tilde{O}(\log^{3/2}(1/\alpha))$}.
This error bound comes close to the information-theoretic optimum
of $\Theta (\sqrt{\log(1/\alpha)})$, established in Theorem~\ref{thm:minimax-main}.
\new{While we do not prove it in this version of the paper, we believe that an adaptation of our algorithm 
works under the optimal separation of  $O(\sqrt{\log(1/\alpha)})$.}

A natural question is whether there exists a $\poly(n/\alpha)$ time list-decodable mean estimation
algorithm with error $\polylog(1/\alpha)$, or even $\Theta (\sqrt{\log(1/\alpha)})$.
In Theorem~\ref{thm:sq-main}, we prove a Statistical Query (SQ) lower bound suggesting
that the existence of such an algorithm is unlikely. More specifically, our SQ lower bound
gives evidence that the complexity of our algorithm is qualitatively best possible.

\paragraph{\em High-Level Overview of Technical Contributions.}
Let $G \sim N(\mu, I)$ be the unknown mean Gaussian from which
the $\alpha$-fraction of good samples $S$ are drawn,
and $T$ be the $\alpha$-corrupted set of points given as input.
We design an algorithm that iteratively detects and removes
outliers from $T$, until we are left with a collection of $s = O(1/\alpha)$ many
subsets $T_1, \ldots, T_s$ of $T$ one of which is substantially ``cleaner'' than $T$.
\new{Specifically, the empirical mean of at least one of the $T_i$'s will be 
\new{$\tilde{O}_d(\alpha^{-1/(2d)})$} close to the unknown mean $\mu$ of $G$.}
Our algorithm is ``spectral'' in the sense that it works by analyzing
the eigendecomposition of certain matrices constructed
from degree-$d$ moments of the empirical distribution.
Specifically, to achieve error of \new{$\tilde{O}_d(\alpha^{-1/(2d)})$},
the algorithm of Theorem~\ref{thm:list-decoding-inf} works
with matrices of dimension $O(n^d) \times O(n^d)$.

At a very high-level, our approach bears a similarity to the ``filter'' method
--- a spectral technique to iteratively detect and remove outliers from a dataset ---
introduced in~\cite{DKKLMS16}, for efficient robust estimation in the ``small error regime''
(corresponding to $\alpha \gg 1/2$). Specifically, our algorithm tries
to identify degree-$d$ polynomials $p:\R^n \to \R$ such that
the behavior of $p$ on the corrupted set of samples $T$
is significantly different from the
expected behavior of $p$ on the good set of samples $S$.
One way to achieve this goal~\cite{DKKLMS16, DiakonikolasKS16c}
is by finding polynomials $p$ with unexpectedly large \new{empirical} variance.
\new{The hope is that} if we find such a polynomial, we can then use it to identify a set
of points with a large fraction of corrupted samples 
and remove it to clean up our data set.
This idea was previously used for robust estimation
in the small error regime.

A major complication that occurs in the regime of $\alpha <1/2$
is that since fewer than half of our samples are good,
the values of such a polynomial $p$ might concentrate in several clusters. 
As a consequence, we will not necessarily be able to identify which cluster contains the
good samples. In order to deal with this issue, we need to develop new
techniques for outlier removal that handle the setting
that the good data is a small fraction of our dataset. Roughly speaking, 
we achieve this by performing a suitable clustering of points based on the values of $p$,
and returning multiple (potentially overlapping) subsets 
of our original dataset \new{$T$} with the guarantee that at least one of them 
will be a cleaner version of \new{$T$}. This new paradigm for performing outlier removal
in the large error regime may prove useful in other contexts as well.

\new{
A crucial technical contribution of our approach is the use of degree
more than one polynomials for outlier removal in this setting. The intuitive reason 
for using polynomials of higher degree is this:
A small fraction of points that are far from the true mean in some particular
direction will have a more pronounced effect on higher degree moments.
Therefore, taking advantage of the information contained in higher moments should allow
us to discern smaller errors in the distance from the true mean. The difficulty
is that it is not clear how to algorithmically exploit the structure 
of higher degree moments in this setting. 

The major obstacle is the following:
Since we do not know the mean $\mu$ of $G$ --- this is exactly the quantity
we are trying to approximate! --- we are also not able
to evaluate the variance $\var[p(G)]$ of $p(G)$. If $p$ was a degree-$1$ polynomial,
this would not be a problem, as  the variance $\var[p(G)]$ does not depend on
$\mu$. But for degree at least $2$ polynomials, the dependence of $\var[p(G)]$
on $\mu$ becomes a fundamental difficulty.
Thus, although we can potentially find polynomials with unexpectedly large \new{empirical} variance,
we will have no way of knowing whether this is due to corrupted points $x \in T$
(on which $p(x)$ is abnormally far from its true mean), or due to errors in our estimation
of the mean of $G$ causing us to underestimate the variance $\var[p(G)]$. 

In order to circumvent this difficulty, we require a number of new ideas,
culminating in an algorithm that allows us to either verify that the
variance of $p(G)$ is close to what we are expecting, or to find some
other polynomial that allows us to remove outliers.
}

\paragraph{Learning Mixtures of Separated Spherical GMMs.}
We leverage the connection between list-decodable learning
and learning mixture models to obtain
an efficient algorithm for learning spherical GMMs
under much weaker separation assumptions. Specifically,
by using the algorithm of Theorem~\ref{thm:list-decoding-inf}
combined with additional algorithmic ideas, we obtain our second main result:

\begin{theorem}[Learning Separated Spherical GMMs] \label{thm:gmm-inf}
There is an algorithm with the following performance guarantee: Given $d \in \Z_+$,
$\alpha > 3\delta \geq 0$, and sample access
to a $k$-mixture of spherical Gaussians $F=\sum_{i=1}^k w_i N(\mu_i,\sigma_i^2 I)$
\new{on $\R^n$, where $n  = \Omega(\log(1/\alpha))$,}
with $w_i \geq \alpha$ for all $i$, and so that $\|\mu_i-\mu_j\|_2/(\sigma_i+\sigma_j)$ is at least
$$
S \eqdef C\left(\alpha^{-1/(2d)} \sqrt{d} \new{(d+\log(1/\alpha))\log\left(2+\log(1/\alpha)\right)^2} + \sqrt{\log(k/\delta)}\right) \;,
$$
for all $i\neq j$, for $C>0$ a sufficiently large constant, 
the algorithm draws $\poly(n, \new{(dk/\delta)^d})$ samples from $F$,
runs in time $\poly(n, \new{(dk/\delta)^d})$, and with high probability returns a list $\{ (u_i, \nu_i, s_i), i \in [k] \}$,
such that the following conditions hold (up to a permutation): $|u_i - w_i|=O(\delta)$,
$\|\mu_i-\nu_i\|_2 / \sigma_i = O(\delta/w_i)$,
and $|s_i-\sigma_i|/\sigma_i = O(\delta/w_i)/\sqrt{n}$.
\end{theorem}

The reader is also referred to Proposition~\ref{prop:gmm-general} for a more detailed statement 
that also allows a small, dimension-independent fraction of adversarial noise in the input samples.

\paragraph{\em Discussion and High-Level Overview.}
To provide a cleaner interpretation of Theorem~\ref{thm:gmm-inf}, 
we focus on the prototypical case of a uniform mixture 
of identity covariance Gaussians. For this case,
Theorem~\ref{thm:gmm-inf} reduces to the following statement \new{(see Corollary~\ref{cor:gmm-identity-dim-reduction})}: 
For any $\eps>0$, if the pairwise separation between the means 
is at least $\Omega(k^{\eps} + \sqrt{\log(k/\delta)})$,
our algorithm learns the parameters up to accuracy $\delta$
in time $\poly\left(n, 1/\delta, (k/\eps)^{1/\eps}\right)$. 
\new{Prior to our work, the best known efficient algorithm~\cite{VempalaWang:02} 
required separation $\Omega(k^{1/4} + \sqrt{\log(k/\delta)})$.}
Also note that by setting $d = \Theta (\log k)$, we obtain a learning algorithm 
with sample complexity and running time $\poly(n, 1/\delta, k^{\log k})$
that works with separation of $\new{\tilde O(\log^{3/2}(k)+\sqrt{\log(1/\delta)})}$. 
\new{This separation bound comes close to
the information-theoretically minimum of $\Omega(\sqrt{\log k})$~\cite{RV17}.
(We also note that improving the error bound in Theorem~\ref{thm:list-decoding-inf} to $O(\sqrt{\log(1/\alpha)})$, 
for $d = O(\log (1/\alpha))$, would directly improve our separation bound to 
$O(\sqrt{\log k})$.)}

We now provide an intuitive explanation of our spherical GMM learning algorithm.
First, we note that we can reduce the dimension of the problem 
from $n$ down to some function of $k$.
When the covariance matrices of the components are nearly identical, this can be done
with a twist of standard techniques. For the case of arbitrary covariances,
we need to employ a few additional ingredients.

When each component has the same covariance matrix,
the learning algorithm is quite simple:
We start by running our list-decoding algorithm (Theorem~\ref{thm:list-decoding-inf}) 
with appropriate parameters to get a small list of hypothesis means. 
We then associate each sample with the closest element of our list. 
At this point, we can cluster the points based on which means 
they are associated to and use this clustering to accurately learn the correct components.

The general case, when the covariances of the components are arbitrary, is significantly
more complex. In this case, we can recover a list $H$ of candidate means
only after first guessing the radius of the component that we are looking for.
Without too much difficulty, we can find a large list of guesses
and thereby produce a list of hypotheses of size $\poly(n/\alpha)$.
However, clustering \new{based on this list} now becomes somewhat more difficult,
as we do not know the radius at which to cluster. We address
this issue by performing a secondary test to determine whether
or not the cluster that we have found contains many points at approximately
the correct distance from each other.

\paragraph{Minimax Error Bounds and SQ Lower Bounds.}
As mentioned in Section~\ref{ssec:list-dec-intro},
even the following information-theoretic aspect
of list-decodable mean estimation is open:
Ignoring sample complexity and running time,
how small a distance from the true mean can be achieved
with $\poly(1/\alpha)$ many hypotheses or number of hypotheses 
that is only a function of $\alpha$, 
i.e., independent of the dimension $n$?

Theorem~\ref{thm:list-decoding-inf} implies that we can achieve
error $\polylog(1/\alpha)$ for Gaussians.
We show that the optimal error bound (upper and lower bound) 
for the case $N(\mu, I)$ and more generally for subgaussian
distributions is in fact $\Theta(\sqrt{\log(1/\alpha)})$. Moreover, under bounded $k$-th
moment assumptions, for even $k$, the optimal error is $\Theta_k(\alpha^{-1/k})$.

\begin{theorem}[Minimax Error Bounds] \label{thm:minimax-main}
Let $0 < \alpha < 1/2$. There exists an (inefficient) algorithm that
given a set of $\alpha$-corrupted samples from a distribution $D$,
where (a) $D$ is \new{subgaussian} with \new{bounded variance in each direction}, 
or (b) $D$ has \new{bounded first $k$ moments, for even $k$}, 
outputs a list of $O(1/\alpha)$ vectors one of which is within distance
$g(\alpha)$ from the mean $\mu$ of $D$, and $g(\alpha) = O(\sqrt{\log(1/\alpha)})$
in case (a) and $g(\alpha) = O_k(\alpha^{-1/k})$ in case (b).
Moreover, these error bounds are optimal, up to constant factors.
\new{Specifically, the error bound of (a) cannot be asymptotically 
improved even if $D = N(\mu, I)$, as long as the list size is $\poly(1/\alpha)$. 
The error bound of (b) cannot be asymptotically improved 
as long as the list size is only a function of $\alpha$.}
\end{theorem}

\noindent 
\new{For the detailed statements, the reader is referred to Section~\ref{sec:minimax-sq}.}

\smallskip
We now turn to our computational lower bounds.
Given Theorem~\ref{thm:minimax-main}, the following natural question arises:
For the case of Gaussians, can we achieve the minimax bound in polynomial time?
We provide evidence that this may not be possible, by proving a Statistical Query (SQ)
lower bound for this problem.
Recall that a Statistical Query (SQ) algorithm~\cite{Kearns:98short} relies
on an oracle that given any bounded function on a single domain element provides
an estimate of the expectation of the function on a random sample from the input distribution.
This is a restricted but broad class of algorithms, encompassing many algorithmic techniques 
in machine learning. 
A recent line of work~\cite{Feldman13, FeldmanPV15, FeldmanGV15, Feldman16}
developed a framework of proving unconditional lower bounds
on the complexity of SQ algorithms for search problems over distributions.

By leveraging this framework, using the techniques of our previous work~\cite{DiakonikolasKS16c}, 
we show that any SQ algorithm for list-decodable Gaussian mean estimation that guarantees error $\alpha^{-1/d}$,
for some $d \geq 2$, requires either high accuracy queries
or exponentially many queries:

\begin{theorem}[SQ Lower Bounds] \label{thm:sq-main}
Any SQ list-decodable mean estimation algorithm for $G \sim N(\mu, I)$
that returns a list of sub-exponential size so that some element in the list is within
distance $O(\alpha^{-1/d})$ of the mean $\mu$ of $G$ requires either queries of accuracy
$2^{O((1/\alpha)^{2/d})} \cdot n^{-\Omega(d)}$ or $2^{n^{\Omega(1)}}$ queries.
\end{theorem}

\noindent \new{The reader is referred to Section~\ref{ssec:sq} for the formal statement and proof.}

\subsection{Related Work} \label{ssec:related}
\noindent {\bf Robust Estimation.} The field of robust 
statistics~\cite{tukey1960, huber1964, Huber09, HampelEtalBook86, rousseeuw2005robust}
studies the design of estimators that are stable to model misspecification.
After several decades of investigation, the statistics community has discovered a number of
estimators that are provably robust in the sense that they can tolerate a constant (less than $1/2$) fraction of
corruptions, independent of the dimension. While the information-theoretic aspects
of robust estimation have been understood, the central algorithmic
question --- that of designing robust and computationally efficient estimators in high-dimensions ---
had remained open.

Recent work in computer science~\cite{DKKLMS16, LaiRV16} shed
light to this question by providing the first efficient robust learning algorithms for a variety
of high-dimensional distributions. Specifically,~\cite{DKKLMS16} gave
the first robust learning algorithms that can tolerate a constant
fraction of corruptions, independent of the dimension. Subsequently, there has been a flurry of research activity
on algorithmic robust high-dimensional estimation. This includes robust estimation of graphical
models~\cite{DiakonikolasKS16b}, handling a large fraction of corruptions
in the list-decodable model~\cite{CSV17, SCV17},
developing robust algorithms under sparsity assumptions~\cite{BDLS17}, obtaining optimal error
guarantees~\cite{DKKLMS17}, establishing computational lower bounds
for robust estimation~\cite{DiakonikolasKS16c},
establishing connections with robust supervised learning~\cite{DKS17-nasty},
and designing practical algorithms for data analysis applications~\cite{DKK+17}.

\medskip

\noindent {\bf Learning GMMs.} 
A long line of work initiated by Dasgupta~\cite{Dasgupta:99}, 
see, e.g.,~\cite{AroraKannan:01, VempalaWang:02, AchlioptasMcSherry:05, KSV08, BV:08},
provides computationally efficient algorithms 
for recovering the parameters of a GMM under various separation
assumptions between the mixture components.
More recently, efficient parameter learning algorithms were
obtained~\cite{MoitraValiant:10, BelkinSinha:10, HardtP15}
under minimal information-theoretic separation assumptions.
Without separation conditions, the sample complexity of parameter estimation
is known to scale exponentially with the number of components,
even in one dimension~\cite{MoitraValiant:10, HardtP15}.
To circumvent this information-theoretic bottleneck of parameter learning,
a related line of work has studied parameter learning in a smoothed setting~\cite{HK13, GoyalVX14,
BhaskaraCMV14, AndersonBGRV14, GeHK15}.
The related problems of density estimation and proper learning
for GMMs have also been extensively studied~\cite{FOS:06, SOAJ14, MoitraValiant:10, HardtP15, ADLS17, LiS15a}.
In density estimation (resp. proper learning), the goal is to output some hypothesis (resp. GMM)
that is close to the unknown mixture in total variation distance.

Most relevant to the current work is the classical work of
Vempala and Wang~\cite{VempalaWang:02}
and the very recent work by Regev and Vijayraghavan~\cite{RV17}.
Specifically, \cite{VempalaWang:02} gave an efficient algorithm that learns the parameters of {\em spherical} GMMs
under the weakest separation conditions known to date. On the other hand,
\cite{RV17}  characterize the separation conditions
under which parameter learning for spherical GMMs can be solved with $\poly(n, k, 1/\delta)$ {\em samples}.
Whether such a separation can be achieved with an efficient algorithm was left open in \cite{RV17}.
Our work makes substantial progress in this direction.

\subsection{Detailed Overview of Techniques} \label{sec:techniques}

\subsubsection{List-Decodable Mean Estimation} \label{ssec:ld-techniques}


\paragraph{Outlier Removal and Challenges of the Large Error Regime.}
\new{We start by reviewing the framework of~\cite{DKKLMS16} for robust mean
estimation in the small error regime, followed by an explanation
of the main difficulties that arise in the large error regime of the current paper.}

In the small error regime, the ``filtering'' algorithm of~\cite{DKKLMS16}
for robust Gaussian mean estimation works by iteratively 
detecting and removing outliers (corrupted samples)
until the \new{empirical} variance in every direction is \new{not much larger than expected}.
If every direction has small empirical variance, then 
the true mean and the empirical mean are close to each other~\cite{DKKLMS16}.
Otherwise, the~\cite{DKKLMS16} algorithm projects the input points in a direction 
of maximum variance and throws away those points 
whose projections lie unexpectedly far from the empirical 
median in this direction. 
While this iterative spectral technique for outlier removal is by now well-understood
for the small error regime (and has been applied to various settings), 
there are two major obstacles that arise if one wants to generalize 
it to the large error regime, i.e., where only a small fraction $\alpha$ of samples are good.


The first difficulty is that \new{even the one-dimensional version of the problem
in the large error regime is non-trivial.
Specifically, consider a direction $v$ of large empirical variance.
The~\cite{DKKLMS16} algorithm exploits the fact that the empirical
median is a robust estimator of the mean in the one-dimensional setting.
In contrast, in the large error regime, it is not clear how} 
to approximate the true mean of a one-dimensional projection. 
This holds for the following reason: The input distribution can
simulate a mixture of $1/\alpha$ many Gaussians whose means are far from each other, 
and the algorithm will have no way of knowing which is the real one. 
In order to get around this obstacle, \new{we construct more elaborate
outlier-removal algorithms, which we call {\em multifilters}.
Roughly speaking, a multifilter can return several (potentially overlapping) 
subsets of the original dataset $T$ with the guarantee that {\em at least one} 
of these subsets is substantially ``cleaner'' than $T$.}

The second difficulty is somewhat harder to deal with.
\new{As already mentioned, the filtering algorithm of ~\cite{DKKLMS16} 
iteratively removes outliers} by looking for directions in which the empirical distribution
has a substantially larger variance than it should.
In the low error regime, this approach does a good job
of detecting are removing \new{the corrupted points 
that can move the empirical mean far from the true mean}. 
In the large error regime, the situation is substantially different. 
In particular, it is entirely possible
that the empirical distribution does not have abnormally large variance
in {\em any} direction, while still the empirical mean
is $\Omega(\sqrt{1/\alpha})$-far from the true mean.
\new{That is, considering the variance of one-dimensional projections of our dataset 
in various directions seems inadequate in order to improve the 
$O(\sqrt{1/\alpha})$ error bound.
This obstacle is inherent: the variance of {\em linear polynomials}
(projections) is not a sufficiently accurate method}
of detecting a small fraction of good samples 
being substantially displaced from the mean of the \new{bad samples}. 
To circumvent this obstacle, we will use 
{\em higher degree polynomials}, which are much more
sensitive to a small fraction of points being far away from the others.
In particular, our algorithms will search for degree-$d$ polynomials that have
abnormally large expectation or variance, and use such polynomials 
to construct our multifilters. 

\paragraph{Overview of List-Decodable Mean Estimation Algorithm.}
The basic overview of our algorithm is as follows: We compute the sample
mean $\mu_T$ of the $\alpha$-corrupted set $T$, and then search for \new{(appropriate)} 
degree-$d$ polynomials whose \new{empirical} expectation or variance is too large relative 
to what it should be, assuming that the good distribution $G$ is $N(\mu_T, I)$
---  an identity covariance Gaussian with mean $\mu_T$.
\new{We note that this task can be done efficiently with an eigenvalue computation, 
by taking advantage of the appropriate orthogonal polynomials.}
If there are no degree-$d$ polynomials with too large variance, we can
show that the sample mean $\mu_T$ is within distance 
$\tilde O_d(\alpha^{-1/(2d)})$ from the true mean. On the other hand,
if we do find a degree-$d$ polynomial with abnormally large variance,
we will be able to produce a multifilter and make progress.
This \new{top-level} algorithm is described in detail in Section~\ref{ssec:combining}.

\new{We now sketch how to exploit the existence of a large variance polynomial 
$p$ to construct a multifilter.} Intuitively, \new{the existence of such a polynomial $p$ 
suggests} that there are many points that 
are far away from other points, and therefore separating these points
into (potentially overlapping) clusters should guarantee that almost all good
points are in the same cluster. Unfortunately, for this idea to work, we
need to know that the variance of $p$ on the good set of points $S$
is not too large. For degree-$1$ polynomials $p$ this condition holds automatically.
\new{If $S$ is a sufficiently large set of samples from $G \sim N(\mu, I)$ and $p$ is a normalized 
linear form, then $\var[p(S)] \approx \var[p(G)] = 1$.} But if $p$ has degree at least $2$, 
the variance $\var[p(S)]$ depends on the true mean $\mu$, which unfortunately is unknown.
Fortunately, there is a way to circumvent this obstacle by either producing
a multifilter or verifying that the variance $\Var[p(G)]$ is not too large.

\new{We do this as follows:}
Firstly, we show that the variance \new{$\var[p(G)]$, $G \sim N(\mu, I)$}, can be expressed as 
an average of $p_i^2(\mu)$ for some explicitly computable,
normalized, homogeneous polynomials $p_i$ (see Lemma \ref{lem:poly-relations}). 
We then need to algorithmically verify that the polynomials
$p_i(\mu)$ are not too large. This is difficult to do directly,
so instead we replace each $p_i$ by the corresponding
{\em multilinear} polynomial $q_i$, and note that $p_i(\mu)$ is the
average value of $q_i$ at many independent copies of $G$. If this is large,
then it means that evaluating $q_i$ at a random tuple of samples will
often have larger than expected size. 

This idea will allow us to produce a multifilter for the following reason: 
Since each $q_i$ is multilinear, this
essentially allows us to write it as a composition of linear functions.
More rigorously, we use the following iterative process:
We iteratively plug-in variables one at a time to $q_i$.
If at any step the size of the resulting polynomial
jumps substantially, then the fact that this size is not well-concentrated 
as we try different samples will allow us to produce a
multifilter. The details of this argument are given in Lemma~\ref{lem:sym-ml} of
Section~\ref{ssec:sym-ml}.

\subsubsection{Learning Spherical GMMs}

\paragraph{The Identity Covariance Case.}
Since a Gaussian mixture model can simultaneously
be thought of as a mixture of any one of its components with some error distribution,
applying our list-decoding algorithm to samples from a GMM
will return a list of hypotheses so that \emph{every} mean
in the mixture is close to some hypothesis in the list.
We can then use this list to cluster our samples by component.

In particular, given samples from a Gaussian $G=N(\mu,I)$ and
many possible means $h_1,\ldots,h_m$, we consider the process
of associating a sample $x$ from $G$ with the nearest $h_i$.
We note that $x$ is closer to $h_j$ than $h_i$
if and only if its projection onto the line between them is.
Now if $h_i$ is substantially closer to $\mu$ than $h_j$ is,
then this requires that this projection (which is Gaussian distributed)
be far from its mean, which happens with tiny probability.
Thus, by a union bound, as long as our list contains some
$h_i$ that is close to $\mu$, the closest hypothesis to $x$
with high probability is not much further. If the separation between
the means in our mixture is much larger than the separation
between the means and the closest hypotheses,
this implies that almost all samples are associated
with one of the hypotheses near its component mean,
and this will allow us to cluster samples by component.
This idea of clustering points based on which of a finite set
they are close to is an important idea that shows
up in several related contexts in this paper.

\paragraph{The General Case.}
The above idea works more or less as stated for mixtures
of identity covariance Gaussians, but when dealing with
more general mixtures of spherical Gaussians several complications arise.
Firstly, in order to run out list-decoding algorithm,
we need to know (a good approximation to) the covariance matrix of each component.
The other difficulty is that, in order to cluster points, we will take a set
of all nearby hypotheses that have reasonable numbers
of samples associated with them. The issue is that we no longer know
what ``nearby'' means, as it should depend on the covariance
matrix of the associated Gaussian.

To solve the first of these problems we use a trick that will
be reused several times. We note that two samples from the same Gaussian
$N(\mu,\sigma^2 I)$ have distance approximately $\Theta(\sigma\sqrt{n})$,
and that even one sample from $N(\mu,\sigma^2 I)$ 
is unlikely to be much closer than this
to samples from different components. Therefore, by simply looking
at the distance to the closest other sample gives us a constant factor
approximation to the standard deviation of the corresponding component.
This allows us to write down a polynomial-size list of viable
hypothesis standard deviations. Running our list decoding algorithm
for each standard deviation, gives us a polynomial-size list of hypothesis means.

To solve the second problem, we use the above idea to approximate
the standard deviations associated to our sample points. When clustering them,
we look for collections of sample points with standard deviations approximately
the same $\sigma$, whose closest hypotheses are within some reasonable
multiple of $\sigma$ of each other. Since we are able to approximate
the size of the component that our samples are coming from,
we can guarantee that we aren't accidentally merging
several smaller clusters together by using the wrong radius.

\paragraph{Dimension Reduction.}
One slight wrinkle with the above sketched learning algorithm 
is that since the number of candidate hypotheses is polynomial in $n$,
the separation between the components will be required to be at least $\sqrt{\log(n)}$.
This bound is suboptimal, when $n$ is very large. Another issue is that 
the overall runtime of the learning algorithm would not be a fixed polynomial in $n$,
but would scale as $n^d$. There is a way around both these issues,
by reducing to a lower dimensional problem.

In particular, standard techniques involve looking at the $k$ largest principle values
that allow one to project onto a subspace of dimension $k$ without losing too much.
Unfortunately, these ideas require that all of the Gaussians involved have roughly the
same covariance. Fortunately, if $n$ is large, our ability to approximate
the covariance associated to a sample by looking at its distances
to other samples becomes more accurate. Using a slightly modification of this idea,
we can actually break our samples into subsets so that each subset
is a mixture of Gaussians of approximately the save covariance.
By projecting each of these in turn, we can reduce the original problem
to a $\poly(k)$ number of dimensions and eliminate this extra term.

\subsubsection{Minimax Error Bounds}
We now explain our approach to pin down the information-theoretic optimal error
for the list-decodable mean estimation problem. Concretely, 
for the identity covariance Gaussian case
we show that there is an (inefficient) algorithm
that guarantees that some hypothesis is within $O(\sqrt{\log(1/\alpha)})$ of the true mean.
The basic idea is that the true mean must have the property that there is an $\alpha$-fraction
of samples that are well-concentrated (in the sense of having good tail bounds in every direction)
about the point. The goal of our (inefficient) algorithm will be to find
a small number of balls of radius $O(\sqrt{\log(1/\alpha)})$ that covers the set of all such points.
We show that such a set exists using the covering/packing duality.
In particular, we note that if there are a large number of such sets with means far apart,
we get a contradiction since the sets must be individually large
but their overlaps must be pairwise small (due to concentration bounds).

This approach immediately generalizes to provide a list-decodable mean estimation algorithm
for any distribution with known tail bounds, providing an error $O(t)$,
where only an $\alpha$-fraction of the points are more than $t$-far
from the mean in any direction. 
This generic statement has a number of implications
for various families. In particular, it gives a (tight) error upper bound of $O(\sqrt{\log(1/\alpha)})$ 
for subgaussian distributions with bounded variance
in each direction. Previously, no upper bound better than 
$\tilde{O}(1/\sqrt{\alpha})$ was known for these families.
For distributions whose first $k$ central moments 
are bounded from above (for even $k$), we obtain
a tight error upper bound of $O_k (\alpha^{-1/k})$.

Regarding lower bounds,~\cite{CSV17} showed an 
$\Omega(\sqrt{\log(1/\alpha)})$ error lower bound 
for $N(\mu, \Sigma)$, where $\Sigma$ is {\em unknown} and 
$\Sigma \preceq I$. We strengthen this result 
by showing that the $\Omega(\sqrt{\log(1/\alpha)})$ 
lower bound holds even for $N(\mu, I)$.
We also prove matching lower bounds of 
$\Omega_k (\alpha^{-1/k})$ 
for distributions with bounded moments.
Our proofs proceed by exhibiting distributions $X$,
so that $X$ can be written as $X=\alpha X_i +(1-\alpha)E_i$
for many different $X_i$ satisfying the necessary hypotheses.
Then any list-decoding algorithm must return a list of hypotheses
close to the mean of \emph{every} $X_i$. If there are many such $X_i$'s
with means pairwise separated, then the list-decoding algorithm
must either return many hypotheses or have large error.

\subsubsection{SQ Lower Bounds}
Finally, we prove lower bounds for list-decoding algorithms in the Statistical Query (SQ) model.
Roughly speaking, we show that any SQ algorithm must either spend $n^d$ time
or have accuracy \new{higher} than $\alpha^{-1/2d}$,
suggesting that our list-decoding algorithm is qualitatively tight in its tradeoff 
between runtime and sample complexity.

We prove these bounds using the technology developed in \cite{DiakonikolasKS16c}.
This basically reduces to finding a $1$-dimensional distribution
whose first many moments agree with the corresponding moments of a standard Gaussian.
In our case, this amounts to constructing a one-dimensional distribution
$A=\alpha N(\alpha^{-1/d},1) +(1-\alpha)E$,
so that $A$'s first $d$ moments agree with those of a standard Gaussian.
This can be done essentially because the $\alpha N(\alpha^{-1/d},1)$ part
of the distribution only contributes at most a constant to any of the first $d$ moments.
This allows us to take $E$ approximately Gaussian
but slightly tweaked near $0$ in order to fix these first few moments.

We note however, that if we move the error component much further from $0$,
its contribution to the $d^{th}$ moment becomes super-constant and thus impossible to hide.
This corresponds to the fact that degree-$d$ moments
are sufficient (and necessary) in order to detect errors of size $\alpha^{-1/d}$.

\subsection{Organization}

The structure of the paper is as follows:
In Section~\ref{sec:prelims}, we provide the necessary definitions and technical facts.
In Section~\ref{sec:list-alg-general}, we present our list-decoding algorithm.
Section~\ref{sec:gmms} gives our algorithm for GMMs.
Finally, our minimax error and SQ lower bounds are given in Section~\ref{sec:minimax-sq}.

\section{Definitions and Preliminaries} \label{sec:prelims}
\subsection{Notation and Basic Definitions} \label{sec:notation}

\paragraph{Notation.} For $n \in \Z_+$, we denote by $[n]$ the set $\{1, 2, \ldots, n\}$.
If $v$ is a vector, let $\| v \|_2$ denote its Euclidean norm.
If $M$ is a matrix, let $\| M \|_F$ denote its Frobenius norm.

Our algorithm and its analysis will make essential use of tensor analysis.
For a tensor $A$, we will denote by $\|A\|_2$ the $\ell_2$-norm
of its entries. 

Let $T \subset \R^n$ be a finite multiset. We will use $X \in_u T$ to denote
that $X$ is drawn uniformly from $T$.
For a function $f: \R^n \to \R$, we will denote by $f(T)$ the random variable
$f(X)$, $X \in_u T$.

\medskip

Our basic objects of study are the Gaussian distribution and
finite mixtures of spherical Gaussians:

\begin{definition} \label{def:gaussian}
The $n$-dimensional \emph{Gaussian} $N(\mu, \Sigma)$ with mean $\mu \in \R^n$
and covariance $\Sigma \in \R^{n \times n}$ is the distribution with density function
$f(x) = (2\pi)^{-n/2}\det(\Sigma)^{-1/2}\exp(- (1/2) (x - \mu)^T\Sigma^{-1}(x - \mu))$.
A Gaussian is called {\em spherical} if its covariance is a multiple
of the identity, i.e., $\Sigma = \sigma^2 \cdot I$, for $\sigma \in \R_+$.
\end{definition}

\begin{definition} \label{def:gmm}
An $n$-dimensional {\em $k$-mixture of spherical Gaussians} (spherical $k$-GMM)
is a distribution on $\R^n$ with density function $F(x) = \sum_{j=1}^k w_j N(\mu_j, \sigma_j^2 \cdot I)$,
where $w_j \geq 0$, $\sigma_j \geq 0$, for all $j$, and $\sum_{j=1}^k w_j = 1$.
\end{definition}

\begin{definition} \label{def:dtv}
The {\em total variation distance} between two distributions (with probability density functions) $P, Q: \R^n \to \R_+$
is defined to be
$\dtv\left(P, Q \right) \eqdef (1/2) \cdot \| P -Q  \|_1 = (1/2) \cdot \int_{ x \in \R^n} |P(x)-Q(x)| dx.$
{The {\em $\chi^2$-divergence} of $P, Q$ is
$\chi^2(P, Q) \eqdef  \int_{ x \in \R^n} (P(x)-Q(x))^2/Q(x)  dx =  \int_{ x \in \R^n} P^2(x)/Q(x) dx -1$.}
\end{definition}

\subsection{Formal Problem Definitions} \label{sec:problem-defs}
We record here the formal definitions of the problems that we study.
Our first problem is robust mean estimation in the list-decodable learning model.
We start by defining the list-decodable model:

\begin{definition}[List Decodable Learning,~\cite{BalcanBV08}] \label{def:list-dec}
We say that a learning problem is {\em $(m, \eps)$-list decodably solvable} if there exists an efficient
algorithm that can output a set of $m$ hypotheses with the guarantee that at least one is accurate
to within error $\eps$ with high probability.
\end{definition}

\noindent Our notion of robust estimation relies on the following model of corruptions:

\begin{definition}[Corrupted Set of Samples] \label{def:cor-s}
Given $0<\alpha \leq 1$ and a distribution family $\mathcal{D}$, an {\em $\alpha$-corrupted set of samples}
$T$ of size $m$ is generated as follows: First, a set $S$ of $\alpha \cdot m$ 
many samples are drawn independently from some
unknown $D \in \mathcal{D}$. Then an omniscient adversary, that is allowed to inspect the set $S$,
adds an arbitrary set of $(1-\alpha)\cdot m$ many points to the set $S$ to obtain the set $T$.
\end{definition}

\noindent We are now ready to define the problem of list-decodable robust mean estimation:

\begin{definition}[List-Decodable Robust Mean Estimation]
Fix a family of distributions $\mathcal{D}$ on $\R^n$.
Given a parameter $0< \alpha \leq 1$ and an $\alpha$-corrupted set of samples $T$
from an unknown distribution $D \in \mathcal{D}$ , with unknown mean $\mu \in \R^n$,
we want to output a list of $s = \poly(1/\alpha)$
candidate mean vectors $\wh{\mu}_1, \ldots, \wh{\mu}_s$
such that with high probability it holds 
$\min_{j=1}^s \|\wh{\mu}_j - \mu \|_2  = g(\alpha)$, for some function $g: \R \to \R$.
We say that $g(\alpha)$ is the {\em error guarantee} achieved by the algorithm.
\end{definition}

Our main algorithmic result is for the important special case that $\mathcal{D}$ is the family of unknown mean
known covariance Gaussian distributions. We also establish minimax bounds that apply 
for more general distribution families.

Our second problem is that of learning mixtures of separated spherical Gaussians:

\begin{definition}[Parameter Estimation for Spherical GMMs]
Given a positive integer $k$ and samples from a spherical $k$-GMM
$F(x) = \sum_{i=1}^k w_i N(\mu_i, \sigma_i^2 \cdot I)$, we want 
to estimate the parameters $\{ (w_i, \mu_i, \sigma_i), i \in [k] \}$ up to a required accuracy
$\delta$. More specifically, we would like to return a list 
$\{ (u_i, \nu_i, s_i), i \in [k] \}$ so that with high probability the following holds: 
For some permutation $\pi\in \mathbf{S}_k$ 
\new{we have that for all $i \in [k]$: $|w_i-u_{\pi(i)}|\leq \delta$, 
$\|\mu_i-\nu_{\pi(i)}\|_2 / \sigma_i \leq \delta/w_i$,
and $|\sigma_i-s_{\pi(i)}|/\sigma_i \leq (\delta/w_i)/\sqrt{n}$.}
\end{definition}


\noindent The above approximation of the parameters implies that  
$\dtv(\sum_{i=1}^k w_iN(\mu_i,\sigma^2_i I),\sum_{i=1}^k u_iN(\nu_i,s^2_i I))= O(k \delta)$.
The sample complexity (hence, also the computational complexity) of parameter estimation 
depends on the smallest weight $\min_i w_i$
and the minimum separation between the components.

\subsection{Basics of Hermite Analysis and Concentration} \label{ssec:hermite}
We briefly review the basics of Hermite analysis over $\R^n$
under the standard $n$-dimensional Gaussian distribution $N(0, I)$.
Consider $L^2(\R^n, N(0, I))$, the vector space of all
functions $f : \R^n \to \R$ such that $\E_{x \sim N(0, I)}[f(x)^2] <\infty$.
This is an inner product space under the inner product
$$\langle f, g \rangle = \E_{X \sim N(0, I)} [f(X)g(X)] \;.$$
This inner product space has a complete orthogonal basis given by
the \emph{Hermite polynomials}. For univariate degree-$i$ Hermite polynomials, $i \in \N$,
we will use the {\em probabilist's} Hermite polynomials, denoted by $\He_i(x)$, $x \in \R$,
which are scaled to be {\em monic}, i.e., the lead term of $\He_i(x)$ is $x^i$.
For $a \in \N^n$, the $n$-variate Hermite polynomial $\He_{a}(x)$, $x = (x_1, \ldots, x_n) \in \R^n$,
is of the form $\prod_{i=1}^n \He_{a_i}(x_i)$, and has degree $\|a\|_1 = \sum a_i$.
These polynomials form a basis for the vector space of all polynomials
which is orthogonal under this inner product.
For a polynomial $p: \R^n \to \R$, its $L^2$-norm is
$\|p\|_2  \eqdef \sqrt{\langle p, p \rangle} = \E_{X \sim N(0, I)} [p(X)^2]^{1/2}$. 

\medskip

We will need the following standard concentration bound
for degree-$d$ polynomials over independent 
Gaussians (see, e.g., \cite{Janson:97}):

\new{
\begin{fact}[``degree-$d$ Chernoff bound''] \label{thm:deg-d-chernoff}
Let $G \sim N(\mu, I)$, $\mu \in \R^n$.
Let $p: \R^n \to \R$ be a real degree-$d$ polynomial. For any $t > 0$, we have that
$\Pr  \left[ \left| p(G) - \E[p(G)] \right| \geq t \cdot  \sqrt{\Var[p(G)]} \right] \leq \exp(-\Omega(t^{2/d}))$.
\end{fact} 
}

\section{List-Decodable Robust Mean Estimation Algorithm} \label{sec:list-alg-general}

In this section, we prove our main algorithmic result on list-decodable mean estimation:

\new{
\begin{theorem}[List-Decodable Mean Estimation] \label{thm:main-ld}
There exists an algorithm {\tt List-Decode-Gaussian} that, 
given $0 < \alpha < 1/2$, $d \in \Z_+$, a failure probability $\tau >0$, 
and a set $T$ of $O(d!^2 \cdot n^{4d} \cdot \log(1/\tau)/\alpha^7)$ points in $\R^n$, 
of which at least a $\new{2} \alpha$-fraction are independent samples from a Gaussian $G \sim N(\mu, I)$, 
runs in time $(n d \log(1/\tau)/\alpha)^{O(d)}$ and
returns a list of $O(1/\alpha)$ points such that, with probability at least $1-\tau$, 
the list contains an $x \in\R^n$ with 
$$\|x-\mu\|_2 = O\left(\alpha^{-1/(2d)} \sqrt{d} (d+\log(1/\alpha))\log(2+\log(1/\alpha))^2 \right) \;.$$
\end{theorem}
}


\new{
\paragraph{Detailed Structure of Algorithm.}
The key idea procedure behind our algorithm is a subroutine that given a set of
samples either cleans it up producing one or two subsets at least one
of which has substantially fewer errors than the original, or certifies
that the mean of $G$ must be close to the empirical mean (Proposition~\ref{mainSubroutineProp}). 
Using this subroutine, our final algorithm can be obtained by repeatedly applying
the subroutine recursively to the returned sets until they produce
vectors. The details of this analysis are in Section~\ref{ssec:mf}.

Before we can get into the detailed overview of this proof, it is
necessary to lay out some technical groundwork. First, we will
want to have a deterministic condition under which our algorithm
will succeed. To that end, we introduce two important definitions. We
say that a set $S$ is representative of $G$ if it behaves like a set of
independent samples of $G$, in particular in the sense that it is a PRG
against low-degree polynomial threshold functions for $G$. We also say
that a larger set $T$ is good if (roughly speaking) an $\alpha$-fraction of
the elements of $T$ are a representative set for $G$. For technical
reasons, will will also want the points of $T$ to be not too far apart
from each other. In Section~\ref{sec:good-sets}, we discuss the definitions of
representative and good sets and provide some basic results. 

In Section~\ref{ssec:naive-clustering}, 
we show that given a large set of points that contain an $\alpha$-fraction of good points 
from, one can algorithmically find $O(1/\alpha)$ many subsets so that with
high probability at least one of them is good (and thus can be fed
into the rest of our algorithm). This would be immediate if t were it not for
the requirement that the points in a good set be not too far apart. As
it stands, this will require that we perform some very basic
clustering algorithms.

The actual design of our multifilter involves working with several types of
``pure'' degree-$d$ polynomials and their appropriate tensors. In
particular, we need to pay attention to harmonic polynomials (which
behave well with respect to $L^2$-norms), homogeneous polynomials, and
multilinear polynomials. In Section~\ref{ssec:polys}, we introduce these and give
several algebraic results relating them that will be required later.

The multifilter at its base level requires a routine that given a polynomial
$p$, where $p(T)$ behaves very differently from $p(G)$, allows us to use the
values of $p(x)$ to separate the points coming from $G$ from the
errors. The basic idea of the technique is to cluster the points $p(x)$,
for $x \in T$, and throw away points that are too far from any cluster
large enough to correspond to the bulk of the values of $p(G)$ (which
must be well-concentrated), or to divide $T$ into two subsets with
enough overlap to guarantee that any such cluster could be entirely
contained on one side or the other. The details of this basic
multifilter algorithm are covered in Section~\ref{ssec:basic-mf}.

Given this basic multifilter, the high-level picture for our main subroutine is as
follows: Using spectral methods we can find if there are any degree-$d$
polynomials $p$ where $\E[p(T)^2]$ is substantially larger than it should
be if $T$ consisted of samples from $N(\mu_T,I)$. If there are no such
polynomials, it is not hard to see that $\|\mu-\mu_T\|_2$ is small giving us
our desired approximation. Otherwise, we would like to apply our basic
multifilter algorithm to get a refined version of $T$. The details of
this routine can be found in Section~\ref{ssec:combining}.

Unfortunately, the application of the multifilter in the application
above has a slight catch to it. Our basic multifilter will only apply
to $p$ if we can verify that $\Var[p(G)]$ is not too large. This would be
easy to verify if we knew the mean of $G$, but unfortunately, we do not
and errors in our approximation may lead to $\Var[p(G)]$ being much
larger than anticipated, and in fact, potentially too large to apply
our filter productively. In order to correct this, we will need new
techniques to either prove that $\Var[p(G)]$ is small or to find a filter
in the process. Using analytic techniques, in Section~\ref{ssec:polys} we show that
the $\Var[p(G)]$ is a weighted average of squares of $q_i(\mu-\mu_T)$ for
some normalized, homogeneous polynomials $q_i$. Thus, it suffices to
verify that each $q_i(\mu-\mu_T)$ is small.

To deal with this issue, it is actually much easier to work with multilinear
polynomials, and so instead we deal with multilinear polynomials $r_i$
so that $r_i(x, \ldots, ,x) = q_i(x)$. We thus need to verify that
$r_i(\mu-\mu_T, \mu-\mu_T, \ldots, \mu-\mu_T)$ is small. The discussion of the
reduction to this problem is in Section~\ref{ssec:harm}, while the techniques
for verifying that the $r_i$ have small mean is in Section~\ref{ssec:sym-ml}.

In order to handle multilinear polynomials, we treat them as a
sequence of linear polynomials. We note that if $r(\mu,\mu,\ldots,\mu)$ is
abnormally large, then so is $\E[r(G,G,\ldots, G)]$. This means that if we
evaluate $r$ at $d$ random elements of $T$, we are relatively likely to get
an abnormally large value, our goal is to find some linear polynomial
$L$ for which the distribution of $L(T)$ has enough discrepancies that we
can filter $T$ based on $L$. To do this, consider starting with
$r(x_1,x_2,\ldots,x_d)$ where $x_i$ are separate $n$-coordinate variables, and
replacing the $x_i$ one at a time with random elements of $T$. Since there
is a decent probability that $r(t_1,\ldots,t_d)$ is large, it is reasonably
likely that at some phase of this process, setting one of the
variables causes the $L^2$-norm of $r$ to jump by some substantial amount.
In particular, there must be some settings of $t_1, \ldots,t_{a-1}$ so that
for a random element $t$ of $T$, we have that
$r(t_1,\ldots,t_{a-1}, t, x_{a+1}, \ldots ,x_d)$ will have substantially larger $L^2$-norm 
than $r(t_1,\ldots,t_{a-1},x_a,x_{a+1},\ldots,x_d)$ with non-negligible
probability. We note that this would only rarely happen if $t$ were
distributed as $N(0,I)$, and this will allow us to filter. This argument
is covered in Section~\ref{ssec:deg-2}.

To make this algorithm work, we note that
$|r(t_1,\ldots,t_{a-1},X,x_{a+1},\ldots,x_d)|_{2,x_{a+1},\ldots, x_d}^2$ is a
degree-$2$ polynomial with bounded trace-norm. Therefore, we need an
algorithm so that if $A$ is such a polynomial where $\E[A(T-\mu_T)]$ is
large, we can produce a multifilter. This is done by writing $A$ as an average
of squares of linear polynomials. We thus note that there must be some
linear polynomial $L$, where $\E[L(T-\mu_T)^2]$ is abnormally large. In
particular, this implies that $L(T)$ and $L(G)$ have substantially
different distributions, which should allow us to apply our basic
multifilter. Also since $L$ is degree-$1$, we have a priori bounds on
$\Var[L(G)]$, which avoids the problem that has been plaguing us for much
of this argument. These details are discussed in Section~\ref{ssec:polys}.

\paragraph{Overview of this Section.}
In summary, the structure of this section is as follows:
In Section~\ref{sec:good-sets}, we define the important notions of a representative
set and a good set and prove some basic properties. 
In Section~\ref{ssec:naive-clustering}, 
we do some basic clustering and show that a good set can be extracted
from a set of corrupted samples. 
In Section~\ref{ssec:mf}, we present
the main subroutine and show how it can be used to produce our final
algorithm.

The remainder of Section~\ref{sec:list-alg-general} will be spent building this subroutine. 
In Section~\ref{ssec:basic-mf}, we produce our most basic tool for
creating multifilters given a single polynomial where $p(T)$ and $p(G)$
behave substantially differently. 
In Section~\ref{ssec:polys}, we present some basic background 
on harmonic, homogeneous and multilinear polynomials
and their associated tensors. 
In Section~\ref{ssec:deg-2}, we use these to
produce routines to find multifilters given degree-$1$ polynomials or
degree-$2$ polynomials with bounded trace norm for which $p(\mu)$ is too
large. In Section~\ref{ssec:sym-ml}, we leverage these results to extend to
produce a similar multifilter for arbitrary multilinear polynomials. In
Section~\ref{ssec:harm}, we use this to get a multifilter for degree-$d$ harmonic
polynomials whose $L^2$ norms are substantially larger than expected, and
in Section~\ref{ssec:combining}, we combine this with spectral techniques to get the
full version of our filtering procedure, thus finishing our algorithm.
}

\subsection{Representative Sets and Good Sets} \label{sec:good-sets}

Let $0< \alpha < 1/2$ be the fraction of good samples.
Recall that our model of corruptions works as follows:
We draw a sufficiently large set $S$ of independent samples from $G \sim N(\mu, I)$,
where $\mu \in \R^n$ is unknown, and then an adversary arbitrarily
adds $|S|(1/\alpha-1)$ points to the set $S$ to obtain the corrupted set $T$.
The corrupted set $T$ is given as input to our learning algorithm that is required
to produce a list of candidates for the unknown mean vector $\mu$.

\paragraph{Representative Sets.}
We define a deterministic condition on the set of ``clean'' samples
$S$ that guarantees that running our algorithm on \emph{any} corrupted set $T$, as defined above,
will succeed. A set of points $S \in \R^n$ satisfying this deterministic condition will be called {\em representative}.
For our purposes, it will suffice that the representative set $S$ approximately gives the correct distributions
to all low-degree polynomial threshold functions.
We will show that our deterministic conditions hold with high probability for a sufficiently large
set of independent samples from $G \sim N(\mu, I)$.
This discussion is formalized in the following definition:

\begin{definition}[Representative Set] \label{def:rep}
Let $G \sim N(\mu, I)$, $\mu \in \R^n$, $0< \alpha \leq 1$, and $d \in \Z_+$.
We say that a set $S \subset \R^n$ is \emph{representative} (with respect to $G$)
if for any degree-at-most-$2d$ real polynomial $p: \R^n \to \R$, it holds
$$
\left|\Pr[p(G)\geq 0] - \Pr_{x\in_u S}[p(x) \geq 0]\right| \leq \alpha^3/(100d! n^{d}) \;.
$$
\end{definition}

We note that even though the definition of ``representativeness'' of a set $S$
depends on the parameters $\alpha$ and $d$, these quantities
will be fixed for the representative set $S$ throughout the execution of our algorithm,
and thus the dependence will be implicit.
Note that as $\alpha$ increases \new{or $d$ decreases}, the representativeness condition becomes weaker.
Thus, $S$ being representative with parameters $\alpha$ and $d$
implies that it is also representative with parameter $\beta$ and $d'$,
for any $\beta \geq \alpha$ and $d' \leq d$.

We start by showing that a sufficiently large set of samples drawn from $G \sim N(\mu, I)$
is representative with high probability. This fact follows from
standard arguments using the VC-inequality:

\begin{lemma} \label{lem:rep-sample}
For $G \sim N(\mu, I)$, $\mu \in \R^n$, if $S$ is a set of $O(d!^2 \cdot n^{4d} \cdot \log(1/\tau)/\alpha^6)$
independent samples from $G$, then $S$ is representative (with respect to $G$) with probability at least $1-\tau$.
\end{lemma}
\begin{proof}
The collection of sets of the form $p(x) \geq 0$, for multivariate polynomials on $\R^n$ of degree at most $2d$,
has VC-dimension $O(n^{2d})$. Thus, by the VC inequality~\cite{DL:01}, we have that
$$
\left|\Pr[p(G)\geq 0] - \Pr_{x\in_u S}[p(x) \geq 0]\right| \leq \alpha^3/(100d! n^{d}) \;,
$$
for all such polynomials $p$ with probability at least $1-\tau$,
if we take
$$O((d! \cdot n^{d}/\alpha^3)^2 \cdot n^{2d} \log(1/\tau)) = O(d!^2 \cdot n^{4d}\log(1\tau)/\alpha^6)$$ samples.

\end{proof}

\paragraph{Good Sets.}
Our list-decodable learning algorithm and its analysis require an appropriate notion of goodness
for the corrupted set $T$. At a high-level, throughout its execution,
our algorithm will produce several different subsets of
the original corrupted set $T$ it starts with, in its attempt to remove outliers.
Intuitively, we want a good set $T' \subseteq T$ to have the property that
a $\beta$-fraction of the points in $T'$, with $\beta\geq \alpha$, come from a representative set $S$.
However, we also need to account for the possibility that --- in the process of removing outliers ---
our algorithm may also remove a small number of the points in the original representative set $S$.
Moreover, for technical reasons, we will require that all points in a good set
are contained in a ball of radius $O(\sqrt{n})$. This is formalized in the following definition:

\begin{definition}[Good Set] \label{def:good}
Let $G \sim N(\mu, I)$, $\mu \in \R^n$, $0< \alpha  <1$, and $d \in \Z_+$.
A multiset $T \subset \R^n$ is \emph{$\alpha$-good (with respect to $G$)}
if it satisfies the following conditions:
\begin{enumerate}
\item All points in $T$ are within distance $O(\sqrt{n})$ of each other.
\item There exists a set $S \subset \R^n$ which is representative (with respect to $G$)
so that $$|S\cap T| \geq \max \left\{ (1-\alpha/\new{6}-e^{-\new{3}n})|S|, \alpha |T| \right\} \;.$$
\end{enumerate}
\end{definition}

\subsection{Naive Clustering} \label{ssec:naive-clustering}

We note that the definition of ``goodness'' of a set $T$
depends on the parameters $\alpha$ and $d$. The parameter $\alpha$
will change multiple times during the execution of the algorithm (it will increase),
while the parameter $d$ does not increase.
This justifies our choice of making $\alpha$
explicit in the definition of ``$\alpha$-good'', while keeping the dependence on $d$ implicit.

The additional constraint that all points in a good set
are contained in a not-too-large ball means that
our original set of corrupted samples $T$
may not form a good set. To rectify this issue, we start by performing
a very basic clustering step as follows: Since an $\alpha$-fraction of the points in $T$
are concentrated within distance $O(\sqrt{n})$
from the true mean $\mu$, by taking a maximal set of non-overlapping,
not-too-large balls each with a large fraction of points, we are
guaranteed that at least one of them will contain a good set.
This is formalized in the following lemma:

\begin{lemma}\label{goodTLem}
Let $G \sim N(\mu, I)$, $\mu \in \R^n$, \new{and $0<\alpha \leq 1/2$}.
Let $T$ be a set of $O(d!^2 \cdot n^{4d} \cdot \log(1/\tau)/\alpha^7)$ points in $\R^n$,
of which at least a $\new{2} \alpha$-fraction are independent samples from $G$.
There is an algorithm that, given $T$ and $\alpha$,
returns a list of at most $1/\alpha$ many subsets of $T$
so that with probability at least $1-\tau$
at least one of them is $\alpha$-good with respect to $G$.
\end{lemma}
\begin{proof}
Let $S$  be the subset of $T$ containing the good samples, i.e.,
the points that are independent samples from $G$.
We have that $S \subset T$ and $|S| = \new{2} \alpha \cdot |T|$.
By Lemma~\ref{lem:rep-sample}, the set $S$ is representative 
with probability at least $1-\tau$.
We henceforth condition on this event.

If all points in $T$ are contained in a ball of radius $O(\sqrt{n})$, there is nothing to prove.
If this is not the case, we will show how to efficiently
find a collection of at most $1/\alpha$ many balls of radius $O(\sqrt{n})$,
so that at least one of the balls contains
at least a $(1-\alpha/\new{6}-e^{-\new{3}n})$-fraction of the points in $S$.

First note that, by the degree-$2$ Chernoff bound,
we have that $\Pr[\|G-\mu\|_2^2 \geq C^2 n] \leq \exp(-\new{3}n)$,
for a sufficiently large universal constant $C$.
Since $\|G-\mu\|_2^2$ is a degree-$2$ polynomial and $S$ is representative,
Definition~\ref{def:rep} implies that 
at least an $(1-\exp(-\new{3}n) - \alpha/\new{6})$
fraction of points in $S$ are within distance $C \sqrt{n}$ of $\mu$.

Our clustering scheme works as follows:
We consider a maximal set of disjoint balls of radius $R = 2C \cdot \sqrt{n}$
centered at points of $T$, such that each ball contains at least an $\alpha$-fraction
of the points in $T$. Note that this set is non-empty:
Since $S$ is representative, the ball of radius $R$
centered at any point $x \in S$ with $\|x-\mu \|_2 \leq C\sqrt{n}$
contains at least a $(1-e^{-3n}-\alpha/6)$-fraction of the points in $S$, and therefore
at least an $\alpha$-fraction of the points in $T$.

Let $B_1, \ldots, B_s$ be the maximal set of disjoint balls described above.
Since each $B_i$ contains an $\alpha$-fraction of the points in $T$,
there are at most $s \leq 1/\alpha$ many such balls. 
Let $B'_i$, $i \in [s]$, be the ball
with the same center as $B_i$ and radius $3R = 6C \cdot \sqrt{n}$.
Consider the subsets $T_i = T \cap B'_i$, $i \in [s]$.
We claim that at least one of the $T_i$'s is $\alpha$-good.

The pseudo-code for this  clustering algorithm is given below.

\medskip

\fbox{\parbox{6.1in}{
{\bf Algorithm} {\tt NaiveClustering}\\
Input: a multiset $T \subset \R^n $ and $\alpha > 0$. \\
\vspace{-0.5cm}
\begin{enumerate}
\item Let $\cal{C}$ be the empty set.
\item For each $x \in T$, proceed as follows:
\begin{itemize}
\item[] If there are $\alpha \cdot |T|$ points in $T$ within distance $R = 2C \cdot \sqrt{n}$
of $x$, and no point $y \in \cal{C}$ has $\|x-y\|_2 \leq 3R$, then add $x$ to $\cal{C}$.
\end{itemize}
\item For each $x_i$ in $\cal{C}$, let $T_i = T \cap B_{2}(x_i, 3R)$.
\item Return the list of $T_i$'s.
\end{enumerate}
}}

\bigskip

We now prove correctness.
By definition, all the points of $T_i$ are contained in a ball of radius $3R = 6C\sqrt{n}$.
Let $x$ be any point of $S$ within distance $C\sqrt{n}$ of the mean $\mu$ of $G$.
(Recall that, since $S$ is representative, most of the points of $S$ satisfy this property.)
Since the initial set of balls $B_1, \ldots, B_s$ was constructed to be maximal,
at least one $B_i$ must intersect the ball $B_x$ of radius $R = 2C\sqrt{n}$ centered at $x$.
This implies that the ball $B_x$ is contained in $B'_i$.
Therefore, $T_i$ contains all of the points of $S$ within $C\sqrt{n}$ of the mean $\mu$ of $G$.
This completes the proof of Lemma~\ref{goodTLem}.
\end{proof}

\bigskip

\subsection{Main Multifilter Algorithm and Proof of Theorem~\ref{thm:main-ld}} \label{ssec:mf}
In this section, we describe the main subroutine that forms the core of our final algorithm.
Let $0< \alpha < 1/2$ be the fraction of good samples,
and let $S$ be the representative set contained in
the initial corrupted set $T$.
Ideally, starting from the corrupted set $T$, we would like to either
certify that the true mean $\mu$ of $G$ is close to the empirical mean $\mu_T$ of $T$, or
we have an efficient method to ``clean up'' the set $T$,
i.e., obtain a set $T' \subset T$ containing a higher fraction of clean samples.

Unfortunately, such a goal seems unattainable in the setting where
the fraction $\alpha$ of clean samples is smaller than $1/2$ for the following reason:
Suppose that $\alpha = 1/k$ for some integer $k \geq 2$.
It is quite possible that the original set of corrupted points
$T$ is a set of samples drawn from a mixture $\sum_{i=1}^k (1/k) \cdot N(\mu_i, I)$
of $k$ spherical Gaussians with separated means $\mu_i$. In this case, the best that we can hope to
achieve is break $T$ into several (potentially overlapping) subsets $T_i$,
with the guarantee that at least one of the subsets $T_i$ has a higher fraction of points from $S$.
As our final algorithm will then have to recurse on the $T_i$'s,
\new{we will need to ensure that they are not too large} in order to obtain a bound
on the total runtime of the algorithm. 

Our main sub-routine satisfies the guarantees of the following proposition:

\begin{proposition}[Main Multifilter Algorithm] \label{mainSubroutineProp}
There exists an algorithm {\tt MainSubroutine} that, given a multiset $T\subset \R^n$,
a parameter $0 < \alpha \leq 1$, a degree bound $d \in \Z_+$, and an error probability $\tau$,
returns one of the following: (1) A vector $x \in \R^n$, (2) ``NO'', or (3)
A list $\{(T_i, \alpha_i)\}_i$ of size at most $2$ where $T_i \subset T$ and the $\alpha_i$'s
satisfy $\sum_i 1/\alpha_i^2 \leq 1/\alpha^2.$

Moreover, with probability at least $1-\tau$, the output of the algorithm satisfies
the following guarantees:
\begin{enumerate}
\item \label{cond:one} In case (1), if $T$ is $\alpha$-good with respect to $G \sim N(\mu, I)$,
then we have that
\begin{equation}\label{eqn:mean-final-error}
\|x-\mu\|_2 = O\left(\alpha^{-1/(2d)} \sqrt{d} (d+\log(1/\alpha))\log(2+\log(1/\alpha))^2 \right) \;.
\end{equation}
\item In case (2), the set $T$ is not $\alpha$-good.
\item In case (3), if $T$ is $\alpha$-good then at least one $T_i$ is $\alpha_i$-good
with respect to $G$.
\end{enumerate}
The algorithm runs in time $O\left( (d+\ln(1/\tau))/\alpha \right)^d \cdot \poly(|T|,n^d)$.
\end{proposition}


The proof of Proposition \ref{mainSubroutineProp} requires several ingredients and 
is given in the following subsections. We start by showing how Theorem~\ref{thm:main-ld} 
follows from this proposition.

\begin{proof}[Proof of Theorem~\ref{thm:main-ld} given Proposition \ref{mainSubroutineProp}]
The overall algorithm proceeds as follows:
Using Lemma \ref{goodTLem}, we produce a list $\{T_1, \ldots, T_s\}$ of subsets of $T$
at least one of which is $\alpha/2$-good with respect to $G$.

The algorithm maintains a list of pairs $\{ (T_i, \alpha_i) \}_{i=1}^m$
with the guarantee that with high probability
at least one of the $T_i$ is $\alpha_i$-good with respect to $G$.
It then iteratively applies the routine from Proposition \ref{mainSubroutineProp},
 to each pair $(T_i, \alpha_i)$ in the list,
and when its output is not a vector, it replaces the pair $(T_i,\alpha_i)$ by
either zero, one, or two such pairs in the list.
Then it outputs the list of all vectors that Proposition~\ref{mainSubroutineProp} ever returned.

\new{
We will show that this process produces a list of $O(1/\alpha^3)$ many candidate hypotheses. 
We provide an efficient black-box way in which we can reduce the size of this list to $O(1/\alpha)$ 
without losing more than a constant factor in the final error guarantee. 
See Proposition \ref{prop:list-reduce}.}

The pseudo-code for the algorithm is given below:

\medskip

\fbox{\parbox{6.1in}{
{\bf Algorithm} {\tt List-Decode-Gaussian}\\
Input: Multiset $T \subset \R^n$ and parameters $0< \alpha < 1/2, d \in \Z_+, \tau>0.$\\
\vspace{-0.7cm}
\begin{enumerate}
\item Run {\tt NaiveClustering}$(T,\alpha)$ from Lemma \ref{goodTLem}.
\item Let $L$ be the list $\{ (T_i,\alpha/2) \}_i$, for 
$T_i$ returned by {\tt NaiveClustering} and let $M$ be the empty list.
\item While $L$ is non-empty:
\begin{enumerate}
\item Get the first element $(T',\alpha')$ from $L$ and remove it from the list.
\item Run {\tt MainSubroutine}($T', \alpha', d, \tau/|T|$).
\item If this routine returns a vector, add it to $M$. 
\item If it returns a list of $(T_i,\alpha_i)$, append that to $L$.
\end{enumerate}
\item Run {\tt ListReduction} (of Proposition \ref{prop:list-reduce}) to the list \new{$M$} and return its output.
\end{enumerate}
}}

\bigskip

We now proceed with the analysis.
First, it is easy to see that $\sum_i \alpha_i^{-2}$ can only decrease.
Therefore, since all $\alpha_i \leq 1$, we can never have more than
$O(\alpha^{-\new{3}})$ many terms in the list. 
Each iteration of the algorithm either increases the number of $T_i$'s or decreases the size
of some $T_i$ by at least $1$. Therefore, the algorithm
terminates after at most $O(|T|/\alpha^{\new{3}})$ many calls
to {\tt MainSubroutine}.

We claim that with probability at least $1-\tau$, at every step,
we either have that one of the $T_i$'s is $\alpha_i$-good, or
we have already produced a good approximation to $\mu$.
We can view the set of all $T_i$'s ever produced
by the algorithm as a tree, where the children of each $T_i$
are all subsets produced by {\tt MainSubroutine} on $T_i$.
We say that $T_i$ belongs to the $j$-th generation if the path from $T$ to $T_i$ in this tree
has length $j$. Then, it is easy to see inductively that either
at least one $T_i$ in the $j$-th generation is $\alpha_i$-good
or a good approximation to the mean was output in an earlier generation,
with probability at least $j \tau/|T|$. 
This is because given a $T_i$ that is $\alpha_i$-good in the $j^{th}$ generation, 
with probability at least $1-\tau/|T|$ it either produces a good approximation 
to the mean or has at least one descendant $T'$ that is $\alpha'$-good. 
Since there can only be $|T|$ generations,
the list of vectors produced contains a good approximation to the mean
with probability at least $1-\tau$. This completes the proof.

To analyze the runtime, note
that there are at most $|T|$ generations each involving running 
{\tt MainSubroutine} 
on $O(\alpha^{-3})$ many $T_i$'s. This clearly dominates 
the runtime from the initial use of {\tt NaiveClustering}, 
the final use of {\tt ListReduction} and the other bookkeeping steps.
\end{proof}

\medskip

In the remaining sections, we describe a number of subroutines
that either verify some desired properties of low-degree polynomials
or produce an output list of subsets satisfying the following condition:

\begin{definition}[Multifilter Condition] \label{def:mf-cond}
We say that a list of pairs $\{(T_i, \alpha_i)\}$, where $T_i \subset T$
and $\alpha_i \in (0, 1)$, satisfies the {\em multifilter condition for $(T, \alpha)$}
if the following hold:
\begin{enumerate}
\item \label{eqn:sos} $\sum_i 1/\alpha_i^2 \leq 1/\alpha^2$, and
\item \label{eqn:good-Ti} If $T$ is $\alpha$-good (with respect to $G$),
then at least one of the $T_i$'s is $\alpha_i$-good (with respect to $G$).
\end{enumerate}
\end{definition}

\subsection{Basic Multifilter Routine} \label{ssec:basic-mf}

How do we produce the multifilters for corrupted sets $T$ required
in Proposition~\ref{mainSubroutineProp}? The basic idea is the following:
If we find a degree-$d$ \new{$n$-variable} polynomial $p$
such that the mean and variance of $p(T)$ behave substantially differently
from the mean and variance of $p(G)$, we will be able to split $T$
into \new{(overlapping)} subsets --- based on the values of $p(x)$, $x \in T$ ---
\new{with the property that at least one of these subsets is ``cleaner'' than $T$.}

\new{The obvious difficulty is that the mean $\mu$ of $G$
is unknown, and therefore so is the mean and variance of $p(G)$.
Suppose for now that we have a low-degree polynomial $p$}
whose variance $\var[p(G)]$ is small. \blue{(As we argue in the following sections,
we can guarantee that this condition is satisfied with high probability for the polynomials
$p$ that we will use.)} By concentration, it then follows that \new{at least a $(1-\alpha/2)$
fraction of} the points of \new{$\{p(x), x\in S\}$} will lie in an interval of length \new{$\log(1/\alpha)^{O(d)}$}.
If a decent number of points of \new{$\{p(x), x\in T\}$} fall outside an interval of this length,
\new{then we know that the distribution of $p(T)$ is very different}. In such a case,
be removing points of $T$ beyond some threshold, we can get rid
of substantially more bad points than good points.

\new{Recall that we assumed that $\var[p(G)]$ is small. However, the mean of $p(G)$ is
still unknown. This creates the following complication:}
There might be several intervals of length \new{$\log(1/\alpha)^{O(d)}$
each of them containing an $\alpha$-fraction of $T$.
Any of these intervals would make a reasonable hypothesis
and we have no way of distinguishing between them.}
In such a setting, we split our set $T$ into several {\em overlapping} subsets.
We make the overlap between them large enough so
that no matter where the (unknown) true mean of $p(G)$ is,
the vast majority of the points of \new{$\{p(x), x\in S\}$} will lie in the same subset.
\new{At the same time}, because there is a substantial number of points in $T$
on \blue{either side of the divide}, we can guarantee that not too many points
are kept on either side.


Formally, we establish the following proposition:

\begin{proposition}[Basic Multifilter] \label{prop:basic-mf}
There is a polynomial time algorithm with the following performance guarantee:
Given a set $T \subset \R^n$, $0< \alpha \leq 1$,
and a degree-$d$ polynomial $p: \R^n \to \R$ for which we have the promise that
if $T$ is $\alpha$-good then $\Var[p(G)] \leq 1$, returns one of the following:
\begin{enumerate}
\item ``YES''. If the algorithm returns ``YES'', then we have that:
\begin{itemize}
\item[(a)] $\Var[p(T)] \leq O(d+\log(1/\alpha))^{d}(\log(2+\log(1/\alpha)))^2$, and 
\item[(b)] if $T$ is $\alpha$-good, 
then $|\E[p(G)]-\E[p(T)]| \leq O(d+\log(1/\alpha))^{d/2}\log(2+\log(1/\alpha))$.
\end{itemize}

\item ``NO''. If the algorithm returns ``NO'', then the set $T$ is {\em not} $\alpha$-good.

\item A list of pairs $\{(T_i, \alpha_i)\}$, $T_i \subset T$, of length one or two
satisfying the multifilter condition for $(T, \alpha)$.
\end{enumerate}
\end{proposition}

\noindent The rest of this subsection is devoted to the proof of Proposition~\ref{prop:basic-mf}.

\subsubsection{Basic Multifilter Pseudocode} \label{sssec:mf-pseudocode}

We start by describing the basic multifilter routine in pseudocode:

\medskip

\fbox{\parbox{6.1in}{
{\bf Algorithm} {\tt BasicMultifilter}\\
Input: a degree-$d$ polynomial $p: \R^n \to \R$, $T \subset \R^n, d \in \Z_+, 0< \alpha < 1, \tau>0$.\\
\vspace{-0.7cm}
\begin{enumerate}
\item Let $C>0$ be a sufficiently large constant.

\item Let $R \eqdef (C \log(1/\alpha))^{d/2}$.

\item \label{step-first-range-check} 
If $\max_{x, y \in T} |p(x)-p(y)|  \blue{\geq (C(1+n/d))^{d/2}}$, 
then return ``NO''.

\item \label{step-one-interval}
If there is an interval $[a, b]$ of length $L \eqdef C \cdot R \cdot \log(2+ \log(1/\alpha))$ that
contains at least $(1-\alpha/2)$-fraction of points in $\{p(x), x \in T\}$, then:
\begin{enumerate}
\item \label{step:var-check}
If $\Var[p(T)] \leq C \cdot (d+C \log(1/\alpha))^{d}(\log(2+\log(1/\alpha)))^2$, return ``YES''.

\item \label{step:var-large} Otherwise, proceed as follows:
\begin{enumerate}
\item \label{step:find-t-one}  Find a threshold $t > 2R$ such that
\begin{equation} \label{eqn:t-prob}
\Pr_{X \in_u T} \left[\min\{|p(X)-a|,|p(X)-b|\} \geq t\right] > (\new{32}/\alpha)\exp(-(t-2R)^{2/d}) + \frac{4 \cdot \alpha^{\new{2}}}{d! n^{d}} \;.
\end{equation}
\item Let $T' \eqdef \{x \in T : \min\{|p(x)-a|,|p(x)-b|\}  \leq t \}$ and \\
\new{$\alpha' \eqdef \alpha \cdot \left( (1-\alpha/\new{8})|T|/|T'| +\alpha/\new{8} \right)$}.
\item  \label{step:alpha-sanity} If $\alpha' \leq 1$, return $\{(T', \alpha')\}$; else return ``NO''.
\end{enumerate}
\end{enumerate}

\item \label{step-not-one-interval}
If there is {\em no} interval of length $L$ satisfying
the condition of Step~\ref{step-one-interval}, then
\begin{enumerate}
\item \label{step:find-t-two} Find a threshold $t$ such that the sets
$T_1= \{x \in T : p(x) \new{>} t-R\}$ and
$T_2=\{x \in T : p(x) \new{<} t+R\}$ satisfy
\begin{equation} \label{eqn:t-inters}
|T_1|^2 + |T_2|^2 \leq |T|^2(1-\alpha/100)^2  \;,
\end{equation}
and $|T|-\max(|T_1|,|T_2|) \geq \alpha|T|/4.$

\item Let $\alpha_i =\alpha \cdot (1-\alpha^2/100) \cdot |T|/|T_i|$, for $i=1,2$.
\item Return $\{(T_1, \alpha_1), (T_2, \alpha_2)\}$.
\end{enumerate}
\end{enumerate}
}}

\medskip

\new{
\paragraph{Intuitive Overview of Basic Multifilter.}
The algorithm begins by defining some appropriate constants
and in particular choosing the parameter $R$ to be polylogarithmic in $1/\alpha$,
so that if $\Var[p(G)]\leq 1$ then all but a $\poly(\alpha)$-fraction of the points
of $p(G)$ must lie in an interval of length $R$.

We begin in Step \ref{step-first-range-check} by doing a basic sanity check,
ensuring that the total range of values of $p(T)$ is not too large.
On the one hand, we are guaranteed by Claim \ref{claim:T-not-good}
in the next section that this will be true if $\Var[p(G)]\leq 1$ and $T$ is good.
This assumption will become technically useful as it will imply that if $p(T)$
has too large a mean or variance, then this cannot be due to a negligible
number of points at extreme distances from the mean,
and instead must be due to a relatively significant number of errors.

We split our analysis into cases roughly corresponding
to the one cluster versus two clusters cases. Step \ref{step-one-interval}
corresponds to the case where there is a single interval of length
not much more than $R$ that contains all but an $\alpha/2$-fraction
of the points of $\{p(x), x \in T\}$. This implies that since an $\alpha$-fraction
of points come from $p(G)$, the mean of $p(G)$ cannot be far outside of this interval.
Thus, we can construct a slightly larger interval that contains all but a tiny fraction
of the points of $p(G)$. From there, it is a relatively easy to show
that since $p(T)$ has substantially different behavior from $p(G)$,
there must be some (potentially larger interval), where the fraction
of points of $p(T)$ outside of this interval is much larger
than the fraction of points of $p(G)$ outside of this interval.
Thus, throwing away the points outside of the latter interval
will serve to clean up our dataset.

Step~\ref{step-not-one-interval} covers the other case where no short interval
contains all but an $\alpha/2$-fraction of the points. In this case,
we try to partition $T$ into sets based on whether $p(x)$ lies in
$(\infty, t+R)$ or $(t-R,\infty)$, for some appropriate threshold $t$.
Notice that since these intervals overlap on an interval of length $2R$,
no matter where the mean of $p(G)$ lies, almost all of the points of $p(G)$
will lie in one of the two intervals (if it's in the middle it could even be the
case that almost all points lie in both). We now need to pick a threshold $t$
so that the number of elements of $p(T)$ \emph{not} on either side is substantial.
We are aided in this endeavor by the fact that there must be some gap
of length $\Omega(R\log(2+\log(1/\alpha)))$ with at least an $\Omega(\alpha)$-fraction
of the points lying on either side of it. By trying many possible values if $t$
within this interval, it is not hard to show that there must be some threshold
where neither side keeps too many.
}

\subsubsection{First Stage of Correctness Proof} \label{sssec:basic-mf-an1}

In this subsection, we prove the correctness conditions of our algorithm
that do not establish goodness of the output.

We start by showing that if the algorithm returns ``NO'' in Step~\ref{step-first-range-check},
then the set $T$ is not $\alpha$-good:

\blue{

\begin{claim} \label{claim:T-not-good}
Suppose that $T$ is $\alpha$-good and $\Var[p(G)] \leq 1$.
Then, it holds that $\max_{x, y \in T} |p(x)-p(y)| \leq O(1+n/d)^{d/2}$.
\end{claim}
\begin{proof}
We show that for any such $p$ with $\E[p(G)]=0$
we have that $|p(x)| \leq O(1+n/d)^{d/2}$
whenever $\|x-\mu\|_2=O(\sqrt{n})$.
Note that
$$|p(y+\mu)|=\left|\sum_{a\in \Z_{\geq 0}^n:\|a\|_1\leq d} c_a \He_a(y)/\sqrt{a!} \right|
\leq \sqrt{\left(\sum_{a\in \Z_{\geq 0}^n:\|a\|_1\leq d} c_a^2 \right)
\left(\sum_{a\in \Z_{\geq 0}^n:\|a\|_1\leq d}\He_a(y)^2/a! \right)} \;.$$
Since $\Var[p(G)] = \E[p^2(G)] = \sum_{a\in \Z_{\geq 0}^n:\|a\|_1\leq d}c_a^2$,
it suffices to bound the sum of $\He_a(y)^2/a!$
under the assumption that $\|y\|_2 = O(\sqrt{n})$.
By Mehler's formula~\cite{mehler1866}, we have that for any $u$ with $|u|<1$ it holds
$$
\sum_{a=0}^\infty \He_a^2(y_i) u^a/a! = \frac{1}{\sqrt{1-u^2}} e^{\frac{u}{1+u}y_i^2} \;.
$$
Taking the product over $i$ from $1$ to $n$ yields:
$$
\sum_{a\in \Z_{\geq 0}^n} \He_a(y)^2 u^{\|a\|_1}/a! = (1-u^2)^{-n/2} e^{\frac{u}{1+u}\|y\|_2^2} \;.
$$
If $u\geq 0$, we have that the left hand side is at least
$$
u^d \sum_{a\in \Z_{\geq 0}^n:\|a\|_1\leq d} \He_a(y)^2/a! \;.
$$
Therefore, we have that
$$
\sum_{a\in \Z_{\geq 0}^n:\|a\|_1\leq d} \He_a(y)^2/a! \leq u^{-d} (1-u^2)^{-n/2}e^{\frac{u}{1+u}\|y\|_2^2} \;.
$$
Letting $u=\min(1/2,d/n)$, we find that if $d\ll n$, then the above sum
is at most $(n/d)^d \exp(O(d^2/n)+O(d\|y\|_2^2/n)) = O(n/d)^d$.
If $d\gg n$, this is $O(1)^{d+n+\|y\|_2^2}$. In either case,
it is $O(1+n/d)^d$, and thus $|p(y)|$ is at most $O(1+n/d)^{d/2}$.
This completes the proof of Claim~\ref{claim:T-not-good}.
\end{proof}

}

\new{The correctness of the other step in which the algorithm can output ``NO'',
Step \ref{step:alpha-sanity}, is deferred to after Lemma \ref{lem:good-single-T} below.}

We next show that the algorithm
finds an appropriate threshold $t$ in Steps~\ref{step:find-t-one}
and~\ref{step:find-t-two}.

\begin{lemma} \label{lem:find-t-one}
If the algorithm reaches Step~\ref{step:find-t-one},
there exists a threshold $t>2R$ satisfying \eqref{eqn:t-prob}.
\end{lemma}
\begin{proof}
The main idea of the proof is quite simple:
If there is no $t > 2R$ satisfying  \eqref{eqn:t-prob},
the random variable $p(T)$ is well-concentrated, which
contradicts the lower bound on its variance
that holds in Step \ref{step:var-large}.
Suppose, for the sake of contradiction, that for all $t  > 2R$, it holds:
$$\Pr_{X \in_u T}[ \min\{|p(X)-a|,|p(X)-b|\} \geq t] \leq (32/\alpha) \exp(-(t-R)^{2/d}) + \new{4} \alpha^2/(d! n^{d}) \;.$$
By a change of variable, this is equivalent to the following statement.
For all $t >  2R + (b-a)/2$, it holds:
$$\Pr_{X \in_u T}[ |p(X)-(a+b)/2| \geq t] \leq (32/\alpha) \exp(-(t-R-(b-a)/2)^{2/d}) + \new{4} \alpha^2/ (d! n^{d}) \;.$$
Let $\beta \eqdef \ln(32/\alpha)^{d/2}$.
Note that the above RHS is at least $1$ when $t < \beta+R+(b-a)/2$, hence the inequality
is trivial unless $t \geq \beta+R+(b-a)/2$

The proof will require some basic properties of the Gamma function.
We will use the following notation:
Let $\Gamma(s,x)=\int_x^\infty \exp(-t) t^{s-1} dt$ be the incomplete gamma function
and $\Gamma(s)=\Gamma(s,0)$ be the Gamma function.
We will need the following technical claim about the incomplete Gamma function,
whose proof can be found in Appendix~\ref{sec:gamma}:
\begin{claim} \label{claim:gamma}
For $s \geq 1, x \geq 0$, we have that $\Gamma(s,x) \leq \exp(-x) (x+s)^{s-1}$.
\end{claim}
We can now bound from below the variance of $p(T)$ as follows:
\begin{align*}
&\Var[p(T)] \leq \E_{X \in_u T}[(p(X)-(a+b)/2)^2] \\
&= \int_0^{O(d! n^{d/2})} 2 \Pr_{X \in_u T}[|p(X)-(a+b)/2| \geq t] t dt \\
& \leq \int_0^{O(d! n^{d/2})} \min \left\{ 1, (32/\alpha) \exp(-(t-R-(b-a)/2)^{2/d}) + 4\alpha^2/(d! n^{d}) \right\} 2t dt \\
& \leq (\beta+R+(b-a)/2)^2 + O(4 \cdot d! \alpha^2) +
(32/\alpha) \int_{\beta}^{\infty} \exp(-t^{2/d}) (2t + 2R+b-a) dt \\
& \leq O((C\log(1/\alpha)+d)^{d}(\log(2+\log(1/\alpha)))^2) +
(32/\alpha) \int_{\ln(32/\alpha)}^{\infty} \exp(-u) (2u^{d/2} + 2R+b-a)  (d/2) u^{d/2-1} du \\
& = O((C\log(1/\alpha)+d)^{d}(\log(2+\log(1/\alpha)))^2) + (32/\alpha) \cdot d \cdot \Gamma(d, \ln(32/\alpha)) \\
&+ (32/\alpha)(2R+b-a) \Gamma(d/2, \ln(32/\alpha))(d/2) \;.
\end{align*}
Using Claim~\ref{claim:gamma} we can write:
\begin{align*}
\Var[p(T)]
& \leq O((C\log(1/\alpha)+d)^{d}(\log(2+\log(1/\alpha)))^2) + d \cdot (d+\ln(32/\alpha))^{d-1} + \\
& O((C \log(1/\alpha))^{d/2+1}) \cdot (d+\ln(32/\alpha))^{d/2-1} \\
& \leq O((d+C \log(1/\alpha))^{d}(\log(2+\log(1/\alpha)))^2) \;.
\end{align*}
By the construction of the algorithm, when the variance $\Var[p(T)]$ is this small,
the algorithm does not reach Step~\ref{step:find-t-one}.
This gives the desired contradiction and completes the proof of Lemma~\ref{lem:find-t-one}.
\end{proof}

We next show that the algorithm finds an appropriate threshold
in Step~\ref{step:find-t-two}:
\begin{lemma} \label{lem:find-t-two}
If the algorithm reaches Step~\ref{step:find-t-two}, there exists a threshold $t$
satisfying \eqref{eqn:t-inters}.
\end{lemma}
\begin{proof}
We will show that if no $t$ satisfies \eqref{eqn:t-inters},
then at least a $1-\alpha/2$ fraction of points $x \in T$ have
$p(x)$ in an interval of length $O(\log(2+\log(1/\alpha)) \cdot R)$.
This would be a contradiction, as no such interval exists
if the algorithm reaches Step~\ref{step-not-one-interval}. 
This is essentially because if there is no such $t$, 
then if we consider the fraction of $x\in T$ with $p(x)\leq t$, 
this must grow very quickly with $t$ after it first exceeds $\alpha/4$, 
and thus there is very little distance between the $\alpha/4$-tails of the distribution and the median. 
This means that the $\alpha/4$-tails on the left and on the right must be close to each other, 
implying that the interval from Step \ref{step-one-interval} must exist.

Let $t_m=2m \cdot R + \textrm{median}(p(T))$, for integer $m$.
If no $t$ satisfies \eqref{eqn:t-inters}, then $t_m$
is not a suitable choice of $t$, and therefore
\begin{equation}\label{sizeRelEqun}
|T_{1,m}|^2+|T_{2,m}|^2 > |T|^2(1-\alpha^2/100)^2 \;,
\end{equation}
where  $T_{1,m}=\{x \in T:p(x) \new{>} t_m-R\}$ and 
$T_{2,m}=\{x \in T:p(x) \new{<} t_m+R\}$ for all $m$ with $|T_{i,m}|\geq \alpha|T|$.

Let $a_m = |T_{1,m}|/|T|$. 
Notice that $|T_{2,m}|/|T| = (1-a_{m+1})$. If Equation~\eqref{sizeRelEqun} 
holds for a given $m$, then we must have $a_m^2+(1-a_{m+1})^2 > (1-\alpha^2/100)^2.$ 
This means that 
$a_{m+1} < (a_m^2 +\alpha^2/16)$. 
If $a_m\geq \alpha/4$, then $a_{m+1} \leq 2a_m^2$. 
Since $a_1 < 1/2$, it is easy to see that for some $m=\log_2\log(1/\alpha) +O(1)$ 
that $a_m \leq \alpha/4$. Thus, $|\{x\in T:p(x)> t + (2m+1)R\}|\leq|T|\alpha/4$. 
Similarly, we can see that $|\{x\in T:p(x)< t - (2m+1)R\}|\leq|T|\alpha/4$. 
This gives us an interval of length $O(R\log(2+\log(1/\alpha)))$ 
that contains all but an $\alpha/2$-fraction of the points of $p(T)$.

This provides the desired contradiction 
and gives the proof of Lemma~\ref{lem:find-t-two}.
\end{proof}

We now show that if the algorithm outputs a list of length one,
the parameter $\alpha$ improves:

\begin{lemma} \label{lem:monotone-alpha}
If the algorithm outputs a single pair $(T', \alpha')$,
then $\alpha' \geq \alpha$.
\end{lemma}
\begin{proof}
By definition, we have that
$$\alpha' - \alpha = \left( (1-\alpha/8)(|T|/|T'|) + \alpha/8 -1 \right) \alpha = (1-\alpha/8)(|T|/|T'|-1)\alpha \;.$$
Note that $\alpha' - \alpha >0$ when $|T'| < |T|$.
Moreover, $|T'| < |T|$ holds true because the algorithm is
guaranteed to remove at least one point,
since $$\Pr_{X \in_u T}[ \min\{|p(X)-a|,|p(X)-b|\} \geq t]$$ is non-zero.
This completes the proof of Lemma~\ref{lem:monotone-alpha}.
\end{proof}

The next lemma shows that whenever the algorithm outputs two sets $T_1, T_2$,
the associated parameters $\alpha_1, \alpha_2$ satisfy $1/\alpha_1^2 + 1/\alpha_2^2 \leq 1/\alpha^2$.
This condition will be crucial for bounding the runtime of the overall algorithm. We also note
that this condition holds even if the set $T$ was not $\alpha$-good.

\begin{lemma} \label{lem:sos}
If the algorithm outputs a list of two pairs $\{(T_1,\alpha_1), (T_2, \alpha_2)\}$,
then we have that $1/\alpha_1^2 + 1/\alpha_2^2 \leq 1/\alpha^2$.
\end{lemma}
\begin{proof}
We have that
\begin{eqnarray*}
\alpha_1^{-2} + \alpha_2^{-2} 
&=& \alpha^{-2}(1-\alpha^2/100)^{-2} |T_1|^2/|T|^2 + \alpha^{-2}(1-\alpha^2/100)^{-2} |T_2|^2/|T|^2 \\
&=& \alpha^{-2}(1-\alpha^2/100)^{-2} (|T_1|^2+|T_2|^2/|T|^2 \\
&\leq& \alpha^{-2} \;.
\end{eqnarray*}
\end{proof}

\subsubsection{Second Stage of Correctness Proof} \label{sssec:basic-mf-an2}
In this subsection, we establish the correctness conditions of the algorithm
that hold under the assumption that the set $T$ is $\alpha$-good.
Specifically, we will show that if $T$ is $\alpha$-good the following hold:
(1) If the algorithm returns ``YES'', the means of $p(G)$ and $p(T)$ are close to each other.
(2) If the algorithm returns a subset $T'$, then $T'$ is good.
(3) If the algorithm returns two subsets $T_1, T_2$, then
at least one of the returned sets is good.

We start by proving (1):

\begin{lemma} \label{lem:close-means}
If $T$ is $\alpha$-good and the algorithm returns ``YES'',
then $|\E[p(G)]-\E[p(T)]| \leq O((d + C \log(1/\alpha))^{d/2}\log(2+\log(1/\alpha)))$.
\end{lemma}

To prove this lemma, we will require two intermediate claims.
We start by showing that most of the good points in $T$
have $p(x)$ close to $\E[p(G)]$:

\begin{claim} \label{lem:good-points-location}
If $T$ is $\alpha$-good,  the samples $x \in T \cap S$ that satisfy
$|p(x)-\E[p(G)]| < R$ constitute at least an $(\alpha -  \alpha^2/100)$-fraction of $T$
and an $(1-\alpha/\new{6}-\exp(-3n) - \alpha^2/100)$-fraction of $S$.
\end{claim}
\begin{proof}
Since $T$ is $\alpha$-good, we have that $\Var[p(G)] \leq 1$.
By the degree-$d$ Chernoff bound and the definition of $R$ we get
$$\Pr_{X \in_u S}[|p(X)-\E[p(G)]| \geq R] \leq \alpha^3/100 \;.$$
We can now apply the bounds given by the definition of $T$ being $\alpha$-good,
which give the claim.
\end{proof}

We next show that if we find an appropriate interval $[a, b]$ in Step \ref{step-one-interval},
then the mean of $p(G)$ is contained in the interval $[a-R, b+R]$.

\begin{claim} \label{lem:where-is-mean}
If $T$ is $\alpha$-good and the interval $[a, b]$ contains at least $(1-\alpha/2)$
fraction of values of $p(x)$ for $x \in T$, then $\E[p(G)] \in [a-R, b+R]$.
\end{claim}
\begin{proof}
By Claim~\ref{lem:good-points-location}, at least an  $(\alpha -  \alpha^2/100)$-fraction
of $T$ has $|p(x)-\E[p(G)]| \leq R$. But only an $\alpha/2$-fraction of points in $T$ are outside of $[a, b]$.
Since $\alpha-\alpha^2/100 > \alpha/2$, there must be some point which does both, and therefore
$\E[p(G)] \in [a-R, b+R]$.
\end{proof}

We are now ready to prove Lemma~\ref{lem:close-means}.

\begin{proof}[Proof of Lemma~\ref{lem:close-means}]
By construction, if the algorithm returns ``YES'',
then $$\Var[p(T)] \leq  (d+C \log(1/\alpha))^{d}(\log(2+\log(1/\alpha)))^2 \;,$$ and
$[a, b]$ is an interval  of length at most $O((C\log(1/\alpha))^{d/2}\log(2+\log(1/\alpha)))$ that
contains all except an $\alpha/2$-fraction of points of $p(x)$, $x \in T$.
We first consider the contribution of the points in $[a, b]$ to the variance
to obtain that
$$\Var[p(T)]  \geq (1-\alpha/2)\max\left(0, |\E[p(T)]-(a+b)/2| - C R\log(1/\alpha) \right)^2 \;,$$
and thus
$$|\E[p(T)]-(a+b)/2| = O((d + C \log(1/\alpha))^{d/2}\log(2+\log(1/\alpha)) ) \;.$$
Using Claim~\ref{lem:where-is-mean}, we get that
$|\E[p(G)]- (a+b)/2| \leq O(C R\log(2+\log(1/\alpha)))$.
The triangle inequality completes the proof.
\end{proof}

We now prove (2).
This proof requires a careful accounting of the number
of points from $S$ and points from $T$ that are removed by our algorithm.

\begin{lemma} \label{lem:good-single-T}
If $T$ is $\alpha$-good and the algorithm returns
a single pair $(T', \alpha')$, then $T'$ is $\alpha'$-good.
\end{lemma}
\begin{proof}
We need to show that $|S \cap T'|/|T'| \geq \alpha'$ and $|S \cap T'|/|S| \geq 1- \alpha'/6 - \exp(-3n)$.
By Claim~\ref{lem:where-is-mean}, we have that
$\E[p(G)] \in [a-R,b+R]$. Since the set $S$ is representative, we have that
$$\Pr_{X \in_u S}\left[\min\{|p(X)-a|,|p(X)-b|\} \geq t\right] \leq 2 \exp(-(t-2R)^{2/d}) + \new{\alpha^3/(d! n^{d})} \;.$$
Since $T$ is $\alpha$-good, it follows that $S \cap T$ contains a
$1-\alpha/6+\exp(-3n) \geq 1/2$ fraction of the points in $S$, hence
$$\Pr_{X \in_u (S \cap T)}[ \min\{|p(X)-a|,|p(X)-b|\} \geq t] \new{\leq} 4 \exp(-(t-2R)^{2/d}) + 2 \cdot \alpha^3/(d! n^{d}) \;.$$
But note that the probability of the above event
is at least $8/\alpha$ times bigger for $X \in_u T$.
Since $T'$ has the filtered samples removed, we have that
$$|T \setminus T'|/|T| \geq (8/\alpha) |S \cap(T \setminus T')|/|S \cap T| \;,$$
and so $|S \cap T'|/|S \cap T| \geq 1 - (\alpha/8)(1-|T'|/|T|)$.
Now we can bound from below the fraction of good samples in $T'$:
\begin{align*}
|S \cap T'|/|T'| &=  (|S \cap T'|/|S \cap T|)(|S \cap T|/|T|)(|T|/|T'|) \\
& \geq (1-\alpha/8 (1-|T'|/|T|)) \alpha (|T|/|T'|)  \\
& = ((1-\alpha/8)(|T|/|T'|) + \alpha/8) \alpha \\
&= \alpha' \;.
\end{align*}
It remains to show that $|S \cap T'|/|S| \geq 1- \alpha'/6 - \exp(-3n)$.
Noting that $(1-\alpha/8 (1-|T'|/|T|)) \alpha = \alpha'|T'|/|T|$, we get that:
\begin{align*}
|S \cap T'|/|S|
& = (|S \cap T'|/|S \cap T|)(|S \cap T|/|S|) \\
& \geq (1-(\alpha/8) (1-|T'|/|T|)) (1-\alpha/6 - \exp(-3n)) \\
& \geq (1-(\alpha/8) (1-|T'|/|T|)) (1-\alpha/6) - \exp(-3n) \\
& =   1-(\alpha/8) (1-|T'|/|T|) - \alpha'|T'|/6|T|- \exp(-3n) \;,
\end{align*}
and so
\begin{align*}
& |S \cap T'|/|S| - (1-\alpha'/6 -\exp(-3n)) \\
\geq & - (\alpha/8) (1-|T'|/|T|) + (\alpha'/6) (1-|T'|/|T|) \\
\geq &(\alpha'-\alpha) (1-|T'|/|T|)/24 \\
 \geq & 0 \;.
\end{align*}
This completes the proof.
\end{proof}

Note that this proof shows that if $T$ is $\alpha$-good, then $T'$ is $\alpha'$-good.
If $\alpha' > 1$, $T$ cannot be $\alpha'$-good, 
as this would require that $|S \cap T'| > |T'|$. 
Thus, if the algorithm outputs ``NO'' in Step~\ref{step:alpha-sanity}, 
then $\alpha' > 1$ implies that $T$ cannot be $\alpha$-good.

Our last lemma establishes (3), handling the case when the algorithm
outputs a list of two sets:

\begin{lemma} \label{lem:good-two-T}
If $T$ is $\alpha$-good and the algorithm outputs a list
$\{(T_1, \alpha_1), (T_2, \alpha_2)\}$, then for some $i \in \{1, 2\}$
$T_i$ is $\alpha_i$-good.
\end{lemma}
\begin{proof}
Recall that Claim~\ref{lem:good-points-location} gives a lower bound on the number of
points $x \in S \cap T$ that satisfy $|p(x)-\E[p(x)]| < R$.
Since $T_1$ and $T_2$ overlap in an interval of length $2R$,
these samples are contained entirely in either $T_1$ or $T_2$.
Let $i \in \{1, 2\}$ be such that $T_i$ contains these samples.
We claim that this $T_i$ is $\alpha_i$-good.

Note that $|S\cap T| \geq |S|(1-\alpha/6-e^{-3n})$. 
We have that $|S\cap (T\backslash T')| \leq |S|\alpha^2/100$. 
Therefore, as long as $\alpha_i/6 \geq \alpha/6 + \alpha^2/100$, this set $T_i$
contains enough points of $S$. These points consist
of at least $|T|(\alpha -  \alpha^2/100)/|T_i|$ fraction of $T_i$.
Thus, when
$$\alpha_i \geq \max \left\{ \alpha + 3\alpha^2/50, (\alpha -  \alpha^2/100) |T|/|T_i| \right\} \;,$$
the set $T_i$ is $\alpha_i$-good.
It remains to show that $\alpha_i = |T| \alpha (1-\alpha/100)/|T_i|
\geq \alpha(1+3\alpha/50)$.

Recall that $|T|-\max\{|T_1|,|T_2|\} \geq \alpha \cdot |T|/4$. Therefore, $|T|/|T_i| \geq (1-\alpha/4)^{-1}$, 
and thus $\alpha_i \geq \alpha(1-\alpha/100)/(1-\alpha/4) \geq \alpha(1+3\alpha/50)$. This completes our proof.
\end{proof}

In the following subsections, we appropriately use the basic multifilter algorithm to handle increasingly
broad families of degree-$d$ polynomials.

\subsection{Useful Results on Polynomials and Tensors} \label{ssec:polys}

Our algorithm and its analysis will make essential use of degree-$d$ multivariate polynomials.
Thus, we will need to establish here some basic facts and terminology.
A basic theme throughout will be that ``pure'' degree-$d$ multivariate polynomials
are in $1-1$ correspondence with symmetric order-$d$ tensors. On the other hand,
we will require three different notions of what it means for a polynomial to be ``pure'' degree-$d$:
homogeneous polynomials, symmetric multilinear polynomials, and harmonic polynomials
(which are useful when considering norms with respect to the Gaussian distribution).

The structure of this section is as follows: In Section~\ref{ssec:polys-def}, we give
the definitions of the various types of polynomials we will require and their associated tensors.
In Section~\ref{ssec:polys-structural}, we establish a key lemma (Lemma~\ref{lem:poly-relations}) 
relating the first two moments of such polynomials under an identity covariance Gaussian. This lemma 
is essential for our algorithm and its analysis.


\subsubsection{Definitions and Basic Facts} \label{ssec:polys-def}

\paragraph{Homogeneous Polynomials.}
A multivariate degree-$d$ polynomial is called {\em homogeneous} if it only contains monomials 
of degree exactly $d$. Given a symmetric order-$d$ tensor over $\R^n$, $A$, 
we define $\hmg_A(x)$, $x \in \R^n$, to be the homogeneous $n$-variate degree-$d$ polynomial with
$\hmg_A(x)= A(x, x,\ldots,x)$.
The polynomial $\hmg_A(x)$ can also be characterized
by the following fact:
\begin{fact} 
Let $A$ be a symmetric order-$d$ tensor over $\R^n$. We have the following:
\begin{itemize}
\item[(i)] $\hmg_A(x) = \sum_{i_1,\dots,i_d} A_{i_1,\dots,i_d} \prod_j x_j^{c_j(i_1,\dots,i_d)}$,
where $c_j(i_1,\dots,i_d)$ is the number of times $j$ appears in $i_1,\dots,i_d$.
\item[(ii)] Any $n$-variate degree-$d$ homogeneous polynomial $p$ satisfies $p(x)=\hmg_A(x)$,
where
$$d! \cdot A_{i_1,\dots, i_d} = \frac{\partial}{\partial x_{i_1}} \dots \frac{\partial}{\partial x_{i_d}} p(x) \;.$$
\end{itemize}
\end{fact}

\paragraph{Harmonic Polynomials.}
A multivariate degree-$d$ polynomial $p(x)$ is called {\em harmonic} (of degree-$d$)
if it is orthogonal to all lower degree polynomials under the inner product weighted by the standard Gaussian,
i.e., $\E_{X \sim N(0,I)}[p(X)q(X)]=0$ for all polynomials $q$ of degree at most $d-1$.
Equivalently, we can write $p$ as a linear combination of multivariate Hermite polynomials \new{of degree exactly $d$}.

The {\em degree-$d$ harmonic part} of a polynomial $p$ is the sum 
of the degree-$d$ terms of $p$ when expressed \new{in the basis of Hermite polynomials}.
For a symmetric order-$d$ tensor $A$, we define $h_A(x)$ to be the degree-$d$ harmonic part of 
the polynomial $\hmg_A(x)/\sqrt{d!}$. This can also be characterized by the following:

\begin{fact} \label{fact:harmonic}
Let $A$ be a symmetric order-$d$ tensor over $\R^n$. 
We have the following:
\begin{itemize}
\item[(i)] $\sqrt{d!} \cdot h_A(x) = \sum_{i_1,\dots, i_d} A_{i_1,\dots,i_d} \prod_j \He_{c_j(i_1,\dots,i_d)}(x_j)$,
where $c_j(i_1,\dots,i_d)$ is the number of times $j$ appears in $i_1,\dots,i_d$.
\item[(ii)] Any degree-$d$ harmonic polynomial $p(x)$ satisfies $p(x)=h_A(x)$, where
$$\sqrt{d!} \cdot A_{i_1,\dots,i_d}= \frac{\partial}{\partial x_{i_1}} \dots \frac{\partial}{\partial x_{i_d}} p(x).$$
\end{itemize}
\end{fact}

\paragraph{Symmetric Multilinear Polynomials.}
For $d, n \in \Z_+$, we can express a polynomial $p: \R^{nd} \to \R$ as $p(x_1,\dots ,x_d)$,
where $x_i \in \R^n$. A polynomial $p: \R^{nd} \to \R$ is called {\em symmetric} if
$p(x_1,\dots ,x_d)=p(x_{\sigma(1)},\dots, x_{\sigma(d)})$, for any permutation $\sigma: [d] \to [d]$.
A polynomial $p: \R^{nd} \to \R$ is called {\em multilinear} if it is linear in each of its $d$ arguments,
i.e., it holds that $p(a \cdot x_1+b \cdot x'_1,x_2,\dots, x_d)=a \cdot p(x_1,x_2,\dots, x_d) + b \cdot p(x'_1,x_2,\dots, x_d)$,
for all $a, b \in \R$ and $x_i, x'_i \in \R^n$,
and similarly for the other arguments. 

Any degree-$d$ multilinear polynomial $p: \R^{nd} \to \R$ can be expressed
as $A(x_1,\dots, x_d)$ for an order-$d$ tensor $A$ over $\R^n$.
Any degree-$d$ {\em symmetric} multilinear polynomial can be expressed as $A(x_1,\dots, x_d)$
for a {\em symmetric} order-$d$ tensor $A$ over $\R^n$.

\paragraph{Normalized Polynomials.}
Recall that $\|A\|_2$ denotes the $\ell_2$-norm of the entries of the tensor $A$, and that
$\|p\|_2 = \E_{X \sim N(0,I)}[p(X)^2]^{1/2}$, for $p: \R^n \to \R$.
For a degree-$d$ harmonic polynomial $h_A(x)$ with corresponding tensor $A$,
the following simple claim relates their $2$-norms:
\begin{claim} \label{claim:2-nom}
For all order-$d$ symmetric tensors $A$, we have that
$\|h_A\|_2=\|A\|_2$. Moreover, if $d > 0$, then $\E_{X \sim N(0,I)}[h_A(X)]=0$.
\end{claim}

\begin{definition}[Normalized Polynomial]
We call a homogeneous, harmonic or symmetric multilinear polynomial {\em normalized}
if the associated tensor $A$ has $\|A\|_2=1$.
\end{definition}

\subsubsection{Useful Lemma on Polynomials} \label{ssec:polys-structural}

In this section, we establish a crucial lemma relating 
the aforementioned classes of polynomials 
with respect to an identity covariance Gaussian distribution.

The top level of our list-decoding algorithm (Section~\ref{ssec:combining}) proceeds by computing the top eigenvector 
of a matrix to find a degree-$d$ harmonic polynomial whose empirical variance is large. 
In order to construct a multifilter from such a polynomial, we need to bound its variance on the true Gaussian, 
which determines its concentration. In general, the variance of a polynomial of an identity covariance Gaussian depends on the mean. 
Lemma~\ref{lem:poly-relations} expresses the expectation and variance of harmonic polynomials of such a Gaussian 
in terms of homogeneous and symmetric multilinear polynomials. Using this connection, we can bound the variance 
of a harmonic polynomial by running multifilters on symmetric multilinear polynomials, 
which are easier to analyze because the variance of linear polynomials of Gaussians 
does not depend on the mean.

\begin{lemma} \label{lem:poly-relations}
Let $X,X_{(1)},\dots, X_{(d)}$ be i.i.d. random variables distributed as $N(\mu,I)$ for some $\mu \in \R^n$.
Then for any order-$d$ symmetric tensor $A$, we have
\begin{equation} \label{eqn:harmonic-exp}
\sqrt{d!} \cdot \E[h_A(X)]=\hmg_A(\mu)=\E[A(X_{(1)},\dots, X_{(d)})] \;,
\end{equation}
and
\begin{equation} \label{eqn:harmonic-exp2}
\E[h_A(X)^2]= \sum_{d'=0}^d \left({d \choose d'}/{(d'-d)!}\right) \cdot \hmg_{B^{(d')}}\new{(\mu)} \;,
\end{equation}
where $B^{(d')}$ is the order-$2d'$ tensor with
$$B^{(d')}_{i_1,\dots, i_{d'}, j_1,\dots ,j_{d'}} =
\sum_{k_{d'+1},\dots, k_d} A_{i_1,\dots, i_{d'}, k_{d'+1},\dots, k_d}  A_{j_1,\dots ,j_{d'}, k_{d'+1},\dots, k_d} \;.$$
\end{lemma}
\begin{proof}
Note that $Y=X-\mu$ is distributed as $N(0,I)$.
We know that $\E[\He_j(Y_i)]=\delta_{j0}$, for all $i, j$, and so we need to express
$h_A(Y+\mu)$ in terms of these expectations. By standard properties of Hermite polynomials, we have that
$d\He_i(x)/dx = i \cdot \He_{i-1}(x)$, and so by Taylor's theorem we can write:
\begin{equation} \label{eq:taylor-hermite}
\He_j(X_i) = \He_j(Y_i+\mu_i)=\sum_{k=0}^j {j \choose k} \mu_{\new{i}}^k \He_{j-k}(Y_i) \;.
\end{equation}
When we take expectations, since $\E[\He_j(Y_i)]=\delta_{j0}$, only the $j=k$ term survives and we get that $\E[\He_j(X_i)]=\mu_i^j$.
Since the coordinates of $X$ are independent, we can take linear combinations of products of this to obtain:
\begin{align*}
\sqrt{d!} \cdot \E[h_A(X)] & = \E\left[\sum_{i_1,\dots,i_d} A_{i_1,\dots,i_d} \prod_j \He_{c_j(i_1,\dots,i_d)}(X_j)\right] \\
& = \sum_{i_1,\dots,i_d} A_{i_1,\dots,i_d} \prod_j \E[\He_{c_j(i_1,\dots,i_d)}(X_j)] \\
& = \sum_{i_1,\dots,i_d} A_{i_1,\dots,i_d} \prod_j \mu_j^{c_j(i_1,\dots,i_d)} \\
& = \hmg_A(\mu) \;.
\end{align*}
We have therefore shown the first equality in \eqref{eqn:harmonic-exp}. For the second equality, we
note that since $A(x_1,\dots, x_d)$ is linear in each argument we have
$\E[A(X_{(1)},\dots X_{(d)})]=A(\mu,\dots,\mu)= \hmg_A(\mu)$.

To get an expression for $\E[h_A(X)^2]$, we need to extend
the Taylor expansion \eqref{eq:taylor-hermite} to the multivariate setting:
\begin{claim} \label{claim:mv-Taylor}
Let $X \sim N(\mu,I)$ and $Y=X-\mu$. Then we have
\begin{equation} \label{eqn:mv-Taylor}
\sqrt{d!} \cdot h_A(X)=\sum_{d'=0}^d {d \choose d'} \sqrt{d'!} \cdot h_{A \mu^{\otimes d-d'}}(Y) \;,
\end{equation}
where $A \mu^{\otimes d-d'}$ is the order-$d'$ tensor with $i_1,\dots, i_{d'}$ entry
$\sum_{i_{d'+1},\dots, i_d} A_{i_1,\dots, i_d} \prod_{j=d'+1}^d \mu_{i_j}$.
\end{claim}
\new{
\begin{proof}
We prove this by Taylor expanding $h_A(X)=h_A(Y+\mu)$ about $Y$. 
In particular, applying the one-dimensional Taylor Theorem to $h_A(Y+t\mu)$ about $t=0$, 
we find that
$$
h_A(X) = \sum_{k=0}^d \frac{1}{k!} D_\mu^{k} h_A(Y) \;,
$$
where $D_\mu (f) = \sum_{i=1}^n \mu_i \frac{\partial f}{\partial x_i}$ 
is the directional derivative in the $\mu$-direction of $f$.

We note that $D_\mu^k h_A$ is a linear combination of $k^{th}$-order partial derivatives of $h_A$. 
Since any $k^{th}$-order partial derivative of $He_a$ 
is a harmonic polynomial of degree $\|a\|_1-k$, 
we have that $D_\mu^k h_A$ is a harmonic polynomial 
of degree $d-k$, say $h_{C_k}$. 
In order to find out which polynomial, recall that by Fact~\ref{fact:harmonic} we have
\begin{align*}
\sqrt{(d-k)!}(C_k)_{i_1,\ldots,i_{d-k}} & = \frac{\partial}{\partial x_{i_1}}\cdots \frac{\partial}{\partial x_{i_{d-k}}} D_\mu^k h_A \\
& = \sum_{j_1,j_2,\ldots,j_k}\frac{\partial}{\partial x_{i_1}}\cdots \frac{\partial}{\partial x_{i_{d-k}}} \frac{\partial}{\partial x_{j_1}}\cdots \frac{\partial}{\partial x_{j_{k}}} 
\prod_{\ell=1}^k \mu_{j_\ell} h_A\\
& = \sqrt{d!} \sum_{j_1,j_2,\ldots,j_k}  \prod_{\ell=1}^k \mu_{j_\ell} A_{i_1,\ldots,i_{d-k},j_1,\ldots,j_k}\\
& = \sqrt{d!} (A \mu^{\otimes k})_{i_1,\ldots,i_{d-k}} \;.
\end{align*}
Therefore, $C_k = \sqrt{\frac{d!}{(d-k)!}} A \mu^{\otimes k}$. Hence, we have 
$$
\sqrt{d!} \cdot h_A(X) = \sum_{k=0}^d \binom{d}{k} \sqrt{(d-k)!} h_{A\mu^{\otimes k}}(Y).
$$
Changing variables to $d'=d-k$ yields our result.
\end{proof}

}

Now note that if $B$ and $C$ are symmetric tensors of different orders,
then $\E[h_B(Y)h_C(Y)]=0$, since there are linear combinations of
different polynomials from the orthogonal basis $\prod_j \He_{b_j}(Y_j)$.
Thus, cross-terms disappear and we obtain:
\begin{align*}
\E[h_A(X)^2] &=(1/d!) \E\left[\left(\sum_{d'=0}^d {d \choose d'} \sqrt{d'!} h_{A \mu^{\otimes d-d'}}(Y) \right)^2 \right] \\
& = \sum_{d'=0}^d (1/d!) {d \choose d'}^2 d'! \E[h_{A \mu^{\otimes d-d'}}(Y)^2] \\
&= \sum_{d'=0}^d \|A \mu^{\otimes d-d'}\|_2^2 {d \choose d'} /(d'-d)! \;.
\end{align*}
Finally, note that
\begin{align*}
\|A \mu^{\otimes d-d'}\|_2^2 &= \sum_{i_1, \dots i_d'} (A \mu^{\otimes d-d'})^2_{i_1, \dots, i_d'} \\
& = \sum_{i_1, \dots, i_d'} \left(\sum_{j_{d'+1},\dots,j_d} A_{i_1, \dots i_d',j_{d'+1},\dots,j_d} \prod_k \mu_k^{c_k(j_{d'+1},\dots,j_d)} \right)^2 \\
& = B^{(d')} \mu^{\otimes 2(d-d')} \\
&= \hmg_{B^{(d')}}(\mu) \;.
\end{align*}
The proof of Lemma~\ref{lem:poly-relations} is now complete.
\end{proof}

\subsection{Multifilter Routine for Degree-$1$ and Degree-$2$ Homogeneous Polynomials} \label{ssec:deg-2}

A normalized degree-$1$ homogeneous polynomial is a polynomial of the form
$\hmg_a(x) =a^T x$, for $a \in \R^n$ with $\|a\|_2=1$.
Note that $\Var[a^T G]=1$, and this holds independently of the (unknown) mean of $G$.
Thus, in this basic case, we can apply {\tt BasicMultifilter} directly to $\hmg_a(x)$.

We proceed to handle degree-$2$ homogeneous polynomials with small trace norm.
Recall that degree-$2$ homogeneous polynomials are polynomials of the form $\hmg_A(x)=x^T A x$,
for an $n \times n$ symmetric matrix $A$. They are called normalized when $\|A\|_F = 1$.
We will need an algorithm that works under the weaker assumption that $\|A\|_{\ast} \leq 1$,
where the {\em trace norm} or {\em nuclear norm} $\|A\|_{\ast}$ is the $\ell_1$-norm of the singular values
(since $\|A\|_F$ is the $\ell_2$-norm of the singular values, we have $\|A\|_\ast \geq \|A\|_F$).
We will denote $\mu_T = \E_{X \in_T} [X]$,
i.e., $\mu_T$ is the empirical mean of $T$.

We prove the following lemma:

\begin{lemma} \label{lem:deg-2}
There is a polynomial time algorithm that, given a degree-$2$ homogeneous polynomial
$p(x)=x^T A x$ over $\R^n$ with $\|A\|_{\ast} \leq 1$ and $(T, \alpha)$, the algorithm
outputs one of ``YES'', ``NO'', or a list of $(T_i, \alpha_i)$, and has the following performance:
\begin{enumerate}
\item If the output is ``YES'', then the following holds:
if $T$ is $\alpha$-good, then 
$$|\E[p(G-\mu_T)]| \leq O((1+\log(1/\alpha)) \new{\log(2+ \log(1/\alpha))^2}) \;.$$
\item If the output is ``NO'', then $T$ is not $\alpha$-good.
\item If the output is a list $\{(T_i, \alpha_i)\}$, then the list satisfies the
multifilter condition for $(T, \alpha)$.
\end{enumerate}
\end{lemma}
\begin{proof}
The basic idea is that we can rewrite our polynomial $p(x)$ as
$\sum_{i=1}^n \lambda_i (v_i\cdot x)^2$,
where $\sum_i |\lambda_i| \leq \|A\|_\ast \leq 1$.
It thus, suffices to verify that $\E[(v\cdot (G-\mu_T))^2]$ is not too large for any $v_i$.
This can be achieved by applying our basic multifilter to a collection of degree-$1$ polynomials.

The algorithm is as follows:

\medskip

\fbox{\parbox{6.1in}{
{\bf Algorithm} {\tt Degree2Homogeneous}\\
Input: a polynomial $x^T A x$, where $A \in \R^{n \times n}$ with $\|A\|_\ast \leq 1$, 
$T \subset \R^n, d \in \Z_+, 0< \alpha <1 , \tau >0$.\\
\vspace{-0.5cm}
\begin{enumerate}
\item Compute (approximations for) the eigenvalues $\lambda_i$ and unit eigenvectors $v_i$  of $A$, $1 \leq i \leq n$.
\vspace{-0.2cm}
\item For each $i \in [n]$, run {\tt BasicMultifilter} using the polynomial $v_i^T x$.
If any of these output ``NO'' or a list $(T_i, \alpha_i)$,
then stop at that iteration and return that output.
\vspace{-0.2cm}
\item Otherwise, if all output ``YES'', then output ``YES''.
\end{enumerate}
\vspace{-0.2cm}
}}

\bigskip

For correctness, note that since the $v_i$'s are unit vectors, we have that $\Var[v_i^T G] = 1$,
and therefore the conditions of Proposition~\ref{prop:basic-mf} hold.
Hence, if {\tt BasicMultifilter} returns ``NO'', then $T$ is not $\alpha$-good;
if it returns a list of $(T_i, \alpha_i)$, then the list satisfies
the multifilter condition for $(T, \alpha)$.
So, the algorithm is correct if any multifilter produces this output.

It remains to show correctness when the algorithm output ``YES''.
In this case, using the guarantees of the ``YES'' case in {\tt BasicMultifilter},
we have that
$$\E[(v_i^T (G-\mu_T))^2] = \Var[v_i^T (G-\mu_T)] + (v_i^T (\mu-\mu_T))^2 \leq O(\sqrt{1+\log(1/\alpha)} \new{\log(2+ \log(1/\alpha))^2}) \;.$$
Since $p(x)=x^T A x = \sum_i \lambda_i (v_i^T x)^2$ and $\sum_i |\lambda_i| = \|A\|_\ast \leq 1$, we have
\begin{align*}
\E[p(G-\mu_T)] & =\sum_i \lambda_i \E[(v_i^T (G-\mu_T))^2]  \\
& \leq \left(\sum_i |\lambda_i| \right) \cdot O(\sqrt{1+\log(1/\alpha)} \new{\log(2+ \log(1/\alpha))^2}) \\
& \leq O((1+\log(1/\alpha)) \new{\log(2+ \log(1/\alpha))^2})  \;,
\end{align*}
as desired.
\end{proof}

\subsection{Multifilter Routine for Symmetric Multilinear Polynomials} \label{ssec:sym-ml}

Recall that a degree-$d$ multilinear polynomial is 
one that is linear as a function of each coordinate and 
can be expressed as $A(x_1,\dots, x_d)$, for an order-$d$ tensor $A$.

\new{A major part of our algorithm is to design a multifilter for degree-$d$ multilinear polynomials. 
In particular, given such a polynomial $p$, we want to either be able to verify that $p(G-\mu_T, \ldots,G-\mu_T)$ 
has similar expectation to $p(T-\mu_T,\ldots,T-\mu_T)=0$ (in both cases the coordinates are independent), 
or to find a filter that allows us to throw away some of the bad points of $T$. In particular, we show:}

\begin{lemma} \label{lem:sym-ml}
There exists an algorithm that, given a degree-$d$ multilinear polynomial
$A(x_1,\dots, x_d)$ over $\R^{nd}$ with $\|A\|_2=1$, $(T, \alpha)$, and a confidence probability $\tau$,
it outputs one of ``YES'', ``NO'', or a list of $(T_i, \alpha_i)$, and has the following performance:
\begin{enumerate}
\item If the algorithm returns ``YES'', then the following holds:
If $T$ is $\alpha$-good, then with probability at least $1-\tau$,
we have that $$|\E[A(G_1-\mu_T,\dots,G_d-\mu_T)]| = O((1+\log(1/\alpha)) \new{\log (2+
\log(1/\alpha))^2})^{d/2} \;,$$
where $G_i$ are independent copies of $G$. 
\item If the algorithm returns ``NO'', then $T$ is not $\alpha$-good.
\item If the algorithm returns a list of $(T_i,\alpha_i)$, then $(T_i,\alpha_i)$ satisfy
 the multifilter condition for $(T,\alpha)$.
\end{enumerate}
The algorithm runs in time $O((\ln(1/\tau)+d)/\alpha)^d \cdot \poly(|T|,n^d)$.
\end{lemma}

\new{The basic idea of the proof is as follows:
We note that if $p(G-\mu_T,G-\mu_T,\ldots,G-\mu_T)$ has mean significantly 
far from $0$, then when we evaluate $p(T-\mu_T,T-\mu_T,\ldots, T-\mu_T)$ 
there will be approximately an $\alpha^d$ probability that all of our samples 
come from $G$, and thus produce a value far from the mean. 
One might naively try to look at $d$ inputs to $G$ at a time 
and perform some clustering based on the values. Unfortunately, this approach 
has several problems including that some $d$-tuples will mix good samples with bad ones.

Instead, we consider assigning all but one of the coordinates random values from $T$ 
and considering the \emph{linear} function of the last coordinate. With decent probability,
the first $d-1$ coordinates will come from $G$ and the average over the samples from $S$ 
will be approximately $p(\mu-\mu_T,\mu-\mu_T,\ldots,\mu-\mu_T)$, 
while the average where the last sample comes from $T$ will be approximately $0$ 
(since the mean of $T$ is $0$, any linear function of $T$ also has mean $0$). 
This should allow us to do clustering \emph{unless} 
$\Var_X[p(X_1-\mu_T,X_2-\mu_T,\ldots,X_{d-1}-\mu_T,X-\mu_T)]$ is too large.

However, if fixing one of the inputs to $p$ causes the variance to increase significantly, 
this should also produce a detectable discrepancy that we can use to filter. 
In particular, the variance of a multilinear polynomial $p$ after setting 
its first coordinate to $X$ is a degree-$2$, trace-norm $1$ polynomial in $X$. 
Therefore, using our multifilter for degree-$2$ polynomials, 
unless it is the case that setting the first variable of $p$ to a random $X$ from $T$ 
does not increase the variance by too much, we can create a filter.

In the end, this gives us a recursive algorithm. We fix the variables of $p$ 
one at a time in many different ways. At each step, we verify that most 
of the ways to fix the next variable does not increase the variance by too much 
(or set the mean to something too large in the last step). If any of these fails to hold, 
we will be able to get a multifilter. If however, none of these settings produce 
any detectable problem, then it is likely the case that many times we saw $d$
 samples from $G$ as inputs, yielding a value of $p$ that is not too large. 
 This will only happen if the mean of $p$ at $d$ copies of $G$ is not too large, which is our ``YES'' case.}

\begin{proof}
The pseudocode of our multifilter for multilinear polynomials
is given below.

\medskip

\fbox{\parbox{6.1in}{
{\bf Algorithm} {\tt MultilinearMultifilter}\\
Input: a degree-$d$ multilinear polynomial
$A(x_1,\dots, x_d)$ over $\R^{nd}$ with $\|A\|_2=1$, $T \subset \R^n, d, \alpha , \tau$\\
\vspace{-0.5cm}

\begin{enumerate}
\item If $d=1$, run {\tt BasicMultifilter} and return its output.
\vspace{-0.2cm}
\item Compute the quadratic polynomial $q(x) = \|A (x-\mu_T) \|_2^2, x \in \R^n$,
where $A x$ is an order $(d-1)$ tensor with
$(A x)_{i_2, \dots, i_d}= \sum_{i_1} x_{i_1} A_{i_1, \dots, i_d}$.
\vspace{-0.2cm}
\item Run algorithm {\tt Degree2Multifter} on $q$. If it returns ``NO'' or a list $(T_i,\alpha_i)$,
then return the same answer.
\vspace{-0.2cm}
\item Let $U$ be a multiset of $m= C \alpha^{-1}\ln(2d/\tau)$ samples chosen uniformly at random
from $T$, where $C=200$.
\vspace{-0.2cm}
\item For each $x\in U$:
\vspace{-0.2cm}
\begin{enumerate}
\item Let $A_x$ be the order $(d-1)$ tensor given by $(1/\sqrt{q(x)}) A (x-\mu_T)$.
\vspace{-0.2cm}
\item Run  {\tt MultilinearMultifilter}$(A_x,T,d-1,\alpha,\tau/2)$.
If it returns ``NO'' or a list $(T_i,\alpha_i)$, then return the same output.
\end{enumerate}
\vspace{-0.2cm}
\item \new{Otherwise,} return ``YES''.
\end{enumerate}
\vspace{-0.2cm}
}}

\bigskip

\new{Firstly, we bound the running time. 
Note that we make $C \alpha^{-1}\ln(2d/\tau)$ calls to the degree-$(d-1)$ {\tt MultilinearMultifilter}, 
each of which has precision $\tau/2$. Inductively, in the recursive calls 
to the degree-$(d-i)$ {\tt MultilinearMultifilter}, for $d-i \geq 2$, 
we make $C \alpha^{-1}\ln(2^{i+1}(d-i)/\tau)$ calls to degree-$(d-i+1)$ {\tt MultilinearMultifilter}. 
Thus, the total number of calls to {\tt MultilinearMultifilter} is $O((\ln(1/\tau)+d)/\alpha)^d$.}

Next, we show that if the algorithm returns ``NO'' or a list $(T_i, \alpha_i)$,
then it is correct. For this, we just need to show that
all the calls to the subroutines satisfy their conditions. In the $d=1$ case,
$A$ is a vector, hence $\Var[A(G)]=\|A\|_2=1$, and therefore
we satisfy the conditions of {\tt BasicMultifilter} and the algorithm returns correctly.

We need to show that $q(x)=x^T B x$, where $B$ is symmetric and positive semi-definite, and $\Tr(B) \leq 1$.
We have that
\begin{align*}
q(x) & = \sum_{i_2, \dots, i_d}  \left(\sum_{i_1} x_{i_1} A_{i_1, \dots, i_d}\right)^2 \\
& = \sum_{i_1,i'_1,i_2, \dots, i_d} x_{i_1} x_{i'_1} A_{i_1, \dots, i_d} A_{i'_1, \dots, i_d} \\
& = \sum_{i_1,i'_1} x_{i_1} x_{i'_1} \sum_{i_2, \dots, i_d}  A_{i_1, \dots, i_d} A_{i'_1, i_2 \dots, i_d}\\
&= x^T B x \;,
\end{align*}
where $B_{ij}= \sum_{i_2, \dots, i_d}  A_{i, i_2, \dots, i_d} A_{j, i_2, \dots, i_d}$.
But note that $B$ is symmetric and
$$\Tr{(B)}=\sum_i B_{ii}= \sum_{i_1} \sum_{i_2,\dots, i_d} A^2_{i_1,\dots, i_d} = \|A\|_2^2=1 \;.$$
Thus, we satisfy the pre-conditions of  {\tt Degree2Multifilter} and the algorithm returns correctly.

Since $A_x =(1/\sqrt{q(x)}) A (x-\mu_T)= (1/\|Ax\|_2) A(x-\mu_T)$, we have $\|A_x\|_2=1$ for all $x \in U$.
Thus, we recursively satisfy the conditions of calls to {\tt MultilinearMultifilter} with degrees lower than $d$.
By a simple induction, they are correct if they return ``NO'' or a list $(T_i, \alpha_i)$.

Now we consider the case when the algorithm outputs ``YES''.
We will show correctness by induction on the degree.
For $d=1$, since we satisfy the conditions of {\tt BasicMultifilter},
we have that 
$$|\E[A(G-\mu_T)]| \leq O(\sqrt{1+\log(1/\alpha)} \new{\log \log(1/\alpha)}) \;.$$
We assume inductively that for $d'=d-1$ and all order $d'$ tensors $A'$ with $\|A'\|_2=1$
the following holds: If we apply the algorithm to $A'$ with failure probability $\tau'$
and it returns ``YES'', and if $T$ is $\alpha$-good, then
$$|\E[A'(G_1-\mu_T,\dots,G_{d'}-\mu_T)]| \leq f(d', \alpha)$$
with probability at least $1-\tau'$,
where $f$ is some function to be decided later.

Let $p(x)=\E[A(x-\mu_T, G_2-\mu_T, \dots,G_d-\mu_T)]$.
Note that this is a linear function in $x$.
We show the following claim:

\begin{claim}
If $T$ is $\alpha$-good and $C \geq 200$, then with probability at least $1-\tau/2d^2$,
there is an $x \in U$ with  $|p(x)| \geq |\E[p(G_1)]|$ and $q(x-\mu_T) \leq 200 \E[q(G-\mu_T)]$ .
\end{claim}
\begin{proof}
Depending on if $\E[p(G_1)]$ is positive or negative,
it suffices to show either that there exists an $x$
with $p(x) \geq \E[p(G_1)]$ or that $p(x) \leq -\E[p(G_1)]$.
Since $p(x)$ is a linear function, $v\cdot x$ for some $v$, we actually need an $x$
that satisfies a particular one of the inequalities $v\cdot x \geq v\cdot \mu$ or $v\cdot x \leq v\cdot \mu$.
We henceforth assume that we need $v\cdot x \geq v\cdot \mu$, the other case being similar.

Note that $\Pr[v\cdot G \geq v\cdot \mu]=1/2$.
Since $S$ is representative and this is a linear inequality,
we have that $|\Pr_{X \in_u S}[X \geq \mu]-1/2|  \leq 1/100$, and therefore
at least  $49|S|/100$ samples in $S$ satisfy this.
By Markov's inequality, we have that
$\Pr_{X \sim G}[q(X-\mu_T) \leq 200 \E[q(G-\mu_T)]] \geq 199/200$.
Since $S$ is representative and this is a quadratic inequality, it follows that
$\Pr_{X  \new{\in_u S}}[q(X-\mu_T) \leq 200 \E[q(G-\mu_T)]] \geq 197/200$.
By a union bound, both of these events hold for at least $95|S|/200$ samples in $|S|$.

Since $T$ is $\alpha$-good, it contains at least an
$1-\alpha/\new{6}+\exp(-\new{3}n) \geq 53/100$ fraction of $S$.
Thus, $S \cap T$ contains at least $|S|/200$ samples satisfying
the condition in the claim statement.
Again, since $T$ is $\alpha$ good, $|S \cap T| \geq \alpha|T|$,
and so at least an $\alpha/200$-fraction of $T$ satisfies this condition.
The probability that $U$ contains such an $x$ is at least
$$(1-\alpha/200)^m=(1-\alpha/200)^{C \ln(2d/\tau)/\alpha} \leq \tau/2d^2 \;,$$
when $C \geq 200$.
\end{proof}

By our inductive assumption and a union bound,
except with probability $1-\tau/2d^2$,
there is such an $x \in U$ \new{satisfying the above claim},
if the recursive call on that $x$ satisfies
$$|\E[A_x(G_2-\mu_T,\dots,G_d-\mu_T)]| \leq f(d-1, \alpha) \;.$$
Noting that $p(x)/\sqrt{q(x-\mu_T)} = \E[A_x(G_2-\mu_T, \dots,G_d-\mu_T)]$,
we have that $p(x) \leq f(d-1, \alpha) \sqrt{q(x-\mu_T)}$.
Since {\tt Degree2Multifter} applied to $q(x)$ returned ``YES'',
\new{by the above claim,}
we have that $q(x-\mu_T) \leq 200 \E[q(G-\mu_T)]=O((1+\log(1/\alpha)) \new{\log(2+ \log(1/\alpha))^2})$.
Thus, we have
$$
|\E[A(G_1-\mu_T,\dots ,G_d-\mu_T)|=|\E[p(G_1)]| \leq |p(x)| \leq f(d-1, \alpha) 
\cdot O(\sqrt{1+\log(1/\alpha)} \new{\log(2+ \log(1/\alpha))}) \;.$$
Our base case was $f(1,\alpha)= O(\sqrt{1+\log(1/\alpha)} \new{\log(2+ \log(1/\alpha))})$,
and we have shown that
$$f(d,\alpha)=f(d-1,\alpha) \cdot O(\sqrt{1+\log(1/\alpha)} \new{\log(2+ \log(1/\alpha))}) \;.$$
We can thus take $f(d,\alpha)= O(\sqrt{1+\log(1/\alpha)} \new{\log(2+ \log(1/\alpha))})^{d}$.

By induction, when $T$ is $\alpha$-good and the algorithm returns ``YES'' unless one of these $d$ recursive calls (one for each degree) fails,
$|\E[A(G_1-\mu_T,\dots,G_d-\mu_T)]| \leq O((1+\log(1/\alpha)) \new{\log (2+
\log(1/\alpha))^2})^{d/2}$. The probability of any of these events failing by the union bound is at most $\tau(1/2+1/8+\ldots+1/2d^2) \leq \tau$.
This completes the proof.
\end{proof}

\subsection{Multifilter Routine for Harmonic Polynomials} \label{ssec:harm}

In this subsection, we build on the results of the previous subsections
to handle all degree-$d$ Harmonic polynomials. Specifically, we prove:

\begin{lemma} \label{lem:harm}
There exists an algorithm, given
a degree-$d$ \new{harmonic }polynomial $h_A(x)$
for a symmetric order-$d$ tensor $A$ with $\|A\|_2=1$,
\new{$(T, \alpha)$, and a confidence probability $\tau$},
\new{it outputs ``YES'', ``NO'', or a list of $(T_i, \alpha_i)$}, and
with probability at least $1-\tau$ satisfies:
\begin{enumerate}
\item If it returns ``YES'', the following holds:
\new{If $T$ is $\alpha$-good, then}
$$\E[h_A(T-\mu_T)^2]= O\left((d+\log(1/\alpha)) \new{\log (2+\log(1/\alpha))^2}\right)^{2d} \;.$$
\item If it returns ``NO'', then $T$ is not $\alpha$-good.
\item If it returns a list of $(T_i,\alpha_i)$, then $(T_i,\alpha_i)$ satisfy
the multifilter condition for $(T,\alpha)$.
\end{enumerate}
When the algorithm outputs a list $(T_i,\alpha_i)$,
we always have that $\sum_i 1/\alpha_i^2 \leq 1/\alpha^2$. 
The algorithm runs in time $O((\ln(1/\tau)+d)/\alpha)^d \cdot \poly(|T|,n^d)$.
\end{lemma}
\new{The algorithm here is now relatively straightforward. 
If we knew that the variance of $p(G-\mu_T)$ were small, 
we could directly apply our basic multifilter. The problem is that if the mean of $G$ 
is far from $\mu_T$ this might not be the case. However, it is the case that $\Var[p(G-\mu_T)]$ 
is an average of $q_i(\E[G-\mu_T])^2$ for some normalized multilinear polynomials $q_i$. Thus, 
it suffices to show that $|q_i(\E[G-\mu_T])|$ is small for each $i$. 
This in turn equals $r_i(G-\mu_T, \ldots, G-\mu_T)$ for some normalized multilinear polynomial $r_i$. 
Using our filter for multilinear polynomials, we can verify that $r_i(G-\mu_T,\ldots,G-\mu_T)$ 
has small expected value (or find a filter, in which case we are already done), 
which will imply that $p(G-\mu_T)$ has low variance.}
\begin{proof}
The algorithm is the following:
\medskip

\fbox{\parbox{6.1in}{
{\bf Algorithm} {\tt HarmonicMultifilter}\\
Input: a degree-$d$ \new{harmonic }polynomial $h_A(x)$
for a symmetric tensor $A$ with $\|A\|_2=1$, $T, \alpha, \tau$.\\
\vspace{-0.5cm}

\begin{enumerate}
\item For $d'=0$ to $d$:
\vspace{-0.2cm}
\begin{enumerate}
\item Let $B^{(d')}$ by the order-$2d'$ tensor with
$$B^{(d')}_{i_1,\dots, i_{d'}, j_1, \dots ,j_{d'}}=
\sum_{k_{d'+1}, \dots, k_d} A_{i_1,\dots, i_{d'},k_{d'+1}, \dots k_d} A_{j_1, \dots, j_{d'}, k_{d'+1}, \dots, k_d} \;.$$
\item We can consider $B^{(d')}$ as a $n^{d'} \otimes n^{d'}$ symmetric matrix
by grouping each of the $i_1,\dots, i_{d'}$ and $j_1, \dots ,j_{d'}$ coordinates together.
By computing an eigendecomposition of this matrix, obtain a decomposition
$$B^{(d')}= \sum_i \lambda_i V_i \otimes V_i \;,$$
where the $V_i$ are a symmetric order $d'$ tensors with $\|V_i\|_2=1$ and $\lambda_i\neq 0$.
\item For each $V_i$, run {\tt MultilinearMultifilter}$(T, V_i, d',\alpha,\tau/(d n^d))$.
If any return ``NO'' or a list of $(T_i, \alpha_i)$, then return the same output.
\end{enumerate}
\vspace{-0.2cm}
\item Run {\tt BasicMultifilter} on $h_A(x-\mu_T)/\beta$,
where $\beta=(C (1+\log(1/\alpha))\log(2+\log(1/\alpha))^2)^{d/2}$ for a sufficiently large constant $C$.
If it returns ``NO'' or a list of $(T_i, \alpha_i)$, then return the same output.
\vspace{-0.2cm}
\item Return ``YES''.
\end{enumerate}
\vspace{-0.2cm}
}}

\bigskip

We first note that, as claimed, $V_i$ is symmetric:
\begin{claim}
The eigenvectors of the matrix corresponding to $B^{(d')}$ with non-zero eigenvalues correspond to symmetric tensors.
\end{claim}
\begin{proof}
If $V$ is such a tensor corresponding to eigenvalue $\lambda$, then
$$\lambda V_{i_1,\dots, i_{d'}} = \sum_{j_1, \dots ,j_{d'}} B^{(d')}_{i_1,\dots, i_{d'}, j_1, \dots ,j_{d'}}  V_{j_1,\dots, j_{d'}} \;.$$
However, since $A$ is symmetric under permutations of its first $d'$ coordinates,
$B^{(d')}$ is also symmetric in its first $d'$ coordinates, and so $\lambda V$ is symmetric. When $\lambda \neq 0$, $V$ is also symmetric.
\end{proof}

Since $\|V\|_2=1$, the preconditions of {\tt MultilinearMultifilter} are satisfied,
and if it returns something other than ``YES'', the algorithm returns correctly.

There are at most ${d n^d}$ calls to {\tt MultilinearMultifilter} in the loop and one outside.
By a union bound, all calls to {\tt MultilinearMultifilter} output ``YES'' correctly
with probability at least $1-\tau$.
We assume this holds.

\begin{lemma}
If all calls to {\tt MultilinearMultifilter} return ``YES'' correctly and
if $T$ is $\alpha$-good, then we have
$\E[h_A(G-\mu_T)^2]=O((1+\log(1/\alpha)) \new{\log (2+\log(1/\alpha))^2})^{d} $.
\end{lemma}
\begin{proof}
By Lemma \ref{lem:poly-relations}, we have
$$\E[h_A(G -\mu_T)^2]= \sum_{d'=0}^d \hmg_{B^{(d')}}(\mu-\mu_T) {d \choose d'} /(d'-d)!  \;.$$
We can further use the decomposition of $B^{(d')}$ to obtain that
\new{when all calls to {\tt MultilinearMultifilter} correctly return ``YES''}:
\begin{align*}
\hmg_{B^{(d')}}(\mu-\mu_T)  & = B^{(d')}(\mu-\mu_T,\dots,\mu-\mu_T) \\
& = \sum_i \lambda_i (V_i \otimes V_i)(\mu-\mu_T,\dots,\mu-\mu_T) \\
& = \sum_i \lambda_i (V_i(\mu-\mu_T,\dots,\mu-\mu_T) )^2\\
& = \sum_i \lambda_i  \E[V_i(G_{(1)}-\mu_T,\dots, G_{(d')}-\mu_T)]^2 \\
& \leq \left(\sum_i |\lambda_i|\right) \cdot O((1+\log(1/\alpha)) \new{\log (2+
\log(1/\alpha))^2})^{d'} \;,
\end{align*}
where $G_{(1)} ,\dots, G_{(d')}$ are independent copies of $G$.

Now we need to bound $(\sum_i |\lambda_i|)$.
Let $B'$ be the matrix that, as in the algorithm, is obtained from $B^{(d')}$
by treating the first and last $d'$ coordinates of $B^{(d')}$ as one coordinate.
Let $A'$ be the matrix obtained from $A$ by similarly treating
the first $d'$ and last $d-d'$ coordinates as a single coordinate.
Then, $B'=A'^T A'$. Thus, $B'$ is positive semidefinite and $\lambda_i \geq 0$.
Thus, we have
$$\sum_i |\lambda_i| = \sum_i \lambda_i = \Tr(B') = \Tr(A'^T A') = \|A'\|_F^2 = \|A\|_2^2 = 1\;.$$
Now, substituting these bounds, we have
$$\E[h_A(G -\mu_T)^2]= \sum_{d'=0}^d  {d \choose d'} /(d'-d)! \cdot O((1+\log(1/\alpha)) \new{\log (2+\log(1/\alpha))^2})^{d'} \;.$$
Since $\sum_{d'=0}^d 1/(d'-d)! \leq e$ and
${d \choose d'} \leq 2^d$, we have
$$\sum_{d'=0}^d  {d \choose d'} /(d'-d)! \cdot O((1+\log(1/\alpha)) \new{\log (2+\log(1/\alpha))^2})^{d'} 
\leq O((1+\log(1/\alpha)) \new{\log (2+\log(1/\alpha))^2})^{d} \;,$$
as required.
\end{proof}

For $C$ sufficiently large,
if $T$ is $\alpha$-good, then $\Var[h_A(G-\mu_T)] \leq \beta^{\new{2}}$
with probability at least $1-\tau$. In this case, $h_A(x-\mu_T)/\beta$
satisfies the pre-conditions of {\tt BasicMultifilter},
and so if it returns ``NO'' or a list, the algorithm is correct.
If it returns ``YES'', then
$\Var[h_A(T-\mu_T)/\beta]=O(d+\log(1/\alpha))^d \cdot \log(2+\log(1/\alpha))^2$,
and so
$$\Var[h_A(T-\mu_T)] \leq \beta^{\new{2}} \cdot O(d+\log(1/\alpha))^d \cdot \log(2+\log(1/\alpha))^2=\new{ O((d+\log(1/\alpha))\log(2+\log(1/\alpha))^2)^{2d}} \;.$$

Additionally, if $T$ is $\alpha$-good, then
$$|\E[h_A(G-\mu_T)]-\E[h_A(T-\mu_T)]| \leq \new{\beta} \cdot O(d+\log(1/\alpha))^{d/2} \cdot \log(2+\log(1/\alpha))^2 \;.$$
On the other hand $|\E[h_A(G-\mu_T)]| \leq \E[h_A^2(G-\mu_T)]^{1/2} \leq O((1+\log(1/\alpha)) \new{\log (2+\log(1/\alpha))^2})^{d/2}.$ Thus, by the triangle inequality,
$$
|\E[h_A(T-\mu_T)]| \leq O((1+\log(1/\alpha)) \new{\log (2+\log(1/\alpha))^2})^{d}.
$$

Combining with the above, we find that
$$
\E[h_A^2(T-\mu_T)] \leq \E[h_A(T-\mu_T)]^2 + \Var(h_A(T-\mu_T)) = O((d+\log(1/\alpha)) \new{\log (2+\log(1/\alpha))^2})^{2d}.
$$
Thus, if our algorithm returns we have $\E[h_A^2(T-\mu_T)] = O((d+\log(1/\alpha)) \new{\log (2+\log(1/\alpha))^2})^{2d}.$ This completes our proof.
\end{proof}

\subsection{Putting Everything Together: Proof of Proposition~\ref{mainSubroutineProp}} \label{ssec:combining}

We are now ready to prove our main proposition. 
\new{The basic idea is to use a linear algebraic computation
(computing the largest eigenvalue of an appropriate matrix) 
to check if there are any degree-$d$ harmonic polynomials 
$p$ with $\E[p(T)^2]$ too large. If so, we can use our multifilter 
for harmonic polynomials to produce a filter. If there is no such polynomial, 
this will imply that the mean of $G$ is close to the empirical mean of $T$, 
as otherwise $He_d(v\cdot(x-\mu_T))$, 
where $v$ points in the direction of $\mu-\mu_T$, 
would have have too large a contribution coming from points 
where $x$ is taken from $G$.}

\begin{proof}[Proof of Proposition~\ref{mainSubroutineProp}]
The algorithm establishing the proposition is presented in pseudocode below:

\medskip

\fbox{\parbox{6.1in}{
{\bf Algorithm} {\tt MainSubroutine}\\
Input: multiset $T \subset \R^n , d \in \Z_+, 0< \alpha < 1, \tau>0$.\\
\vspace{-0.5cm}

\begin{enumerate}
\item Let $\mathbf{P}_{n, d}(x)$ be the vector of polynomials $\prod_i \He_{a_i}(x_i)/\sqrt{a_i!}$,
which form a basis for all degree-$d$ $n$-variable harmonic polynomials, i.e.,
with a coordinate for each $a \in \Z^n$ with $a_i \geq 0$ and $\sum_i a_i=d$.
\vspace{-0.2cm}
\item Compute  the top eigenvalue $\lambda$ and associated unit eigenvector $v^\ast$
of the $O(n^{d}) \times O(n^{d})$ matrix $\Sigma=\E\left[\mathbf{P}_{n, d}(T-\mu_T) \mathbf{P}_{n, d}(T-\mu_T)^T\right]$.
\vspace{-0.2cm}
\item Let $h_A(x)=v^{\ast T} \mathbf{P}_{n, d}(x)$, i.e., $A_i=v^{\ast}_{c(i)}/\sqrt{{d \choose c_1(i),\dots,c_n(i)}}$,
\new{where $i = (i_1, \ldots, i_d)$.}
\vspace{-0.2cm}
\item Run {\tt HarmonicMultifilter} on $h_A(x)$ with failure probability $\tau$.
If it returns ``NO'' or a list of $(\alpha_i,T_i)$ then stop and return this output.
\vspace{-0.2cm}
\item Output $\mu_T$.
\end{enumerate}
\vspace{-0.2cm}
}}

\bigskip

We first establish the claims in the algorithm.
By Fact \ref{fact:harmonic}, all degree-$d$ Harmonic polynomials
can be expressed as $h_A(x)$ and as a linear combination of $\prod_i \He_{a_i}(x_i)/\sqrt{a_i!}$.
So, $\mathbf{P}_{n, d}$ is a basis.

\begin{lemma}
For $A_i=v^{\ast}_{c(i)}/\sqrt{{d \choose c_1(i),\dots,c_n(i)}}$,
we have $h_A(x)=v^{\ast T} \mathbf{P}_{n, d}(x)$ and $\|A\|_2=\|v^{\ast}\|_2=1$.
\end{lemma}
\begin{proof}
Firstly, we note that for a given $a \in \Z^n$ with $a_j \geq 0$ and $\sum_j a_j=d$,
there are ${d \choose a_1,\dots,a_n}$ ways such that $c_j(i)=a_j$ for all $j$,
i.e., that $x_j$ appears $a_j$ times in $i=(i_1,\dots i_d)$.
Thus, the  $\prod_i \He_{a_i}(x_i)$ term in
$$\sqrt{d!} \cdot h_A(x) =  \sum_i A_i \prod_j \He_{c_j(i)}(x_j)$$
is $\sum_{i:c(i)=a} A_i=A_{c^{-1}(i)} {d \choose c_1(i),\dots,c_n(i)}$.
Therefore, the $\prod_i \He_{a_i}(x_i)/\sqrt{a_i!}$ term in
$h_A(x)$ is
$$A_{c^{-1}(i)} {d \choose a_1,\dots,a_n} \cdot \prod{\sqrt{a_i!}}/\sqrt{d!}
=A_{c^{-1}(a)} \sqrt{{d \choose a_1,\dots,a_n}}=v^{\ast}_{a} \;.$$
Now we have $h_A(x)=v^{\ast T} \mathbf{P}_{n, d}(x)$.

Next note that
$$\|A\|_2^2 = \sum_i A_i^2 = \sum_a {d \choose a_1,\dots,a_n} A_{c^{-1}(a)}^2 =\sum_a v^{\ast 2}=\|v^{\ast}\|_2^2 \;,$$
as desired.
 \end{proof}

Since $\|A\|_2=1$, we satisfy the preconditions of {\tt HarmonicMultifilter}.
By construction, the routine {\tt HarmonicMultifilter} returns correctly with probability at least $1-\tau$.
We assume that this holds.
In this case, if it returns ``NO'' or a list of $(T_i,\alpha_i)$,
then they are \new{the correct output}.

The following lemma establishes that $A$ has the largest eigenvalue:

\begin{lemma} \label{lem:harmL2}
Any degree-$d$ Harmonic polynomial $h_B(x)$ has
$\E[h_B(T-\mu_T)^2] \leq \lambda \|B\|_2^2$, with equality if $B=A$.
If  {\tt HarmonicMultifilter} returns ``YES'' correctly,
then $\lambda = O\left((d+\log(1/\alpha)) \new{\log (2+\log(1/\alpha))^2}\right)^{2d}$.
\end{lemma}
\begin{proof}
We have $h_B(x)=v^T p(x)$,
where $v_a= \sum_{i:c(i)=a} B_i$.
Thus, we obtain
$$\E[h_B(T-\mu_T)^2] = \E[(v^T p(T-\mu_T))^2]=v^T \Sigma v
\leq \lambda \|v\|_2^2=\lambda \|B\|_2^2 \;.$$
In the case $B=A$, $v=v^{\ast}$ has eigenvalue $\lambda$,
and so  $\E[h_A(T-\mu_T)^2] = \lambda \|A\|_2^2 = \lambda$.
When the routine {\tt HarmonicMultifilter} returns ``YES'' correctly,
Lemma~\ref{lem:harm} ensures that  
$$\E[h_A(T-\mu_T)^2]= O\left((d+\log(1/\alpha)) \new{\log (2+\log(1/\alpha))^2}\right)^{2d} \;.$$
\end{proof}

Consider the polynomial $\He_d(v^T(x -\mu_T))$, $x \in \R^n$, for a unit vector $v$.
Since for $X \sim N(\mu_T, I)$, we have that
$\E[\He_d(v^T(X -\mu_T))^2]=d!$,
for the $B$ with $\He_d(v^T(x -\mu_T))=h_B(x-\mu_T)$, we have that
$\|B\|_2^2 = \E[h_B(X-\mu_T)^2]=d!$.
Thus, by Lemma~\ref{lem:harmL2}, it follows that
$\E[\He_d(v^T(T -\mu_T))^2] \leq  \lambda \cdot d!$.
If $T$ is $\alpha$-good, then $T \cap S$ is at least an $\alpha$-fraction of the points in $T$, and so:
\begin{equation} \label{eq:ST}
\E[\He_d(v^T(S \cap T -\mu_T))^2] \leq \lambda \cdot d!/\alpha \;.
\end{equation}

Now we consider the distribution of this polynomial under $G$:

\begin{lemma}
The polynomial $\He_d{(v^T(G -\mu_T))}$ has mean $(v^T(\mu-\mu_T))^d$
and variance at most $2 \max\{d, v^T(\mu-\mu_T)\}^{2(d-1)}$.
\end{lemma}
\begin{proof}
Using the Taylor series expansion (\ref{eq:taylor-hermite}), we obtain that
$$\He_d(v^T(G -\mu_T)) = \He_d(v^T(G -\mu) + v^T(\mu -\mu_T))=
\sum_{k=0}^d {d \choose k} (v^T (\mu-\mu_T))^{d-k} \He_{k}(v^T(G -\mu)) \;.$$
For the mean, all but the $\He_0$ term have expectation zero,
and so $\E[\He_d(v^T(G -\mu_T))] = (v^T (\mu-\mu_T))^d$.
For $\E[\He_d(v^T(G -\mu_T))^2]$,
we see that different terms are orthogonal, and so we get no cross terms:
$$\E[\He_d(v^T(G -\mu_T))^2] = \sum_{k=0}^d {d \choose k}^2 (v^T (\mu-\mu_T))^{2(d-k)} \cdot k! \;.$$
Note the $k=0$ term is $(v^T (\mu-\mu_T))^{2d}$ and
that the ratio of the $k+1$ term to the $k$ term, for $k \geq 0$, is
$$(d-k)^2 /((k+1)(v^T (\mu-\mu_T))^2) \leq d^2/2(v^T (\mu-\mu_T))^2 \;.$$
Finally, we have that:
\begin{align*}
\Var[\He_d(v^T(G -\mu_T))]
&= \E[\He_d(v^T(G -\mu_T))^2] - \E[\He_d(v^T(G -\mu_T))]^2 \\
&= \sum_{k=1}^d {d \choose k}^2 (v^T (\mu-\mu_T))^{2(d-k)} \cdot k! \\
& \leq \sum_{k=1}^d {d \choose k}^2 \max \{d,v^T (\mu-\mu_T)\}^{2(d-k)} \cdot k! \\
& \leq 2\max \{d, v^T (\mu-\mu_T)\}^{2(d-1)} \;.
\end{align*}
\end{proof}

Thus, by Cantelli's inequality we get that
$$\Pr\left[\He_d{(v^T(G -\mu_T))} \geq (v^T(\mu-\mu_T))^d - \sqrt{2} \max\{d,(v^T(\mu-\mu_T))\}^{(d-1)}\right] \geq 1/2 \;.$$
Since $S$ is \new{representative}, we have that
$$\Pr\left[\He_d{(v^T(S -\mu_T))} \geq (v^T(\mu-\mu_T))^d - \sqrt{2} \max\{d,(v^T(\mu-\mu_T))\}^{(d-1)}\right] \geq \new{49}/100 \;.$$
Now if $T$ is $\alpha$-good, then $S \new{\cap T}$ contains at least
a $1-\alpha/\new{6} - \exp(-\new{3}n) \geq \new{76}/100$
fraction of the samples in $S$. Thus, we have that
$$\Pr\left[\He_d{(v^T(S\cap T -\mu_T))} \geq (v^T(\mu-\mu_T))^d - \sqrt{2} \max\{d,(v^T(\mu-\mu_T))\}^{(d-1)}\right] \geq 1/4 \;.$$
However, by Markov's inequality applied to (\ref{eq:ST}), we get
$$\Pr\left[\He_d(v^T(S \cap T -\mu_T)) \geq \sqrt{4\lambda d!/\alpha} \right] \leq 1/4 \;,$$
and so
$$(v^T(\mu-\mu_T))^d - \sqrt{2} \max\{d, (v^T(\mu-\mu_T))\}^{d-1} \leq  \sqrt{4\lambda d!/\alpha} \;.$$
Note that $(v^T(\mu-\mu_T))^d \leq  \sqrt{2} \max\{d, v^T(\mu-\mu_T)\}^{d-1}$
only when $v^T(\mu-\mu_T) \leq 2 d$,
and so we have that:
$$v^T(\mu-\mu_T) \leq \max \{ 2d, \sqrt[2d]{4\lambda d!/\alpha} \} \leq
O((d+\log(1/\alpha)) \new{\log (2+\log(1/\alpha))^2}) \cdot \sqrt{d} \cdot \alpha^{-1/(2d)} \;.$$
Since this holds for all unit vectors $v$, we have that when
the algorithm returns $\mu_T$ and $T$ is $\alpha$-good,
then
$$\|\mu-\mu_T\|_2 \leq O\left(\alpha^{-1/(2d)} \cdot \sqrt{d} \cdot O((d+\log(1/\alpha)) \new{\log (2+\log(1/\alpha))^2})^{d}\right) \;.$$
This completes the proof or Proposition~\ref{mainSubroutineProp}.
\end{proof}

\section{Learning Spherical Gaussian Mixture Models} \label{sec:gmms}

In Section~\ref{ssec:gmm-identity}, we present a simpler learning algorithm 
that works when the components have the same covariance matrix.
The general case of unknown (potentially different) covariances is more complex
and is handled in Section~\ref{ssec:gmm-general}. Section~\ref{ssec:dim-red}
contains our dimension-reduction procedures. In Section~\ref{ssec:gmm-final}, we put
everything together to obtain our final learning guarantees, 
including Theorem~\ref{thm:gmm-inf}.

\subsection{Learning Spherical GMMs:  The Identity Covariance Case} \label{ssec:gmm-identity}

We start by handling the important special case of this problem where each Gaussian component
has identity covariance matrix. Note that our learning algorithm is robust to a small
constant fraction of corrupted samples:

\begin{proposition} \label{prop:gmm-identity}
There is an algorithm that given a positive integer $d$, constants $1/2>\alpha > 4\eps \geq 0$, $0 < \tau < 1$,
and sample access to a probability distribution $X=(1-\eps)M+\eps Y$,
where $M=\sum w_i N(\mu_i,I)$ is a mixture of identity covariance Gaussians
with $w_i\geq \alpha$ for all $i$, and so that $\|\mu_i-\mu_j\|_2$ is at least
$$
S \eqdef C (\alpha^{-1/(2d)} \sqrt{d} (d+\log(1/\alpha))(\log(2+\log(1/\alpha)))^2+ \sqrt{\log(1/\eps)})
$$
for all $i\neq j$, takes \new{$\poly( (nd)^d \log(1/\tau)/(\eps \alpha))$} samples from $X$,
runs in time  \new{$\poly( (nd\log(1/\tau)/\alpha)^{O(d)}, 1/\eps)$}, 
and, with probability at least $1-\tau$, returns a list of pairs $(u_i,\nu_i)$,
so that up to some permutation $|u_i-w_i|=O(\eps)$ and $\|\mu_i-\nu_i\|_2 = \tilde O(\eps/w_i)$.
\end{proposition}
\begin{proof}
The algorithm itself is very simple.
We run our list-decoding algorithm to get a list of hypothesis means.
We then associate each sample with the closest element of our list.
We can then cluster points based on which means they are associated to
and use this to learn the correct components. The algorithm is as follows:

\medskip

\fbox{\parbox{6.1in}{
{\bf Algorithm} {\tt LearnIdentityCovarianceGMM}\\
Input: Parameters $k , d \in \Z_+$, $\tau, \eps>0$ and sample access to $X$.\\
\vspace{-0.5cm}
\begin{enumerate}
\item Let $T$ be a set of sufficiently many \new{$\poly( (nd)^d \log(1/\tau)/\alpha)$} samples from $X$.
\item Run Algorithm {\tt List-Decode-Gaussian} using $T$
to obtain a list $H=(h_1,h_2,\ldots,h_{m})$ with $m=O(1/\alpha)$.
\item Let $T'$ be a set of sufficiently many $\poly(nk/\eps)$ samples from $X$.
\item For each sample from $T'$ associate it to the closest element of $H$ in $\ell_2$-distance.
\item Let $H'$ be the set of $h \in H$ so that at least a $2\alpha/3$-fraction of the elements of $T'$
are associated to an element of $H$ at most $S/10$ away from $h$.
\item \label{cluster1Step} Define the relation on $H'$ that $h\sim h'$ if and only if $\|h-h'\|_2 \leq S/3$.
If this does not define an equivalence relation on $H'$ return ``FAIL''.
\item \label{cluster2Step} For each equivalence class $C$ of $H'$,
let $T_C$ be the set of points in $T'$ that are associated to elements of $C\subset H$.
\item Let $u_C=|T_C|/|T'|$ for each $C$.
\item For each $C$, run {\tt Filter-Gaussian-Unknown-Mean} from \cite{DKKLMS16} on $T_C$,
and let $\nu_C$ be the approximation of the mean obtained.
\item Return the list of $(u_C,\nu_C)$.
\end{enumerate}
\vspace{-0.2cm}
}}

\bigskip

Note that for each $i$, $X$ is simultaneously a mixture of $N(\mu_i,I)$ with weight $\alpha$ and 
some other distribution with weight $(1-\alpha)$. Therefore, for each $i$ with probability 
at least $1-\tau/(10 k)$, there is some $h_j\in H$ with 
$$\|h_i-\mu_i\|_2  =  O(\alpha^{-1/(2d)} \sqrt{d} (d+\log(1/\alpha))(\log(2+\log(1/\alpha)))^2) \leq S/100 \;,$$
for $C$ is sufficiently large. 
By a union bound, with probability at least $1-\tau/10$, this occurs for every $i$. 
We assume this holds throughout the remainder of our analysis.

Let $S_i$ be the set of elements of $T'$ drawn from the component $N(\mu_i,I)$.
\begin{lemma}\label{clusterLem}
With probability $1-\exp(-\Omega(S^2))$ over the samples from $T'$, all but an $\exp(-\Omega(S^2))$-fraction
of the elements of $S_i$ are associated
with elements $h\in H$ with $\|h-\mu_i\|_2\leq  S/20.$
\end{lemma}
\begin{proof}
The basic idea of the proof is the following:
For any given $h\in H$ that is far from $\mu_i$,
there will be some $h'\in H$ that is much closer.
A given sample point $x$ will only be closer to $h$ than $h'$ if its projection
to the line between them is more than half way there.
However, this projection is distributed as a Gaussian,
and therefore the probability that it is much larger than its mean is small.

It suffices to show that for each $h\in H$ with $\|h-\mu_i\|_2 >  S/20$
the following holds: less than a $\exp(-\Omega(S^2))$-fraction of the elements of $S_i$ are associated with $h$.

Firstly, assuming that the first step was successful,
we know that there is an $h'\in H$ with $\|h'-\mu_i\|_2 < S/100$.

Let $v$ be the unit vector in the direction of $h-h'$. We note that $x$ is closer to $h$ than $h'$
if and only if $v \cdot x \geq v \cdot (h+h')/2$. However, we note that $v\cdot \mu_i \leq v\cdot h' + S/100$,
whereas, $v \cdot h = v \cdot h' + \|h-h' \|_2 \geq v \cdot  h' + S/20$.
The probability that $v \cdot X \geq \E[v \cdot X] + S/50$ for $X$ drawn from $N(\mu_i,I)$ is $\exp(-\Omega(S^2))$.
Thus, the probability that a sample drawn from $N(\mu_i,I)$ is closer to $h$ than $h'$ is  $\exp(-\Omega(S^2))$.

Thus, by Markov's inequality, the probability that more 
than a $\exp(-\Omega(S^2))$-fraction of the elements of $S_i$ 
are associated to $h$ (with suitably small constant in the $\Omega(\cdot)$), 
is $\exp(-\Omega(S^2))$. Taking a union bound over $h$, does not change this asymptotic.
\end{proof}

Taking a union bound over $i$, we can assume that, with probability at least $1-\tau/10$, 
we have that all but an $\exp(-\Omega(S^2))$ fraction of the points of $S_i$ 
are associated with some $h_j$ where $\|h_j-\mu_i\|_2 \leq S/20$.
In particular, this implies that every element of $H$ within distance 
$S/20$ of some $\mu_i$ is in $H'$. Indeed, this holds for the following reason: 
With high probability, $|S_i|\geq (3\alpha/4) \cdot |T'|$ and at least $8/9$ fraction 
of elements in $S_i$ are associated with $h_j$'s that are within distance $S/20$ of $\mu_i$. 
By the triangle inequality these $h_j$'s are within distance $S/3$ of $h$. 
Conversely, any element of $H$ not within distance $S/20$ of some $\mu_i$
has associated with it at most an $\eps/10$-fraction of the elements of the union of the $S_i$'s.
This implies that with high probability less than $1.2\eps < 2\alpha/3$ fraction of points in $T$ 
are associated to \emph{any} point of $H$ \new{not within distance $S/20$ of some $\mu_i$}.
Therefore, all points of $H'$ are within distance $3S/20$ of some $\mu_i$, which implies that the relation
on $H'$ is an equivalence relation. Specifically, each equivalence class consists 
of the points in $H'$ within distance $3S/20$ of some particular mean $\mu_i$. 
Note in particular that this implies that there is exactly one equivalence
class $C$ for each $\mu_i$.

Furthermore, Lemma~\ref{clusterLem} implies that all but an $\eps/(10k)$-fraction of the samples from $N(\mu_i,I)$
are associated with elements of $H$ in the class associated with $\mu_i$. Furthermore, at most an
$\eps$-fraction of the other samples from $T$ are associated to elements of this class. From this it immediately
follows that $|u_i-w_i| \leq 1.2\eps$. Furthermore, the points associated with this class are
an $O(\eps/w_i)$-noisy version of $N(\mu_i,I)$. Therefore, {\tt Filter-Gaussian-Unknown-Mean}
returns a mean $\nu_i$ with $\|\nu_i-\mu_i\|_2 =\tilde O(\eps/w_i)$.
This completes the proof of Proposition~\ref{prop:gmm-identity}.
\end{proof}

\subsection{Learning Spherical GMMs:  The General Case} \label{ssec:gmm-general}
We now generalize the algorithm from the previous subsection
to handle arbitrary mixtures of spherical Gaussians.
When it is not the case that all of the covariance matrices are the same,
things are substantially more complicated. We can recover a list $H$ of candidate means
only after first guessing the radius of the component that we are looking for.
We can produce a large list of guesses 
and thereby obtain a list of hypotheses of size $\poly(n/\alpha)$.
However, clustering becomes somewhat more difficult,
as we do not know the radius at which to cluster. 
In particular, Steps \ref{cluster1Step} and \ref{cluster2Step} 
become difficult not knowing at what distance to stop 
considering two hypotheses part of the same cluster. This difficulty can be dealt with
by doing a secondary test to determine whether or not the cluster
that we have found contains many points at approximately the correct distance from each other.

\begin{proposition} \label{prop:gmm-general}
There is an algorithm that, given a positive integer $d$,
constants $1/2>\alpha > 3\eps \geq 0$, and sample access
to a probability distribution $X=(1-\eps)M+\eps Y$,
in dimension $n$ larger than a sufficiently large multiple of $\log(\tau/\alpha)$,
where $M=\sum_{i=1}^k w_i N(\mu_i,\sigma_i^2 I)$ is a mixture of spherical Gaussians with $w_i\geq \alpha$ for all $i$,
and so that $\|\mu_i-\mu_j\|_2/(\sigma_i+\sigma_j)$ is at least
$$
S \eqdef C( \alpha^{-1/(2d)} \sqrt{d} (d+\log(1/\alpha))(\log(2+\log(1/\alpha)))^2 + \sqrt{\log(nk/\eps)}) \;,
$$
for all $i\neq j$, where $C$ is a sufficiently large universal constant,
takes $\poly((dn)^d \log(1/\tau)/(\alpha\eps))$ samples from $X$, runs in time 
\new{$\poly((dn \log(1/\tau)/\alpha)^d /\eps)$} 
and, with probability at lest $1-\tau$, returns a list of triples $(u_i,\nu_i,s_i)$,
that satisfy the following conditions (up to some permutation): $|u_i-w_i|=O(\eps)$,
$\|\mu_i-\nu_i\|_2 = \tilde O(\eps/w_i)\sigma_i$, and $|s_i-\sigma_i|/\sigma_i = \tilde O(\eps/w_i)/\sqrt{n}$.
\end{proposition}
\begin{proof}
Throughout this proof ``with high probability'' will always mean 
with probability at least $1-\tau (\alpha/n)^c$ for a sufficiently large constant $c$.

The algorithm is as follows:

\medskip

\fbox{\parbox{6.1in}{
{\bf Algorithm} {\tt LearnSphericalGMM}\\
Input: Parameters $k , d \in \Z_+$, $\tau, \eps>0$ and sample access to $X$.\\
\vspace{-0.5cm}
\begin{enumerate}
\item Draw a set $T$ of $\poly(d!^2 n^{4d} \log(k/\tau)k/\eps)$ (for a sufficiently large polynomial) independent samples from $X$.

\item Let $L$ be the set of powers of $(1+1/(10n))$ that are within a constant multiple of $\|x-y\|_2^2/{n}$,
for at least an $\alpha^2/4$ fraction of pairs of $x$ and $y$ from the sample set $T$.

\item For each $s^2 \in L$, run Algorithm {\tt List-Decode-Gaussian}  on samples from $X$
to obtain a list $(h_1,h_2,\ldots,h_{m'})$ with $m'=O(1/\alpha)$,
and so that with high probability if $X$ is a mixture of $N(\mu, s^2I)$ \new{with weight $2\alpha/3$},
then $\mu$ is within distance $sS/10000$ of some $h_i$.
Let $H$ be the union of these lists over elements of $L$.
Let $H=\{h_1,\ldots,h_m\}$ with $m=\poly(n/\alpha)$.
To each $h_i$ we associate an $s_i$, the value of $L$ used 
in the application of Algorithm {\tt List-Decode-Gaussian} that was used to find it.
\item Draw a set $U$ of $O(k\log(k/\tau)/\alpha^2)$ samples from $X$.
\item For each sample in $U$ 
associate to it the closest element of $H$.
\item For each sample $x\in U$, let $r(x)$ be the minimum $r$ such that 
an $\alpha/2$-fraction of $x \in U$ are within distance $r \sqrt{n}$.
\item We let $h_i$ be in $H'$ if at least an $4\alpha/5$-fraction of the elements of $x\in U$
have $s_i/3 \leq r(x) \leq 3s_i$ and are associated to an element of $H$ at most $s_i S/1000$ away from $h_i$.
\item Define the relation on $H'$ that $h\sim h'$ if and only if $\|h-h'\|_2 \leq (s_i+s_j)S/10$.
If this does not define an equivalence relation on $H'$, return ``FAIL''.
\item Let $T'$ be a set of $\poly(nk/\eps)$ (for a sufficiently large polynomial) independent samples from $X$.
\item For each equivalence class $C$ of $H'$, let $T_C$ be the set of points in $T'$ whose closest neighbor in $H'$ is an element of $C\subset H'$.
\item Let $u_C=|T_C|/|T'|$ for each $C$.
\item For each $C$, run {\tt Filter-Gaussian} from \cite{DKKLMS16} on $T_C$, so that if the samples come
from an $O(\eps/u_C)$-noisy Gaussian, the mean is approximately $\nu_C$
and the covariance matrix has trace approximately $n\cdot s_C^2$.
\item Return the list of $(u_C,\nu_C,s_C)$.
\end{enumerate}
\vspace{-0.2cm}
}}

\bigskip

The algorithm begins by producing a list $L$ of hypothesis standard deviations. 
We note that if $x$ and $y$ are both drawn from $N(\mu_i,\sigma^2_i I)$ then 
with high probability the distance between $x$ and $y$ is $\Theta(\sigma_i\sqrt{n})$. 
Therefore, since with high probability $T$ contains at least $(2\alpha/3)^2$ such pairs, 
with high probability the largest power of $(1+1/(10n))$ less than $\sigma_i^2$ is in $L$ for all $i$. 
We also note that $L$ cannot be too large. In particular, of $s\in L$ than an $\Omega(\alpha^2)$-fraction 
of the pairs in $T$ have distance $\Theta(s\sqrt{n})$. This cannot happen for more than 
$O(\alpha^{-2})$ values of $s$ that differ pairwise by sufficiently large constant multiples. 
From this it is easy to see that $|L| = O(n\alpha^{-2})$. This implies that $|H| = O(n\alpha^{-5})$.

Additionally, note that $N(\mu_i, \sigma^2_i I)$ is a mixture of $N(\mu_i, s^2_i I)$ 
with some other distribution. In particular:
\begin{claim}
Given $\sigma,\sigma'>0$ with $\sigma^2/(1+1/10n) \leq \sigma'^2 \leq \sigma^2$,
$N(\mu,\sigma^2 I)$ can be considered as a mixture
$0.9 N(\mu,\sigma'^2 I)+0.1 E$, for some distribution $E$.
\end{claim}
\begin{proof}
It suffices to show that the probability density function of $N(\mu,\sigma^2 I)$
is bigger than $0.9$ times that of $N(\mu,\sigma'^2 I)$
at every point $x \in \R^n$, since then we can define $E$
as the normalization of their difference. The pdf of $N(0,\sigma^2 I)$ is
$(1/\sqrt{2\pi \det(\sigma^2 I)}) \exp(-\|x-\mu\|_2^2/2\sigma^2)$,
and so the ratio of this and the pdf of $N(0,\sigma'^2 I)$
is $(\sigma'/\sigma)^n \exp(-(\|x-\mu\|_2^2/2)(1/\sigma^2 - 1/\sigma'^2))$.
Now if $ \sigma^2 \geq \sigma'^2 \geq \sigma^2/(1+1/(10n))$, we have:
\begin{align*}
(\sigma'/\sigma)^n \exp(-(\|x-\mu\|_2^2/2)(1/\sigma^2 - 1/\sigma'^2)) & \geq (\sigma'/\sigma)^n \\
& \geq (1+1/(10n))^{-n/2} = \exp(-n \ln(1+1/10n)/2) \\
& \geq \exp(-1/20) \geq 19/20 \tag*{(since $\ln(1+y) \leq y$ for all $y >-1)$.} \\
\end{align*}
This completes the proof of the claim.
\end{proof}

This means that $X$ is a mixture of $w_i(9/10)N(\mu_i,s_i^2 I)$ 
with some other distribution, when $s_i$ is the largest power of $(1+1/(10n))$ less than $\sigma_i$. 
Therefore, with high probability over $T$, our list $H$ includes a hypothesis $h$ (associated with $s_i$) 
so that $\|h-\mu_i\|_2 \leq s_i S/10000$. We assume that this is true in the remainder.

Let $U_i$ be the set of samples of $U$ drawn from $N(\mu_i,\sigma_i^2 I)$, 
and $S_i$ the set of these samples from $T'$. Note that \new{with high probability}
$|U_i| \geq (9\alpha/10) |U|$ for all $i$, and $|S_i|=(w_i+O(\eps))|T'|$ for all $i$.

We would like to claim that, for $x$ drawn from $U_i$, 
with high probability that $\sigma_i/2 \leq r(x) \leq 2\sigma_i$. 
For the upper bound, note that any pair of elements from $N(\mu_i,\sigma_i^2 I)$ 
have distance at most $2\sigma_i\sqrt{n}$ with high probability. 
Thus, as long as this holds for at least $90\%$ of the other elements of $U_i$, 
we will have $r(x)\leq 2\sigma_i$. 
For the lower bound, consider fixing the elements of $U$ other than $x$, 
and then drawing $x$ uniformly at random from $N(\mu_i,\sigma^2_i I)$. 
We note that for any other $y\in U$, except with $\exp(-\Omega(n))$-probability, 
it holds $\|x-y\|_2 \geq \sigma_i \sqrt{n}/2$. 
Since $n \gg \log(1/(\tau\alpha))$, with high probability, this cannot happen 
for more than an $\alpha/2$-fraction of the elements of $U$. 
Thus, we may assume that for all $x\in U_i$ we have $\sigma_i/2 \leq r(x) \leq 2\sigma_i$.

We next want to show that each component Gaussian 
has a corresponding cluster of hypotheses in $H'$. 
To do this, we want to show that the closest element of $H$ 
to any element of $U_i$ is not too far from $\mu_i$.
We show the following lemma:

\begin{lemma} \label{lem:closest-close-again}
With high probability, all of the elements $x$ of $U_i$ are associated
with elements $h\in H$ with $\|h-\mu_i\|_2\leq  \sigma_i S/2000.$
\end{lemma}
\begin{proof}
The proof is essentially identical to
that of Lemma \ref{clusterLem}.

It suffices to show that an element $x$ drawn from $N(\mu_i, \sigma_i^2)$ 
is not associated with a specific $h$ with $\|h-\mu_i\|_2 > \sigma_i S/200$ with high probability. 
Taking a union bound over $h$ and a Chernoff bound over $U_i$ will complete the proof. 
We note that there is always an $h'$ at distance at most $\sigma_i S/10000$ from $\mu_i$. 
We note that $x$ is closer to $h$ than $h'$ only if its inner product with $v$, 
the unit vector in the $h-h'$ direction is more than $v\cdot (h+h')/2$. 
However, this is greater than the mean of $v\cdot x$ by at least $S/10000$, 
so it happens with probability at most $\exp(-\Omega(S^2))$. 
This completes the proof.
\end{proof}

We make two claims about $H'$. 
First, that all elements $h_i \in H'$ are within distance $\sigma_j S /1000$ of $\mu_j$ for some $j$, 
and that $6\sigma_j \geq s_i \geq \sigma_j/6$. Second, for each $j$, that there is 
such an $h_i\in H'$. This will tell us that the elements of $H'$ correspond to clusters 
around each true mean.

On the one hand, with high probability, for each $j$, 
at least $9\alpha/10$-fraction of elements in $U_j$ satisfy 
$2\sigma_j \geq r(x) \geq \sigma_j/2$, 
and at least $4\alpha/5$ fraction of them are associated to elements of $H$ 
within distance $\sigma_j S/2000$ of $\mu_j$. Furthermore, $H$ contains 
an $h_i$ with $\sigma_j(1+1/(10n))\leq s_i\leq \sigma_j$ 
and with $h_i$ within distance $s_i S/1000$ of $\mu_j$. 
If this is all the case, then a $4\alpha/5$-fraction of the elements of $U$ 
are associated to elements of $H$ within distance $s_i S/1000$ of $h_i$, 
and thus $h_i\in H'$.

In the other direction, we note that if $h\in H'$ then 
at least a $(4\alpha/5-2\eps)$-fraction of the points of $x\in U$ 
are from some $U_i$, satisfy $s_i/3 \leq r(x)\leq 3s_i$, 
and are associated with a point of $H$ within distance 
$s_i S/1000$ of $h$. Therefore, for some $j$, 
at least an $\alpha/5$-fraction of the elements of $U_j$ have this property. 
We claim that with high probability this cannot happen unless $h$ 
is within distance $\sigma_j S/100$ of $\mu_j$. 
For one, if $s_i/3 \leq r(x)\leq 3s_i$, it must be the case that 
$s_i/6\leq \sigma_j \leq 6 s_i$. By Lemma \ref{lem:closest-close-again}, 
we must have that $x$ is associated to an element within distance 
$\sigma_j S/ 2000$ of $\mu_j$. 
This is within distance $s_i S/1000\leq 6\sigma_j S/1000$ of $h_i$, 
and therefore $\|h_i-\mu_j\|_2 \leq 7\sigma_j S/1000 \leq \sigma_j S/100$.

This implies that the elements of $H'$ cluster into equivalence classes as desired. 
Two distinct points $h_i$ and $h_{i'}$ close to the same $\sigma_j$ will have 
$\|h_i-h_{i'} \|_2 \leq \sigma_j S/100 \leq 3 (s_i+s_{i'})/100 < (s_i+s_{i'})/10$, 
and thus are equivalent. On the other hand, if $h_i$ is close to $\mu_j$ and $h_{i'}$ is close to $\mu_{j'}$, then
\begin{eqnarray*}
\|h_i-h_{i'} \|_2 
&\geq \|\mu_j-\mu_{j'} \|_2 - \|h_i-\mu_j \|_2 - \|h_{i'}-\mu_{j'}\|_2 \\
&\geq S(\sigma_j + \sigma_{j'})-\sigma_j S/100 - \sigma_{j'} S/100 \\ 
&\geq  99 S (\sigma_j+\sigma_{j'})/100 \geq S(s_i+s_{i'})/10 \;.
\end{eqnarray*}
We claim that with high probability 
all but an $\eps$-fraction of the elements of $S_i$ are in $T_C$, 
where $C$ is the equivalence class containing elements of $H'$ close to $\mu_i$. 
This is for the same reason as Lemmas \ref{lem:closest-close-again} and \ref{clusterLem}. 
$T_C$ contains elements of $H'$ that are at most $\sigma_i S/100$ far from $\mu_i$. 
Other $T_{C'}$ only contain elements of $H'$ that are $\sigma_j S/100$-far 
from other $\mu_j$, and therefore, at least $99(\sigma_i+\sigma_j)S/100 > \sigma_i S/2$ far from $\mu_i$. 
A given sample from $N(\mu_i, \sigma_i^2 I)$ is closer to the former 
except with probability $\exp(-\Omega(S^2))$, and thus with high probability 
all but an $\eps$-fraction of the elements of $S_i$ are associated with $T_C$.

With high probability, we also have that $|S_i| = |T'|(w_i+O(\eps))$. 
Combining with the above, it is easy to see that with high probability 
$|T_C|=|T'|(w_i+O(\eps))$, since $T_C$ contains the elements of $S_i$ (minus an $\eps$-fraction) 
plus at most an $O(\eps)$-fraction of other points. This shows that $u_C = w_i +O(\eps)$. 
Furthermore, the points of $T_C$ are a $1-O(\eps/w_i)$ fraction of the points from $S_i$, 
which are i.i.d. samples from $N(\mu_i,\sigma_i^2 I)$,
mixed with $O(\eps)|T'| = O(\eps/w_i) |T_C|$ other points. 
Therefore, \new{{\tt Filter-Gaussian}} learns (with high probability) the Gaussian 
$N(\mu_i,\sigma_i^2 I)$ within error $\tilde O(\eps/w_i)$. 
Thus, we have that $\|\nu_C -\mu_i\|_2 = \sigma_i \tilde O( \eps/w_i)$, 
and the returned matrix is within $\sigma_i^2 \tilde O(\eps/w_i)$ of $\sigma_i^2 I$ in Frobenius norm. 
This implies that its trace is $n\sigma^2 (1+\tilde O(\eps/w_i)/\sqrt{n})$. 
Therefore, $s_C = \sigma_i(1+\tilde O(\eps/w_i)/\sqrt{n})$.
This completes the proof of Proposition~\ref{prop:gmm-general}.
\end{proof}

\subsection{Dimension Reduction} \label{ssec:dim-red}
In this section, we describe our dimension reduction scheme
for the case of spherical mixtures. \new{When the components
have the same covariance, dimension reduction is quite simple
and allows us to assume without loss of generality that the ambient dimension
is $k-1$. The effect of dimension reduction for this case is that the runtime
of the learning algorithm becomes somewhat better as a function of $n$.

When the components have arbitrary spherical covariances, 
we require a more complicated procedure that allows us to
reduce the dimension down to $\poly(k/\eps)$. 
In addition to improving the dependence on $n$ in the runtime,
this has the effect of removing the $\Omega(\sqrt{\log(n)})$ 
dependence in the separation condition
of Proposition \ref{prop:gmm-general}.
}

For the case of identity covariance components, 
we will require the following generalization of Theorem 4.2 of~\cite{RV17} 
or Corollary~3 of \cite{VempalaWang-journal}:
\begin{lemma} \label{lem:dim-reduce} 
Given $\eps > 0$, suppose we take $\Omega(n\log(k/\tau)/(\eps^4 w_{\min}^4))$ independent samples 
from $X=\sum_{i=1}^k w_i N(\mu_i,\sigma_i^2 I)$, where $w_i \geq w_{\min}$, 
and let $W$ be the affine subspace of dimension $k-1$ containing 
the empirical mean $\wt \mu$ and spanned by the top $k-1$ eigenvectors 
of the empirical covariance $\wt \Sigma$. Then, with probability at least $1-\tau$, 
for all $i$, $\mu'_i$, the orthogonal projection of $\mu_i$ onto $W$, satisfies 
$\|\mu'_i-\mu_i\|_2 \leq \eps \sigma$, where $\sigma^2=\sum_i w_i \sigma_i^2$. 
\end{lemma}

Note that unlike Corollary~3 of~\cite{VempalaWang-journal}, 
we only need $W$ to be $k-1$ dimensional 
and unlike Theorem~4.2 of \cite{RV17}, 
we do not need the means $\mu_i$ to be bounded.

\begin{proof}
We first use standard facts about the empirical mean and covariance matrix of a single Gaussian:
\begin{fact} 
If we take $\Omega(n\log(1/\tau)/\eps^2)$ independent samples from a Gaussian $N(\mu,\Sigma)$, 
then, except with probability $\tau$, we have that the empirical covariance $\wt \Sigma$ 
and empirical mean $\wt \mu$ satisfy $(1-\eps) \Sigma \preceq \wt \Sigma \preceq (1+\eps) \Sigma$ 
and $(\wt \mu - \mu)^T \Sigma (\wt \mu - \mu) \leq \eps^2$. 
\end{fact}

Let $\delta=\eps^2 w_{\min}/12$. 
By Chernoff bounds, the above fact and a union bound, 
we have that except with probability $\tau$, 
since we have $\Omega(n\log(k/\tau)/(\delta^2 w_{\min}^2))$ samples, 
the fraction of samples from $N(\mu_i,\sigma_i^2 I)$, 
$\wt w_i$, satisfies $(1-\delta)w_i \leq \tilde w_i \leq (1+\delta) w_i$, 
and the empirical covariance $\wt \Sigma_i$ and mean $\wt \mu_i$ 
of the samples coming from $N(\mu_i, \sigma_i^2 I)$ satisfy 
$(1-\delta) \sigma_i^2 I \preceq \wt \Sigma_i \preceq (1+\delta) \sigma_i^2 I$ 
and $\|\wt \mu_i - \mu_i\|_2 \leq \delta \sigma_i$. We assume that this holds.

Next note that we can write the empirical covariance as
$$\wt \Sigma = \sum_{i=1}^k \wt w_i \left( \wt \Sigma_i + (\wt \mu_i - \wt \mu)(\wt \mu_i - \wt \mu)^T \right) \;.$$
Since $\wt \mu$ is a convex combination $\wt \mu= \sum_{j=1}^k \wt w_j \wt \mu_j$ of the $\wt \mu_j$, 
the vectors $\wt \mu_i - \wt \mu$ span a $(k-1)$-dimensional subspace. 
For any unit vector $v$ in the $(n-k+1)$-dimensional subspace 
orthogonal to this subspace, we have
$$v^T \wt \Sigma v = \sum_{i=1}^k \wt w_i v^T \wt \Sigma_i v \leq \sum_{i=1}^k (1+\delta) w_i (1+\delta) \sigma_i^2 \leq (1+3\delta) \sigma^2 \;.$$
Thus, the bottom $n-k+1$ eigenvalues 
of $\wt \Sigma$ are at most $(1+3\delta) \sigma^2$.

Now consider $\wt \mu'_i$, the orthogonal projection of $\wt \mu_i$ onto $W$. 
Let $v= (1/\|\wt \mu_i - \wt \mu'_i\|_2) (\wt \mu_i - \wt \mu'_i)$. 
Since $v$ is orthogonal to the top-$k$ eigenvectors of $\wt \Sigma$, 
it follows that $v^T \wt \Sigma v \leq (1+3\delta) \sigma^2$. 
Since $v$ is orthogonal to $W$ which contains $\wt \mu'_i$ and $\wt \mu$, 
we have $v^T (\wt \mu'_i - \wt \mu)=0$. Thus, we have
\begin{align*}
(1+3\delta) \sigma^2 & \geq v^T \wt \Sigma v \\
& = \sum_{j=1}^k \wt w_j \left( v^T \wt \Sigma_j v + (v^T (\wt \mu_j - \wt \mu))^2 \right) \\
& \geq (1-\delta)  \sum_{j=1}^k w_j \left( (1-\delta) \sigma_j^2  + (v^T (\wt \mu_j - \wt \mu))^2 \right) \\
& \geq (1-2\delta) \sigma^2 + (1-\delta) w_i  (v^T (\wt \mu_i - \wt \mu))^2 \\
& = (1-2\delta) \sigma^2 + (1-\delta) w_i  (v^T (\wt \mu_i - \wt \mu'_i) + v^T (\wt \mu'_i - \wt \mu) )^2 \\
& = (1-2\delta) \sigma^2 + (1-\delta) w_i  (\|\wt \mu_i - \wt \mu'_i\|_2 +0)^2 \; .
\end{align*}
Re-arranging, we have 
$\|\wt \mu_i - \wt \mu'_i\|_2 \leq \sqrt{5\delta\sigma^2/((1-\delta) w_i)}$. 
Setting $\delta=\eps^2 w_{\min}/12 $
gives $\|\wt \mu_i - \wt \mu'_i\|_2 \leq \eps \sigma/2$. 

Noting that projecting onto an orthogonal space reduces Euclidean distance, 
we have $\|\wt \mu'_i-\mu'_i\|_2 \leq \|\wt \mu_i-\mu_i\|_2$. 
The triangle inequality gives 
$\|\mu_i-\mu'_i\|_2 \leq \|\wt \mu_i-\mu_i\|_2 + \|\wt \mu_i - \wt \mu'_i\|_2  + 
\|\wt \mu'_i-\mu'_i\|_2  \leq \eps \sigma /2 + 2\delta \sigma_i \leq \eps \sigma/2 + 
\eps \sqrt{w_i} \sigma_i/2 \leq \eps \sigma$.
This completes the proof.
\end{proof}

\new{We now handle the general case:}

\begin{proposition}\label{dimReduxProp}
Let $X=\sum_{i=1}^k w_i N(\mu_i,\sigma_i^2 I)$ be a $k$-mixture of spherical Gaussians in $\R^n$ 
with $w_i\geq \eps$ for all $i$, for some $\eps>0$. There exists an algorithm that given $k$ and $\eps$, 
draws $\poly(nk/\eps)$ samples from $X$, runs in sample-polynomial time, 
and returns an affine subspace $W$ of dimension $\poly(k/\eps)$, 
so that with high probability each $\mu_i$ is within distance $O(\eps \sigma_i)$ of its projection onto $W$.
\end{proposition}

\noindent {\bf Remark.}
Note that the above proposition can be combined with our algorithm from Section~\ref{ssec:gmm-general}
by first finding $W$ and then learning the projection of $X$ onto $W$. 
Since we are now working in only $\poly(k/\eps)$ dimensions,
the latter does not require a $\log(n)$ dependence on $S$.

\begin{proof}
We start with the following observation: 
If we knew that all of the $\sigma_i$'s were 
within a constant multiple of some known $\sigma$, 
we could simply scale $X$ down by a factor of $\sigma$ 
and then project onto the top $k$ eigenvectors of the empirical covariance. 
This approach is used in~\cite{VempalaWang-journal}. 
The difficulty comes when the $\sigma_i$'s are not close to each other.
Doing this in this case would only give error $O(\eps \sqrt{\sum_j w_j \sigma_j^2})$, 
which will be larger than $O(\eps \sigma_i)$ for some $i$. 
To deal with this issue, we notice that similar to the proof of Proposition~\ref{prop:gmm-general}, 
we can approximate the $\sigma$ associated to a given sample 
by measuring how close it is to other samples. This will allow us to break our samples 
into several subsets each of which is a mixture of Gaussians with similar covariances.

It will also be important to note that our accuracy in measuring the radius of a Gaussian 
based on a few samples gets better as the dimension increases. Fortunately, 
we can assume without loss of generality 
that $n$ is sufficiently large, as otherwise we can simply return $W=\R^n$.
The dimension-reduction algorithm is as follows:

\medskip

\fbox{\parbox{6.1in}{
{\bf Algorithm} {\tt DimensionReduce}\\
Input: Parameters $k \in \Z_+ , \eps>0$, and sample access to $X$.\\
\vspace{-0.7cm}
\begin{enumerate}
\item If $n$ is not larger than a sufficiently large polynomial in $k/\eps$, return $W=\R^n$.
\item Let $U$ be a set of $N=\poly(nk/\eps)$ (for a sufficiently large polynomial) samples from $X$.
\item For each $x\in U$, let $r(x)=\min_{y\in U, y\neq x} \|x-y\|_2/\sqrt{n}.$
\item Define a relation $x\sim y$ if $r(x)$ and $r(y)$ are within a multiplicative factor of
$(1\pm n^{-1/3})$. 
Let $\{C_j\}$ be the equivalence classes under the transitive closure of $\sim$.
\item For each $C_j$:
\begin{enumerate}
\item let $s_j$ be the minimum value of $r(x)$ for $x\in C_j$.
\item Compute  $\wt \mu_j$ and $\wt \Sigma_j$, the empirical mean and covariance matrix of $C_j$.
\item Use PCA to find the $k-1$ eigenvectors $v_1,\dots ,v_{k-1}$ of $\wt \Sigma_j$ with the largest eigenvalues.
\item Let $W_j = \wt \mu_j + \mathrm{span}<v_1,\dots ,v_{k-1}>$.
\end{enumerate}
\item Return $W$, the affine span of the $W_j$'s.
\end{enumerate}
}}

\bigskip

We note that we can assume that $n\gg \poly(N)$, 
for a sufficiently large polynomial 
or the algorithm trivially terminates in Step~1. 
We assume this throughout the rest of this proof.

In order to analyze the algorithm, we need to understand the distribution of the $r(x)$. 
To begin with, we note that:
\begin{lemma}\label{rConcLem}
With high probability over our samples, for every $x\neq y$ from $U$,
with $x$ drawn from $N(\mu_i,\sigma_i^2 I)$ 
and $y$ drawn from $N(\mu_j,\sigma_j^2 I)$, we have that 
$\|x-y\|_2^2 = \left(\|\mu_i-\mu_j\|_2^2 + (\sigma_i^2+\sigma_j^2)n\right)(1 + o(n^{-1/3})).$
\end{lemma}
\begin{proof}
We note that, for any given choice of $i$ and $j$, 
a random pair of $x$ and $y$ satisfy this except 
with $\exp(-\Omega(n))$ probability. 
The lemma follows from a union bound over $x$ and $y$.
\end{proof}

\noindent Taking a minimum, we find that:

\begin{lemma}
With high probability, for all $x\in U$ drawn from $N(\mu_i,\sigma_i^2 I)$, 
we have that $$r(x) = \min_j \left(\sqrt{\|\mu_i-\mu_j\|_2^2/n + (\sigma_i^2+\sigma_j^2)}\right)(1+o(n^{-1/3})).$$
\end{lemma}
\begin{proof}
Assuming the conclusion of Lemma \ref{rConcLem} holds, then the $r(x)$ 
is automatically at least this big and is at most this large assume 
that at least one (other) sample was drawn from $N(\mu_j,\sigma_j^2 I)$ 
for the minimizing $j$. This of course happens with high probability.
\end{proof}

\begin{corollary}
With high probability, for all $x$ drawn from $N(\mu_i,\sigma_i^2 I)$ 
we have that $\sigma_i(1-o(n^{-1/3})) \leq r(x) \leq \sqrt{2}\sigma_i(1+o(n^{-1/3}))$.
\end{corollary}
\begin{proof}
The lower bound is immediate. The upper bound follows from taking $j=i$.
\end{proof}

The above Corollary also has several other consequences. 
All of the $x$'s coming from the same component 
will have $r(x)$ close to $\min_j (\sqrt{\|\mu_i-\mu_j\|_2^2/n + (\sigma_i^2+\sigma_j^2)})$ 
and thus all lie in the same $C_j$. This implies that there are at most $k$ many classes $C_j$.
Furthermore, since the $x$'s coming from a single Gaussian component 
all have $r(x)$ within a $1+o(n^{-1/3})$ multiple of each other, 
it means that all of the $r(x)$, for $x\in C_j$, 
are within a $(1+n^{-1/3})^{O(k)}$ multiple of each other. 
Therefore, for each $j$, all of the $x\in C_j$ have $r(x)$ 
within a constant multiple of $s_j$. Thus, all of these $x$'s 
come from Gaussians with $\sigma_i$ 
within a constant multiple of $s_j$. 
Let $S_j$ be the set of $i$ such that all samples from $N(\mu_i,\sigma_i)$ are in $C_j$. 
By Lemma \ref{lem:dim-reduce}, the orthogonal projection $\mu'_i$ 
of $\mu_i$ for $i \in S_j$ onto $W_j$ satisfies 
$\|\mu'_i - \mu_i\|_2 \leq \eps \sigma$, 
where $\sigma^2= \left( \sum_{i \in S_j} w_i \sigma_i^2 \right)/\sum_{i \in S_j} w_i$. 
Since $\sigma=\Theta(s_j)=\Theta(\sigma_i)$ for each $i \in S_j$, 
we have that $\|\mu'_i - \mu_i\|_2 \leq O(\eps \sigma_i)$.
Therefore, since $W$ contains $W_j$, 
$\mu_i$ is within $O(\eps\sigma_i)$ 
of its projection onto $W$ for all $i$.

Finally, since each $W_j$ has dimension at most $k-1$ 
and since $W$ is the sum of at most $k$ of them, 
we have that $\dim(W)\leq k^2$.
This completes the proof.
\end{proof}

\subsection{Putting Everything Together} \label{ssec:gmm-final}

By combining Proposition~\ref{prop:gmm-identity} and Lemma~\ref{lem:dim-reduce}, 
we immediately obtain the following corollary:

\begin{corollary} \label{cor:gmm-identity-dim-reduction}
There is an algorithm that given a positive integer $d$, constants $1/2>\alpha > 4\eps \geq 0$, $0 < \tau < 1$,
and sample access to a probability distribution $M=\sum w_i N(\mu_i,I)$ 
with $w_i\geq \alpha$ for all $i$, and so that $\|\mu_i-\mu_j\|_2$ is at least
$$
S \eqdef C (\alpha^{-1/(2d)} \sqrt{d} (d+\log(1/\alpha))(\log(2+\log(1/\alpha)))^2+ \sqrt{\log(1/\eps)})
$$
for all $i\neq j$, takes \new{$\poly\left(n(kd)^d \log(1/\tau)/(\eps \alpha)\right)$} samples from $X$,
runs in time  \new{$\poly\left(n (kd\log(1/\tau)/\alpha)^{O(d)}/\eps\right)$}, 
and, with probability at least $1-\tau$, returns a list of pairs $(u_i,\nu_i)$,
so that up to some permutation $|u_i-w_i|=O(\eps)$ and $\|\mu_i-\nu_i\|_2 = \tilde O(\eps/w_i)$.
\end{corollary}

\paragraph{Proof of Theorem \ref{thm:gmm-inf}.}
To prove this theorem, we will combine Proposition \ref{dimReduxProp} 
with Proposition \ref{prop:gmm-general} and a few additional ingredients. 
In particular, running Proposition \ref{dimReduxProp}, we find a subspace $W$, as required. 
We note that the projection of $X$ onto $W$ is still a mixture of Gaussians 
with appropriate separations between the means to run Proposition \ref{prop:gmm-general}. 
We note that because the dimension is now only $\poly(k/\delta)$, 
the $\log(n)$ term in $S$ becomes a $\log(k/\delta)$, and the dependence on $n$ 
in the sample complexity disappears. We can then learn $w_i$ to error $O(\delta)$ 
and the projection of $\mu_i$ to $W$ to error $\sigma_i O(\delta/w_i)$, 
which is within $\sigma_i O(\delta/w_i)$ of the true value of $\mu_i$.

Learning approximations to the $\sigma_i$ is slightly more difficult. 
Naively, we should only be able to learn it to error 
$\sigma_i O(\delta/w_i)/\sqrt{\dim(W)}$, which is not good enough. 
However, we note that samples from $X$ can be reliably sorted (with $\delta$ probability of error) 
by which Gaussian they came from by considering which equivalence class $C$ 
from Proposition \ref{prop:gmm-general} the sample came from. 
Looking at the distances between pairs of the original samples in $\R^n$ 
whose projections end up in the same class, and taking the median, 
we can approximate $\sigma_i$ to error $\sigma_i O(\delta/w_i)/\sqrt{n}$.
This completes the proof. \qed

\section{Minimax Error Bounds and SQ Lower Bounds} \label{sec:minimax-sq}

\subsection{Minimax Error Bounds for List-Decodable Mean Estimation} \label{ssec:minimax}

We consider the following question: How small a distance from the true mean
can be achieved with a number of hypotheses at most $\poly(1/\alpha)$ 
or only as a function of $\alpha$ ---
independently of the runtime or sample complexity of the algorithm used?
Our algorithm for Gaussians shows that we can obtain an error of $\polylog(1/\alpha)$,
but what is the information-theoretically optimal error? How about for sub-Gaussian
distributions or distributions whose first few moments are bounded?
We provide tight minimax error bounds for such distribution families.

Specifically, we show that the information-theoretically 
optimal error is $\Theta(\sqrt{\log(1/\alpha)})$ for the case of 
$N(\mu, I)$ (upper and lower bound).
More generally, we give a tight error upper bound of $O(\sqrt{\log(1/\alpha)})$ 
for sub-gaussian distributions with bounded variance
in each direction. Previously, no upper bound better than 
$\tilde{O}(1/\sqrt{\alpha})$ was known for these families.
Regarding lower bounds,~\cite{CSV17} showed an 
$\Omega(\sqrt{\log(1/\alpha)})$ error lower bound 
for $N(\mu, \Sigma)$, where $\Sigma$ is {\em unknown} and 
$\Sigma \preceq I$. We strengthen this result 
by showing that the $\Omega(\sqrt{\log(1/\alpha)})$ 
lower bound holds even for $N(\mu, I)$.

We now summarize our contributions for distributions
with bounded moments.
Recall that for distributions with bounded covariance,
\cite{CSV17} gave an efficient algorithm with error $\tilde{O}(1/\sqrt{\alpha})$, 
and it was open whether this error bound can be improved.
We establish that the optimal error achievable in this regime is 
$\Theta(1/\sqrt{\alpha})$ (upper and lower bound). 
More generally, we obtain tight information-theoretic upper and lower bounds on the error, 
assuming that the first $k$ central moments of the distribution 
are bounded from above (for even $k$). Roughly speaking, we show 
that the optimal error in this regime is $\Theta_k (\alpha^{-1/k})$.

\smallskip
The structure of this section is as follows:
First, we describe a generic (inefficient) algorithm that only requires concentration
bounds and applies to all aforementioned families. As a corollary, we obtain our 
information-theoretic upper bounds.
We then give matching information-theoretic lower bounds for the three aforementioned 
distribution families.

\paragraph{Minimax Error Upper Bounds.}
We show the following:

\begin{proposition}[Generic Error Upper Bound] \label{thm:generic-inefficient}
Let $\mathcal{D}$ be a family of distributions on $\R^n$
with the following concentration property: There exists $u \in \R_+$ such that
for any $D \in \mathcal{D}$ with mean $\mu$ and any direction $v$ with $\|v\|_2=1$ we have that
$\Pr_{X \sim D}[|v \cdot (X - \mu)| \geq u ] \leq \alpha/20$.
Then there exists an (exponential time) algorithm for list-decodable mean estimation 
of $\mathcal{D}$ with $O(1/\alpha)$ candidate means and error at most $2u$.
\end{proposition}
\begin{proof}
Let $T$ be an $\alpha$-corrupted set of $\Omega(n/\alpha^3)$ samples from some $D \in \mathcal{D}$
with unknown mean $\mu$.
By definition, there exists a set $S \subseteq T$ of independent samples
from $D$ with $|S|  = \alpha \cdot |T|$. Furthermore, these samples
should be representative of $D$ in the sense that for any unit vector $v$:
$$\Pr_{X \in_u S}[|v \cdot (X - \mu)| \geq u ] \leq \alpha/10 \;.$$
The above statement holds with high probability by the VC-inequality and the assumption
in the statement of the proposition. We henceforth condition on this event.

Let $H$ be the set of all points in $\R^n$ with the following property:
A point $x_i \in \R^n$ is in $H$ if
there exists a subset $S_i \subseteq T$ of cardinality $|S_i|\geq \alpha \cdot |T|$
so that, in any direction, all but an $\alpha/10$-fraction of the points in $S_i$
are within distance $u$ of $x_i$. That is, for any unit vector $v$, we have that:
\begin{equation} \label{H-def}
\Pr_{X \in_u S_i}[|v \cdot (X - x_i)| \geq u ] \leq \alpha/10 \;.
\end{equation}
By the above conditioning, we have that $\mu\in H$.
The following important claim motivates the algorithm:
\begin{claim} \label{claim:h-cover}
There exists a covering of $H$ by $O(1/\alpha)$ balls of radius $2u$.
\end{claim}

Given the claim, the algorithm is very simple:
Our algorithm will return as its list the set of centers of the balls in such a covering.
We will therefore have that $\mu \in H$ is within distance $2u$
of at least one of our candidate hypotheses.
The proof of the claim below completes the proof of Proposition~\ref{thm:generic-inefficient}.

\begin{proof}[Proof of Claim~\ref{claim:h-cover}]
The proof is by contradiction.
Suppose that there exist points 
$x_1, \ldots, x_m\in H$ with $m = 5/\alpha$ 
so that $\|x_i-x_j\|_2 \geq 2u$, for all $i \neq j$.
By the definition of the set $H$, to each $x_i$ we have an associated set $S_i$ satisfying \eqref{H-def} 
and containing at least an $\alpha$-fraction of the points in $T$. 
In addition, we claim that for each $i \neq j$ the size 
of the intersection $S_i \cap S_j$ is at most $(\alpha/10)(|S_i|+|S_j|)$.
We show this as follows: Since $\|x_i-x_j\|_2 \geq 2u$, 
if $v_{ij}$ is the unit vector in the direction of $x_i-x_j$, 
every point \new{$y \in \R^n$} must satisfy 
(1) $|v_{ij} \cdot (y-x_i)|\geq u$ or (2) $|v_{ij} \cdot (y-x_j)|\geq u$.
By property \eqref{H-def}, at most $\alpha/10$-fraction of 
points $y \in S_i$ satisfy (1) and 
at most $\alpha/10$-fraction of points $y \in S_j$ satisfy (2). 
If $L$ is the subset of points $y \in S_i\cap S_j$ satisfying (1), 
then $|L| \leq (\alpha/10) |S_i|$. Similarly, if $R$ is the subset of points 
$y \in S_i\cap S_j$ satisfying (2), then $|R| \leq (\alpha/10) |S_j|$.
Therefore, we have that $|S_i\cap S_j|\leq (\alpha/10)(|S_i|+|S_j|)$.
By the inclusion-exclusion formula, we have that
\begin{align*}
|T| & \geq |\cup_{i=1}^m S_i| \geq \littlesum_{i=1}^m |S_i| - \littlesum_{1\leq i<j\leq m}|S_i\cap S_j| \\
& \geq \littlesum_{i=1}^m |S_i| - \littlesum_{i=1}^m \littlesum_{j\neq i} |S_i|\alpha/10 
\geq \littlesum_{i=1}^m |S_i|(1-m\alpha/10)\\
& \geq \littlesum_{i=1}^m \alpha|T|/2  > |T| \;.
\end{align*}
This yields the desired contradiction.
\end{proof}
\end{proof}

Proposition~\ref{thm:generic-inefficient} yields tight error upper bounds for a number
of distribution families. We explicitly state its implications for the following 
families: sub-gaussian distributions, distributions with bounded covariance, and distributions
whose first $k$ moments are appropriately bounded.

\begin{corollary} \label{cor:ub-error}
Let $0< \alpha <1/2$. Given as input a sufficiently large 
$\alpha$-corrupted set of samples from a distribution $D \in \mathcal{D}$ on $\R^n$, 
there exists a list-decodable mean estimation algorithm for $D$
that outputs a list of size $O(1/\alpha)$ and whose error guarantee is at most:
\begin{itemize}
\item $O(\sqrt{\log(1/\alpha)}) \cdot \sigma$, if $\mathcal{D}$ 
is the family of sub-gaussian distributions with parameter $\sigma$.
\item $O(1/\sqrt{\alpha}) \cdot \sigma$, if $\mathcal{D}$ 
is the family of distributions with covariance matrix 
$\Sigma \preceq \sigma^2 \cdot I$.

\item $O((C/\alpha)^{1/k})$, if $\mathcal{D}$ is the family 
of distributions whose $k^{th}$ central moments in any direction, 
for some even $k>0$, are at most $C$.
\end{itemize}
\end{corollary}
\begin{proof}
All three statements follow from Proposition~\ref{thm:generic-inefficient} 
by applying the appropriate concentration inequality.

Recall that if $D$ is sub-gaussian on $\R^n$ with mean vector $\mu$ and parameter $\nu>0$, 
then for any unit vector $v \in \R^n$ we have that 
$\Pr_{X \sim D}\left[|v \cdot (X-\mu)| \geq u \right] \leq \exp(-u^2/2\nu^2)$.
By taking $u = \Theta(\log^{1/2}(1/\alpha)) \cdot \nu$, the assumption in the statement of 
Proposition~\ref{thm:generic-inefficient} is satisfied.

\new{
If $D$ has covariance matrix $\Sigma \preceq  \sigma^2 I$, we can take 
$u = \Theta(1/\alpha^{1/2}) \cdot \sigma$ by Chebyshev's inequality. 
If $\E[(v \cdot X)^k] \leq C$, then $\Pr[|v \cdot (X - \mu)| \geq u ] =  \Pr[(v \cdot (X - \mu))^k \geq u^k ]\leq C/u^k$,
and so with bounds on the $k^{th}$ moments, we can take $u=(20C/\alpha)^{1/k})$.
This completes the proof.}
\end{proof}

\paragraph{Minimax Error Lower Bounds.} We prove the following:

\begin{proposition}[Error Lower Bound] \label{thm:minimax-lb}
Let $0< \alpha <1/2$. Any list-decodable mean estimation algorithm for the family $\mathcal{D}$ 
must have error guarantee at least:
\begin{itemize}
\item[(i)] $\Omega(\sqrt{\log(1/\alpha)})$, if $\mathcal{D}$ is the family $N(\mu,I)$, $\mu \in \R^n$,
and the size of the list is at most $\poly(1/\alpha)$.
\item[(ii)] $\Omega(1/\sqrt{\alpha})$, if $\mathcal{D}$ is the family of distributions 
with bounded covariance $\Sigma \preceq I$, and the list size depends only on $\alpha$. 
\item[(iii)] $\Omega_k(\alpha^{-1/k})$, if $\mathcal{D}$ is a family of distributions having 
its first $k$ mixed central moments agree with the corresponding moments of the standard Gaussian, 
and the list size depends only on $\alpha$. 
\end{itemize}
\end{proposition}

\begin{proof}
We prove each part separately.

\medskip

\noindent {\bf \em Proof of (i).}
We start with our lower bound for $N(\mu, I)$. In more detail, we will show that
for any $d > 0$, any list-decoding algorithm for $N(\mu, I)$, that achieves 
error $o(\sqrt{\log(1/\alpha)/d})$, must return more than $(1/\alpha)^d$ hypotheses, 
as long as the dimension $n \geq \ln(1/\alpha)^{Cd}$, for some universal constant $C > 0$.

Let $m = \lceil (1/\alpha)^d \rceil$. 
We will construct a distribution $X$ that can be considered as a mixture 
$$\alpha N(\mu_i,I) +(1-\alpha)E_i \;,$$ for $i \in [m]$, 
such that for any $1 \leq i \neq j \leq m$, we have 
$\| \mu_i -\mu_j \|_2 \geq c \cdot \sqrt{\ln(1/\alpha)}$, 
for some suitable constant $1/4 > c>0$ depending on $d$.
Specifically, we can take $c =1/(8\sqrt{d+3})$.

For $1 \leq i \leq m$, we let $\mu_i = 2 c \sqrt{\ln(1/\alpha)} \cdot v_i$, 
where $v_i$ is chosen uniformly at random from the unit sphere in $\R^n$. 
For $n$ large enough, the following claim holds: with high probability
no pair $(\mu_i, \mu_j)$, $i \neq j$, is closer than 
$c \cdot \sqrt{\ln(1/\alpha)}$. Indeed, by Lemma 3.7 of \cite{DiakonikolasKS16c}, 
a set of $2^{\Omega(n^{1/3})}$ 
unit vectors chosen uniformly at random from the unit sphere 
with high probability satisfy that any two of them $v, v'$ 
have $v \cdot v' \leq O(n^{-1/6})$. 
Thus, if $n \geq \ln(1/\alpha)^{Cd}$, it follows that 
$$2^{\Omega(n^{1/3})}=2^{\Omega (\log(1/\alpha))^{Cd/3}} \geq m \;,$$ 
for a sufficiently large universal constant $C$. 
We therefore get that
$$\|\mu_i-\mu_j\|_2= 2c\sqrt{\ln(1/\alpha)} \sqrt{1-v_i \cdot v_j} \geq c \cdot \sqrt{\ln(1/\alpha)} \;,$$
as desired.

\new{Let $dN(\mu,I)$ denote the pdf of $N(\mu,I)$.}
We will pick our distribution $X$ to be $Y/\|Y\|_1$, where $Y$ is
the pseudo-distribution with density $dY = \max_{i \in [m]} dN(\mu_i,I)$. 
It suffices to show that $\|Y\|_1 \leq 1/\alpha$, since then 
the pdf of $X$ satisfies $dX \geq \alpha dN(\mu_i, I)$, for each $i$, 
and so we can write $X=\alpha N(\mu_i,I) +(1-\alpha)E_i$, 
for a suitable distribution $E_i$.

We note that the pdf $dN(\mu_i, I)$ of $N(\mu_i, I)$ is 
$$(1/\sqrt{2\pi}) \exp\left(-\|x\|^2/2 +  2c\sqrt{\ln(1/\alpha)}(v_i \cdot x) - 2\sqrt{c} \ln(1/\alpha)\right)\;.$$ 
Comparing this to the pdf of $N(0,I)$, $(1/\sqrt{2\pi})\exp(-\|x\|^2/2)$, we see that their ratio 
$$\exp\left(2c\sqrt{\ln(1/\alpha)}(v_i \cdot x) - 2\sqrt{c} \ln(1/\alpha)\right)$$ 
is less than $1/(2 \alpha)$, so long as 
$$2c \sqrt{\ln(1/\alpha)}(v_i \cdot x) \geq 2\sqrt{c} \ln(1/\alpha) + \ln(1/(2 \alpha)) \;.$$
Using our assumptions that $c \leq 1/4$ and $\alpha \leq 1/2$, 
this  holds unless $(v_i \cdot x) \geq \sqrt{\ln(1/\alpha)}/(2c)$. 

Therefore, $dN(\mu_i,I)$ is less than $1/(2\alpha) dN(0,I)$, 
unless $(v_i \cdot x) \geq \sqrt{\ln(1/\alpha)}/(2c)$. 
Let $T_i$ be the pseudo-distribution obtained by restricting 
$N(\mu_i,I)$ to the set $(v_i \cdot x) \geq \sqrt{\ln(1/\alpha)}/(2c)$. 
Since $c \leq 1/4$,  note that 
$$\sqrt{\ln(1/\alpha)}/(2c) - 2c \sqrt{\ln(1/\alpha)} \geq \sqrt{\ln(1/\alpha)}/(4c) \;,$$ 
and therefore, by standard tail bounds, the total mass of $T_i$ is
$$\Pr_{Z \sim N(0,1)}[Z \geq \sqrt{\ln(1/\alpha)}/(4c)] \leq 
2 \exp\left(-\ln(1/\alpha)/(32c^2)\right) = 2\alpha^{1/(32c^2)} \;.$$
By definition, we have that $dY \leq (2/\alpha) dN(0,I) + \sum_i dT_i$. Therefore, 
$\|Y\|_1 \leq 1/(2\alpha) + 2m\alpha^{1/(32c^2)}$. 
Since $m \leq 2(1/\alpha)^d$, when $1/(32c^2) \geq d+3$ 
we have $\|Y\|_1 \leq 1/\alpha$, which completes the proof. 
Note that we can take $c=1/(8\sqrt{d+3})$.

\medskip

\noindent {\bf \em Proof of (ii).}
Let $p$ be a discrete $1$-dimensional random variable supported on the points
$\{0, \pm (2\alpha)^{-1/2} \}$ such that its expectation is $0$ and its variance is $1$.
In particular $p$ is plus or minus $(2\alpha)^{-1/2}$ each with probability $\alpha$, and otherwise $0$.
Let $X$ be an $n$-dimensional distribution whose coordinates are independent copies of $p$.
Let $D_i$ be the distribution $X$ conditioned on the $i^{th}$ coordinate being positive.
Notice that $X=\alpha D_i +(1-\alpha)Y_i$ for some distribution $Y_i$.
Note also that the mean of $D_i$ is $(2\alpha)^{-1/2}e_i$, and its covariance is at most the identity.
Suppose that $D=D_i$ for some randomly chosen value of $i$, and the algorithm is given sample access to $X$.
If the algorithm returns a list of fewer than $n/2$ hypotheses, then for at least half of possible $i$,
the mean of $D_i$ will be at least $\alpha^{-1/2}/2$-far from the closest hypothesis.
Since the list of returned hypotheses was assumed to have size independent of $i$,
this means with probability at least $1/2$ the algorithm fails to produce
a hypothesis closer than $\alpha^{-1/2}/2$ to the mean of $D$.

\medskip

\noindent {\bf \em Proof of (iii).}
To generalize the proof of (ii) 
to the case when $D$ matches the first $k$ moments with a Gaussian, 
we will need the following technical lemma:
\begin{lemma}\label{momentMatchLem}
For $k>0$, there exists a one-dimensional distribution 
$A=\alpha N(\mu,1)+(1-\alpha)E$, for some distribution 
$E$ and $\mu=\Omega((k\alpha)^{-1/k})$, 
so that $A$'s first $k$ moments agree with those of $N(0,1)$. 
Furthermore, the pdf of $E$ can be taken to be point-wise at most twice the pdf of a standard Gaussian.
\end{lemma}
\begin{proof}
We take $A$ to have pdf:
$$A(x) = \alpha \cdot G'(x) + (1-\alpha)(G(x)+p(x) \cdot \mathbf{1}_{[-1, 1]}) \;,$$
where $G'(x)$ is the pdf of a Gaussian $N(\mu, 1)$,
$G(x)$ is the pdf of $N(0, 1)$,
$\mu=O(k\alpha)^{-1/k}$,
and $p$ is an appropriately selected degree-$k$ polynomial.
Specifically, we select $p$ to be the unique degree-$k$ polynomial
whose first $k$ moments on $[-1, 1]$ match the first $k$ moments of
$\alpha (G-G')$, where $G$ is $N(0, 1)$. 
We will need the following technical claim:
\begin{claim} \label{clm:linf} 
Let $p(x)$ be the unique polynomial of degree $k$ such that 
$$\int_{-1}^{1} p(x) x^i dx = \alpha \int_{-\infty}^{\infty} (G(x)-G(x-\mu)) x^i dx \;,$$ 
where $|\mu| \geq 1$ and $G(x)$ is the pdf of $N(0,1)$. 
Then, $\max_{x \in [-1,1]} |p(x)| \leq O(\alpha  (2k\mu)^{k})$. 
\end{claim}
\begin{proof}
We can write $p(x)$ as a linear combination of Legendre polynomials $P_k(x)$ 
as $p(x)=\sum_{i=1}^k a_iP_i(x)$, 
where $\int_{-1}^1 P_i(x) P_j(x) dx= (2/(2i+1)) \delta_{ij}$. 
Now we have for $1 \leq i \leq k$ that
$$a_i =(2i+1)/2  \int_{-1}^1 p(x) P_i(x) dx = \alpha (2i+1)/2 \int_{-\infty}^{\infty} (G(x)-G(x-\mu)) P_i(x) dx \;.$$
By Taylor's theorem, we have
$$G(x)-G(x-\mu) = - \sum_{j=1}^\infty (-\mu)^j \He_j(x) G(x)/j! \;.$$
Now we can write one basis of orthogonal polynomials 
in terms of another as 
$P_i(x)=\sum_{j=0}^i b_{ij} \He_j(x)/\sqrt{j!}$ for some $b_{ij}$. 
We thus have
\begin{align*}
a_i &= \alpha (2i+1)/2 \int_{-\infty}^{\infty} (G(x)-G(x-\mu)) P_i(x) dx \\
&= \alpha (2i+1)/2 \sum_{j=0}^i b_{ij}/\sqrt{j!} \int_{-\infty}^{\infty} \He_j(x) \sum_{\ell=1}^\infty (-\mu)^{\ell} \He_{\ell}(x) G(x)/\ell! dx \tag*{$(\ast)$}\\
&= \alpha (2i+1)/2  \sum_{j=0}^i b_{ij}/\sqrt{j!} \sum_{\ell=1}^\infty (-\mu)^{\ell} \int_{-\infty}^{\infty} \He_j(x) \He_{\ell}(x) G(x)/\ell! dx \\
&= \alpha (2i+1)/2  \sum_{j=1}^i b_{ij} (-\mu)^j / \sqrt{j!} \;,
\end{align*}
where switching the order of summation and integration at $(\ast)$ 
is justified by Fubini-Tonelli, since
\begin{align*}
& \sum_{\ell=1}^\infty |\mu|^{\ell}/(\ell! \sqrt{j!})  \int_{-\infty}^{\infty} |\He_j(x)| |\He_{\ell}(x)| G(x) dx \\
& \leq \sum_{\ell=1}^\infty |\mu|^{\ell}/(\ell!\sqrt{j!}) \sqrt{ \int_{-\infty}^{\infty} \He_j(x)^2 G(x) dx \int_{-\infty}^{\infty} \He_{\ell}(x)^2 G(x) dx} \\
& =  \sum_{\ell=1}^\infty |\mu|^{\ell}/\sqrt{\ell!} < \infty \;.
\end{align*}
Noting, as in Corollary 5.4 of \cite{DiakonikolasKS16c}, that 
$|P_j(x)| \leq (4|x|)^j$ for $|x| \geq 1$, we have that
\begin{align*}
\sum_{j=0}^i b_{ij}^2 & = \int_{-\infty}^\infty P_j(x)^2 G(x) dx \\
& \leq \int_{-1}^1 P_j(x)^2 + \int_{-\infty}^\infty (4|x|)^{2j} G(x) dx \\
& \leq O(\sqrt{2j!}) \;.
\end{align*}
Thus, we obtain that
\begin{align*}
|a_i| &= \alpha (2i+1)/2  \sum_{j=0}^i |b_{ij}| |\mu|^j / \sqrt{j!}  \\
&\leq \alpha (2i+1)/2  \sqrt{\left(\sum_{j=1}^i b_{ij}^2 \right) \left(\sum_{j=1}^i \mu^{2j}/j! \right)} \\
&\leq \alpha (2i+1)/2  \cdot  O((2i!)^{1/4}) \cdot \mu^i \sqrt{e} \\
&\leq O(\alpha \cdot (2i)^{i/4+1} \cdot \mu^i) \;.
\end{align*}
Finally, since $|P_k(x)| \leq 1$ for all $x \in [-1,1]$, 
for any $x \in [-1,1]$ we have
\begin{align*}
|p(x)| & \leq \sum_{i=0}^k |a_i| |P_k(x)| \leq \sum_i |a_i| \\
& = \sum_i O(\alpha \cdot (2i)^{i/4+1} \cdot \mu^i) \\
& \leq O(\alpha  (2k)^{k/4+2} \mu^k) \\
& \leq  O(\alpha  (2k\mu)^{k}) \;.
\end{align*}
This gives Claim~\ref{clm:linf}.
\end{proof}
By Claim \ref{clm:linf}, we have that 
$\|p\|_{\inf} \leq 1/\sqrt{2\pi e}$ for $\mu=(Ck\alpha)^{-1/k}$, 
when $C$ is sufficiently large, and therefore 
$A$ is non-negative and $E$ satisfies the claimed bound. 
We note that $p$ is normalized since its $0^{th}$ moment is correct. 
Therefore, we have constructed an appropriate probability distribution,
which completes the proof of Lemma~\ref{momentMatchLem}.
\end{proof}
To prove the lower bound for matching $k$ moments, 
we let $A$ be the distribution constructed in Lemma \ref{momentMatchLem}, 
and let $X$ be a product of copies of $A$. We note that $D$ could be a 
product of copies of $A$ in all but one direction and a copy of $N(\mu,1)$ 
in the remaining direction. Thus, $D$ could have mean $\mu e_i$ for any $i$. 
Once again, any algorithm that reliably produce some element of its list within $\mu/2$ 
of the mean of $D$, must return a list of length at least $\Omega(n)$. 
This gives (iii).
\end{proof}

\subsection{SQ Lower Bounds for List-Decodable Mean Estimation} \label{ssec:sq}

Recall that our list-decoding algorithm from Theorem~\ref{thm:list-decoding-inf} achieves
error $O(\alpha^{-1/d})$ in time $(n/\alpha)^{O(d)}$. Can this runtime be improved?
In particular, is there an algorithm that achieves the optimal error of $\Theta(\sqrt{\log(1/\alpha)})$
and runs in time $\poly(n/\alpha)$? We show that this is not the case, if we restrict
ourselves to Statistical Query (SQ) algorithms. Specifically, we prove the following theorem:

\begin{theorem} \label{thm:sq}
For $1/2 > c > 0$, any SQ list-decoding algorithm that returns 
a hypothesis within $c_k \alpha^{-1/k}$ (for some constant $c_k>0$) of the true mean, 
does one of the following:
\begin{itemize}
\item Uses queries with error tolerance at most $\exp(O(\alpha^{-2/k}))O(n)^{k(1/4-c/2)}$.
\item Uses a number of queries at least $\exp(\Omega(n^{c/2}))$.
\item Returns a list of more than $\exp(\Omega(n)^{c})$ hypotheses.
\end{itemize}
\end{theorem}
\begin{proof}
We will use the terminology and techniques of our recent work on this topic~\cite{DiakonikolasKS16c}.
We provide a sketch of the proof here.
Using the generic construction of~\cite{DiakonikolasKS16c},
for the distribution $A$ given in Lemma \ref{momentMatchLem}, 
we have that:
\begin{enumerate}
\item[(i)] $A$ matches the first $k$ moments with $N(0, 1)$, and
\item[(ii)] $\chi^2(A,N(0,1))=\exp(O(\mu^2))$.
\end{enumerate}
Note that (i) follows immediately from Lemma~\ref{momentMatchLem} 
and (ii) follows from noting that
\begin{align*}
\chi^2(A,N(0,1)) & \ll \chi^2(N(0,\mu),N(0,1))+\chi^2(N(0,1),N(0,1))\\
& \leq \exp(\mu^2) +O(1) \;.
\end{align*}

Let $\p_v$ be the $n$-dimensional distribution that is distributed as $A$ in the $v$-direction
and an independent standard Gaussian in orthogonal directions. Then $p_v$ will be an $\alpha$-mixture
of a Gaussian $N(v|\mu|,I)$ with some other distribution. 

The proof of Proposition 3.3 from \cite{DiakonikolasKS16c} goes through, 
except for the detail that more than one distribution $\p_v$ in the finite set  $\mathcal{D}_D$ 
of pairwise correlations may have the returned list as a valid output. 
Formally, we need to define list-decoding as a search problem in the notation of that paper. 
The search problem $\mathcal{Z}$ is to find a list of length at most $\ell$ such that 
it contains a hypothesis within $|\mu|/10 =c_k \alpha^{-1/k}$ of $v|\mu|$. 
The space of solutions is lists of length $\ell$ of vectors in $\R^n$.
For any such list $L$, the set of distributions for which it is a valid solution, $\mathcal{Z}^{-1}(L)$, 
is the set of $\p_v$ with $\|v |\mu| - w\|_2 \leq |\mu|/10$ for some $w$ in $L$. 
For $\mathcal{D}_D$ as defined in the proof of Proposition 3.3 from \cite{DiakonikolasKS16c},  
$\p_v, \p_{v'} \in \mathcal{D}_D$ only when $v \cdot v' \leq 1/2$, 
and so $\|v |\mu| - v' |\mu|\|_2 >  |\mu|/5$. 
Thus, any element $w$ of $L$ can only be close to one distribution in $\mathcal{D}_D$, 
and so $|\mathcal{D}_D \setminus \mathcal{Z}^{-1}(L)| \geq |\mathcal{D}_D|-\ell$. 
If $\ell = \exp(o(n)^c)$, then $\ell \leq |\mathcal{D}_D|/2$. 
In this case, the Statistical Query dimension for this search problem 
is similar to that in that proof and it shows that no SQ algorithm 
taking fewer than $\exp(\Omega(n^{c/2}))$ queries with larger 
than $O(n)^{k(1/4-c/2)}/\chi^2(A,N(0,1)) = \exp(O(\alpha^{-2/k}))O(n)^{k(1/4-c/2)}$ 
accuracy can solve the search problem. 

This completes the proof.
\end{proof} 

\bibliographystyle{alpha}
\bibliography{allrefs}

\appendix

\section*{APPENDIX}

\section{Proof of Claim~\ref{claim:gamma}} \label{sec:gamma}

We have the following sequence of (in-)equalities:
\begin{align*}
& \exp(x) \Gamma(s,x)/(s+x)^{s-1} =\int_x^\infty \exp(x-t) (t/(s+x))^{s-1} dt \\
& = \int_0^\infty \exp(-u) ((x+u)/(s+x))^{s-1} du \tag*{(where $u=t-x$)} \\
& = \int_0^\infty \exp(-u)(1+(u-s)/(s+x))^{s-1} du \\
& \leq \int_0^\infty \exp(-u +(s-1)(u-s)/(s+x)) du \tag*{(since $1+y \leq e^y$ for all $y \in \R$)} \\
& = \exp(-s(s-1)/(s+x)) \int_0^\infty \exp(-u(x+1)/(s+x)) du\\
& = \exp(-s(s-1)/(s+x)) (s+x)/(x+1)  \\
&= \exp(-s(s-1)/(s+x))(1+(s-1)/(x+1))\\
& \leq \exp(-s(s-1)/(s+x) + (s-1)/(x+1)) \\
& =\exp (-(s-1)^2 x /(s+x)(x+1)) \leq 1 \;.
\end{align*}
This completes the proof.

\section{Reducing the List Size to $O(1/\alpha)$ }
In this section we show:

\begin{proposition} \label{prop:list-reduce}
Fix $\alpha, \delta, \beta, t > 0$, and let $\mu^* \in \R^d$
Let $T \subset \R^d$ be finite, and let $S \subseteq T$ be so 
that (i) $|S| / |T| \geq \alpha$, and 
(ii) for all unit vectors $v \in \R^d$, we have $\Pr_{X \in_u S} [v^T (X - \mu^*) > t] < \delta$.
Then, given $M = \{\mu_1, \ldots, \mu_N\} \subset \R^d$ 
so that $\delta N = o(1)$ and there is some $i$ so that $\| \mu - \mu^* \|_2 \leq \beta$ for some $\mu \in M$, 
there is a polynomial time algorithm which outputs $M' \subseteq M$
so that $|M'| \leq \frac{1}{\alpha} (1 + O(\delta N))$ 
and $\| \mu' - \mu^* \|_2 \leq 3 (\beta + t)$ for some $\mu' \in M'$.
\end{proposition}

To apply this to a list of $O(1/\alpha^3)$ points including 
an approximation to $\mu^*$ when $T$ contains a set $S$ 
representative for $N(\mu^*,I)$, we may take $\delta = 1/(C \log(1/\alpha)$ 
and $t=\sqrt{\log(C \log(1/\alpha))}$ for a sufficiently large $C$. 
This $t$ will be smaller than $\beta$ in our applications.

\begin{proof}
The algorithm proceeds as follows:

\medskip

\fbox{\parbox{6.1in}{
{\bf Algorithm} {\tt ListReduction}\\
Input: a list $M = \{\mu_1, \ldots, \mu_N\}$, a set of samples $T$, $\alpha, \beta, \delta,t > 0$. \\
\vspace{-0.5cm}
\begin{enumerate}
\item For all $i, j \in [n]$ so that $i \neq j$, let $v_{ij}$ denote the unit vector in the $\mu_i - \mu_j$ direction.
\item Let 
$
T_i = \bigcap_{j \neq i} \{X \in T: |v_{ij}^T (X - \mu_i) | < \beta + t  \}$.
\vspace{-0.2cm}
\item Let $M'$ be the empty list
\item For each $i$, if  $|T_i| \geq \alpha (1 - \delta N) |T|$, and no $\mu_j \in M'$ has  $\| \mu_i - \mu_j \|_2 < 2 (\beta + t)$, then add $\mu'$ to $M'$
\item Return $M'$. 
\end{enumerate}
\vspace{-0.2cm}
}}

\bigskip

It is easy to see these operations can be done in polynomial time.

We claim the output of this algorithm satisfies the desired guarantees.
We first show that $M'$ contains a $\mu'$ with $\| \mu' - \mu^* \|_2 \leq 3 (\beta + t)$.
By assumption, if $\| \mu_i - \mu^*  \|_2 \leq \beta$, then $\mu_i$ satisfies this property.
By the triangle inequality, we have that for every $j \neq i$, 
\[
\Pr_{X \sim S} [v_{ij}^T (X - \mu_i) > \beta + t] < \delta \; ,
\]
and hence \new{by a union bound $|T_i| \geq \alpha (1 - \delta N) |T|$}.
Since such a $\mu_i$ is contained in $M$ by assumption, this implies that $M'$ is non-empty and either $\mu_i \in M'$, 
when  or there is some other $\mu' \in M'$ so that $\| \mu' - \mu_i \|_2 \leq 2(\beta + t)$, 
and therefore $\| \mu' - \mu^*\|_2 \leq 3 (\beta + t)$ by the triangle inequality.

It remains to bound the size of $M'$.
Observe that if $\| \mu_i - \mu_j \|_2 \geq 2 (\beta + t)$, then $T_i \cap T_j = \emptyset$.
Therefore we have
\[
\sum_{\mu_i \in M'} |T_i| \leq |T| \;,
\]
but since $|T_i| \geq \alpha (1 - \delta N) |T|$ by assumption, this implies that 
$ |M'| \alpha (1 - \delta N) |T| \leq |T| \; ,$ or $|M'| \leq \frac{1}{\alpha} (1 + O(\delta N))$, as desired.
\end{proof}

\end{document}